\documentclass[11pt]{article}
\usepackage{setspace}
\usepackage{comment}

\usepackage[margin=1in]{geometry}
\usepackage{amsmath,amssymb,amsfonts,graphicx,amsthm,nicefrac,mathtools,bm}
\usepackage{hyperref}
\usepackage{natbib}
\usepackage[english]{babel}
\usepackage{times}
\usepackage[T1]{fontenc}
\usepackage{bbm}
\usepackage{color}
\usepackage[usenames,dvipsnames,svgnames,table]{xcolor}
\usepackage{xspace}
\usepackage{rotating,multirow}
\usepackage[final]{showlabels}
\usepackage{setspace}

\usepackage{xargs}    
\usepackage[colorinlistoftodos,prependcaption,textsize=tiny]{todonotes}
\newcommandx{\unsure}[2][1=]{\todo[linecolor=red,backgroundcolor=red!25,bordercolor=red,#1]{#2}}
\newcommandx{\change}[2][1=]{\todo[linecolor=blue,backgroundcolor=blue!25,bordercolor=blue,#1]{#2}}
\newcommandx{\info}[2][1=]{\todo[linecolor=OliveGreen,backgroundcolor=OliveGreen!25,bordercolor=OliveGreen,#1]{#2}}
\newcommandx{\improvement}[2][1=]{\todo[linecolor=Plum,backgroundcolor=Plum!25,bordercolor=Plum,#1]{#2}}

\newtheorem{theorem}{Theorem}

\newtheorem{lemma}{Lemma}

\newtheorem{assumption}{Assumption}
\theoremstyle{definition}
\newtheorem{definition}{Definition}
\theoremstyle{remark}

\makeatletter
\newtheorem*{rep@theorem}{\rep@title}
\newcommand{\newreptheorem}[2]{%
\newenvironment{rep#1}[1]{%
 \def\rep@title{#2 \ref{##1}}%
 \begin{rep@theorem}}%
 {\end{rep@theorem}}}
\makeatother
\newreptheorem{theorem}{Theorem}
\newreptheorem{lemma}{Lemma}
\newreptheorem{corollary}{Corollary}
\newreptheorem{proposition}{Proposition}

\DeclareMathOperator{\diag}{diag}

\DeclareMathOperator{\supp}{support}
\DeclareMathOperator*{\argmin}{arg\,min}

\newcommand{\One}[1]{{\mathbbm{1}}\left\{{#1}\right\}}
\newcommand{\norm}[1]{\lVert{#1}\rVert}
\newcommand{\Norm}[1]{\left\lVert{#1}\right\rVert}
\newcommand{\PP}[1]{\mathbb{P}\left\{{#1}\right\}} % Probability
\newcommand{\EE}[1]{\mathbb{E}\left[{#1}\right]} % Expectation
\newcommand{\EEn}[1]{\mathbb{S}_n\left[{#1}\right]} % Empirical Expectation Empirical Sum!!!

\newcommand{\EEst}[2]{\mathbb{E}\left[{#1}\ \middle| \ {#2}\right]} % Conditional expectation
\newcommand{\PPst}[2]{\mathbb{P}\left\{{#1}\ \middle| \ {#2}\right\}} % Conditional probability
\def\R{\mathbb{R}}

\newcommand{\ignore}[1]{}
\newcommand{\rbr}[1]{\left(#1\right)}
\newcommand{\sbr}[1]{\left[#1\right]}
\newcommand{\cbr}[1]{\left\{#1\right\}}

\newcommand{\abr}[1]{\left|#1\right|}

\DeclareMathOperator*{\ind}{1{\hskip -2.5 pt}\hbox{$\rm{I}$}}  % Indicator
\newcommand{\Gn}[1]{{\rbr{\mathbb{S}_n - \mathbb{ES}_n}\left[{#1}\right]}} % Empirical Sum - Expectation of Sum!!!
\newcommand{\bE}[1]{\mathbb{ES}_n\left[{#1}\right]} % Expectation of the sum!!!

\usepackage{fancyhdr}
\newcommand{\theauthor}{}
\newcommand{\thetitle}{}

\date{}
\author{\theauthor}
\title{\thetitle}

\fancypagestyle{plain}{%
  \fancyhf{}%
  \lhead{\thepage}
}

\newcommand{\as}{{a^\star}}
\newcommand{\bs}{{b^\star}}

\newcommand{\cs}{{c^\star}}
\newcommand{\ah}{\widehat{a}}
\newcommand{\at}{\widetilde{a}}
\newcommand{\bh}{\widehat{b}}
\newcommand{\bc}{\check{b}}
\newcommand{\ac}{\check{a}}

\newcommand{\bt}{\widetilde{b}}
\newcommand{\ch}{\widehat{c}}
\newcommand{\ct}{\widetilde{c}}

\newcommand{\dv}{\delta_v}

\newcommand{\db}{\delta_b}

\newcommand{\Loss}{\mathcal{L}}

\newcommand{\cK}{{\cal K}}
\newcommand{\cE}{{\cal E}}
\newcommand{\cF}{{\cal F}}
\newcommand{\cG}{{\cal G}}

\newcommand{\cU}{{\cal U}}

\newcommand{\DTay}{\Delta_{\textnormal{Taylor}}}

\newcommand{\qt}{\tilde{q}}

\newcommand{\fq}{{q_i}}
\newcommand{\fql}{{\tilde {q}_i}}
\newcommand{\fhq}{{\hat {q}_i}}

\newcommand{\ve}{\varepsilon}

\newcommand{\ES}{\mathbf{E}S}
\newcommand{\Hh}{\widehat{H}}

\newcommand{\Vs}{V^\star}
\newcommand{\Vh}{\widehat{V}}

\newcommand{\tr}{\textnormal{trace}}

\title{Post-selection inference on high-dimensional varying-coefficient quantile regression model}

\author{Ran Dai\thanks{Department of Biostatistics, University of Nebraska Medical Center, ran.dai@unmc.edu},  Mladen Kolar\thanks{Booth School of Business, University of Chicago, mkolar@chicagobooth.edu}}

\begin{document}

\maketitle

\begin{abstract}
Quantile regression has been successfully used to study heterogeneous
and heavy-tailed data. Varying-coefficient models 
are frequently used to capture changes in the effect of input 
variables on the response  as a function of an
index or time. In this work, we study high-dimensional
varying-coefficient quantile regression models and
develop new tools for statistical inference. We focus 
on development of valid confidence intervals and honest
tests for nonparametric coefficients at a fixed time point and quantile,
while allowing for a high-dimensional setting
where the number of input variables exceeds
the sample size. Performing statistical inference in this regime is
challenging due to the usage of model selection techniques in
estimation. Nevertheless, we can develop valid inferential
tools that are applicable to a wide range of data generating processes
and do not suffer from biases introduced by model selection. 
We performed numerical simulations to demonstrate the
finite sample performance of our method, and we also illustrated the
application with a real data example.

%%% Local Variables:
%%% TeX-master: "QR_score.tex"
%%% End:
\end{abstract}

\section{Introduction}

Most statistical work on regression problems has centered on the
problem of modeling the mean of a response variable $Y\in\R$ as a
function of a feature vector $X\in\R^p$. Under some assumptions, for
instance assuming homoscedastic Gaussian noise, modeling the mean is
sufficient to capture the entire distribution of $Y$ conditioned on
the observed features $X=x$. In many applications, however, where
these types of assumptions may not be appropriate, it is often far
more meaningful to model the median (or some other specified quantile)
of $Y$ given the observed feature vector $X$. In particular, in
applications where we are interested in extreme events---for instance,
modeling changes in stock prices, or modeling birth weight of
infants---modeling, e.g.,~the 90\% quantile may be far more informative
than modeling the mean. In other settings, the mean is overly
sensitive to outliers, while the median or some other quantile does
not have this disadvantage. Fixing $\tau$ to be the desired quantile
(e.g.,~$\tau=0.5$ for the median), we write $q(x;\tau)$ to be the
$\tau$th quantile for the variable $Y$ conditional on observing $X=x$,
that is, $q(x;\tau)$ is the function that satisfies
\[q(x;\tau) = \inf_{q\in \R} \{\PPst{Y\leq q}{X=x}\geq \tau\}.\]

In this paper, we are interested in a high-dimensional setting, where
the vector $X$ includes an extremely large number of measured
features---perhaps larger than the sample size itself. A linear model,
$q(x;\tau) = x^\top \beta(\tau)$, may be considered to be a reasonable
approximation in many settings, but if the measurements are gathered
across different points in time, the effect of the features on the
response $Y$ may not be stationary. To achieve broader applicability
of our model, we are furthermore interested in models with
time-varying coefficients for the $\tau$th quantile for the 
variable $Y$ conditional on observing $X=x$ at index $U=u$,
\[q(x;\tau,u) = \inf_{q\in \R} \{\PPst{Y\leq q}{X=x,U=u}\geq \tau\},\] 
where $x\in\R^p$ is the
feature vector as before, $\tau\in(0,1)$ is the desired quantile, and
$u\in\cU$ represents the time of the measurement or any other index
variable that captures non-stationary effects of the features---for
example, $u$ may be used to encode spatial location. 
We assume that $q(x;\tau,u)$ approximately follows a 
linear model $x^\top \beta(\tau,u)$.

Fixing a quantile $\tau$ and a time point (or index value) $u$, we are
interested in performing inference on a
low-dimensional subset of coefficients of interest,
$\beta_{\mathcal{A}}(\tau,u)$ for some fixed subset 
$\mathcal{A} \subset \{1,\dots,p\}$.
Specifically, we want to construct
confidence intervals for these
parameters or test null hypotheses such as
$H_0:\beta_j(\tau,u)=0, \forall j \in \mathcal{A}$. 
In practice, we may have in mind some
particular features of interest, and the other features are
confounding variables that we need to control for; or, we may be
interested in testing each of the $p$ features individually, cycling
through them in turn and treating the others as confounders.

\paragraph{Prior work}
Our work is related to the literature on high-dimensional inference, 
varying-coefficient models, and quantile regression.
Statistical inference for parameters in high-dimensional models
has received a lot of attention recently. For
example, in 
the $\ell_1$-regularized linear regression
model (LASSO) \citet{tibshirani96regression}
one can quantify the uncertainty about the unknown parameters
by debiasing 
the estimator 
\citep{Zhang2011Confidence,Geer2013penalized,Javanmard2013Nearly,Javanmard2013Hypothesis}
or using a 
double LASSO selection procedure \citep{Belloni2013Least}.
Extensions to generalized linear models 
were investigated in
\citet{Belloni2013Honest},
\citet{Geer2013asymptotically},
and \citet{Farrell2013Robust}.
\citet{Meinshausen2013Assumption} studied
construction of one-sided confidence intervals for groups of variables
under weak assumptions on the design matrix.
\citet{Lockhart2013significance} studied significance of the input
variables that enter the model along the lasso path.
\citet{Lee2013Exact} and \citet{Taylor2014Post} performed post-selection
inference conditional on the selected model.
\citet{Kozbur2013Inference} extended approach developed in
\citet{Belloni2012Inference} to a nonparametric regression setting,
where a pointwise confidence interval is obtained based on the
penalized series estimator, while
\citet{Lu2019Kernel} studied a  kernel-sieve hybrid estimator
for inference in sparse additive models.
\citet{Yu2019Constrained} considered testing in high-dimensional
parametric models with cone constraints. 
Hypothesis testing 
and confidence intervals for low-dimensional parameters
in graphical models were studied in
\citep{Ren2013Asymptotic, Wang2014Inference, Jankova2014Confidence, Jankova2017Honest},
elliptical copula models \citep{Barber2015ROCKET, Lu2015Posta}, 
Markov networks \citep{Wang2016Inference, Yu2016Statistical, Yu2019Simultaneous},
differential networks \citep{Xia2015Testing, Belilovsky2016Testing, Liu2017Structural, Kim2019Two},
and networks of point processes \citep{Wang2020Statistical}.
Varying-coefficient
models were introduced as a general framework that tied
together generalized additive models and dynamic generalized linear
models \citet{Hastie1993Varying}. Estimation and inference for
varying coefficient models in the mean have been widely studied.
See, for example, \citet{Fan2000Simultaneous}, 
\citet{Hoover1998Nonparametric}, 
\citet{Zhang2002Local}, 
\citet{Huang2004Polynomial},
\citet{Na2019High}, and
\citet{Na2018High}.
Quantile regression was studied in the presence of outliers and
non-normal errors in \citet{Koenker1984note}, while 
quantile regression with time-varying coefficient models 
was studied in, for example, \citet{Kim2007Quantile} and \citet{Kai2011New}.
Statistical inference for high-dimensional linear quantile regression
was studied in \citet{Belloni2013Robust, Belloni2013Uniform, Belloni2016Quantile, Bradic2017Uniform}
and a closely related problem of inference in composite quantile
regression was investigated in \citet{Zhao2014General}.
\citet{Tang2013Variable} studied estimation of quantile 
varying-coefficient models in a high-dimensional setting. 
However, how to perform statistical inference for high-dimensional 
varying coefficient models remains an open question.  

\paragraph{Our contribution} Below, we summarize the main contributions of this work.

\begin{itemize}
    \item We propose several approaches for constructing valid post-selection 
    confidence intervals for the varying-coefficient quantile regression model. 
    These approaches are asymptotically equivalent and rely on finding an 
    approximate root of the 
    decorrelated score.
    To make the construction computationally feasible with the non-differentiable 
    loss that is used in quantile regression, 
    we rely on a one-step approximation and reparameterization. 
    
    \item We provide the asymptotic normality results for the
    proposed estimators. Establishing this results requires 
    a novel analysis that generalizes the existing techniques. 
    Specifically, we carefully overcome the challenges that arise 
    from the non-differentiable loss used in quantile regression, 
    the bias from the penalized regression to handle the high 
    dimensionality, and the bias from linear approximation to
    handle the nonparametric component in appearing in varying coefficient
    models.
    
    \item We use extensive simulation studies 
    and real data analysis to demonstrate the finite sample performance
    of our proposed estimators.
    
\end{itemize}

%%% Local Variables:
%%% TeX-master: "QR_score.tex"
%%% End:

\section{Preliminaries}

In this section, we carefully develop background 
necessary to understand the algorithms that are presented
in the subsequent section.
In Section \ref{sec:VCQR_background} we provide 
a brief overview of estimation in the
varying-coefficient quantile regression model.
Next, we  illustrate the challenges in the high-dimensional inference 
in Section \ref{sec:2.2}. In
Section \ref{sec:2.3}, we describe 
the decorrelated score method that can be used for high-dimensional
inference when the loss function is twice differentiable. 
Finally, we modify
the decorrelated score method to suit the non-differentiable setting of
varying-coefficient quantile regression and sketch the main steps of
the analysis in Section \ref{sec:2.4}.
Note that the results in Section \ref{sec:VCQR_background}-\ref{sec:2.3}
are not new and are presented for ease of readability.

\subsection{Varying-coefficient quantile regression}\label{sec:VCQR_background}

For a random variable $Y \in \R$, its $\tau$-quantile can be equivalently described as the
value $q$ that minimizes the expectation
$\EE{\tau\cdot (Y-q)_+ + (1-\tau) \cdot (Y-q)_-}$ (for any $t\in\mathbb{R}$, 
we write $t_+=\max\{t,0\}$ and $t_- = \max\{-t,0\}$).  For a linear
quantile regression problem, at a particular value of the index
variable $u\in\cU \subseteq \R$, we are therefore interested in estimating
\begin{equation}\label{eqn:VCQR_E}\beta(\tau,u) = \argmin_{b\in\R^p} \EEst{\tau\cdot (Y-X^\top b)_+ + (1-\tau) \cdot (Y-X^\top b)_-}{U=u},\end{equation}
where the expectation is taken over a draw of the random pair $(X,Y)$
when the index variable is equal to $U=u$ (in other words, we can
think of drawing the triplet $(X,Y,U)$ and conditioning on the event
$U=u$).

Of course, we cannot compute this expected value or even obtain an
unbiased estimate, unless by some chance our training data contains
many data points $(x_i,y_i,u_i)$ with $u_i=u$. Instead, by assuming
that $\beta(\tau,u)$ is reasonably smooth with respect to the index
variable $u\in\cU$, we can use a kernel method, and approximate the
expected value in~\eqref{eqn:VCQR_E}
%over the index variable space $\cU$
with
\[\sum_{i=1}^n w_i \cdot \Big[\tau\cdot (y_i-x_i^\top b)_+ + (1-\tau) \cdot (y_i-x_i^\top b)_-\Big],\]
where the weights are given as
$w_i = (nh)^{-1} K\rbr{h^{-1}\rbr{u_i-u}}$,
the function $K(\cdot)$ is
the kernel function, and $h$ is the bandwidth. This approximation can be
interpreted as assuming that $\beta(\tau,u)$ is locally approximately constant for
values $u_i\approx u$, and thus defines a loss function on the sampled
data that would hopefully be minimized at some
$b(\tau, u) \approx \beta(\tau,u)$, but would suffer bias from the error in this
approximation. We can reduce the approximation bias by instead treating
$\beta(\tau,u)$ as locally approximately linear for values $u_i\approx u$, that is,
\[
  x_i^\top \beta(\tau,u_i) \approx x_i^\top \beta(\tau,u) + (u_i -u)\cdot x_i^\top \nabla_u\beta(\tau,u).
\]
Defining $\Gamma_i = (x_i^\top, (u_i-u)\cdot x_i^\top)^\top \in\R^{2p}$ 
for each observation $i=1,\dots,n$,
this yields a new loss function,
\begin{equation}\label{eqn:VCQR_Loss}
  \Loss(b)
  = \sum_{i=1}^n w_i \cdot \Big[\tau\cdot (y_i-\Gamma_i^\top b)_+ + (1-\tau) \cdot (y_i-\Gamma_i^\top b)_-\Big]
  = \sum_{i=1}^n w_i \cdot \rho_{\tau} \rbr{y_i-\Gamma_i^\top b},
\end{equation}
where the function $\rho_{\tau}(v) = v(\tau-\One{v<0})$.
We are now interested in
minimizing \eqref{eqn:VCQR_Loss} over a larger parameter vector,
$b=(b_0^\top,b_1^\top)^\top \in\R^{2p}$, where $b_0, b_1 \in \R^p$.
We would expect the minimum to be attained at some
$\bs = (b^{\star\top}_0,b^{\star\top}_1)^\top\approx (\beta(\tau,u)^\top ,\nabla_u \beta(\tau,u)^\top)^\top$ 
if the local linear approximation is sufficiently accurate. 
Note that we omit the indices $(\tau,u)$ to simplify the notation,
as they are fixed.

In a high dimensional setting where the dimension of the covariates
$X$, $p$, is growing faster than the sample size $n$,
we use a group $\ell_1$-penalty to estimate $\bs$ under the assumption
that the coefficient functions are approximately sparse.
In particular, we minimize the following optimization program
\begin{equation}\label{eqn:VCQR_bh}
  \bh = \argmin_{b \in \R^{2p}} \sum_{i=1}^n w_i \cdot \rho_{\tau} \rbr{y_i-\Gamma_i^\top b} + \lambda_b \norm{b}_{1,2},
\end{equation}
where $\norm{b}_{1,2} = \sum_{j=1}^p \sqrt{b_j^2 + b_{j+p}^2}$
is the $\ell_{1,2}$ group norm that simultaneously shrinks 
the coefficients $b_j$ and $b_{j+p}$, $j=1,\ldots,p$, to zero. 
Consistency results for $\bh$ have not been established in the existing literature
as it is challenging to deal with both the non-differentiable loss function
and a nonparametric model. Analysis for this model 
is more challenging compared to the partially
linear varying-coefficient model \citep{Wang2009Quantile},
where the nonparametric part is low-dimensional. 
Furthermore, the model in \eqref{eqn:VCQR_E} is strictly more general
than the partially linear varying-coefficient model.

\subsection{High dimensional inference}\label{sec:2.2}

We describe the challenges that arise in high-dimensional inference. 
Suppose first that we are interested in performing inference on a
low-dimensional parameter $\bs\in\R^p$, where the dimension $p$ is fixed
as the sample size $n$ tends to infinity.  After observing data, we
can estimate $\bs$ by minimizing some loss function
$\Loss(b) =\Loss(b;\text{data})$. For instance, in a regression
problem with features $x_i$ and response $y_i$, $i=1,\dots,n$, 
typically we would have 
$\Loss(b) ={n^{-1}} \sum_{i=1}^n\ell(b;x_i,y_i)$, 
where $\ell(\cdot)$ is the
negative log-likelihood under some assumed model.

In this classical setting, we can derive the well-known
asymptotically normal distribution of the estimator $\bh$ around the
true parameter value $\bs$, by considering the score
$\nabla\Loss(b)$. Namely, assuming that the loss is twice
differentiable, by taking a Taylor expansion, we can see that the
estimator $\bh$ satisfies
\begin{equation}\label{eqn:a_only_zero_score}
  0 = \nabla\Loss(\bh) = \nabla \Loss(\bs) + \nabla^2\Loss(\bs)\cdot(\bh - \bs) + \DTay,
\end{equation}
where $\DTay$ is the error in the Taylor expansion, equal to
\[
  \DTay = \rbr{\int_0^1 \nabla^2\Loss\left((1-t)\bs + t\bh\right)dt - \nabla^2\Loss(\bs) }\cdot(\bh - \bs).
\]
Then by solving for $\bh$, we have
\[\bh = \bs + \Big(-\nabla^2\Loss(\bs)\Big)^{-1} \cdot \Big(\nabla \Loss(\bs) + \DTay\Big).\]
Asymptotic normality of the error $\bh - \bs$ then follows from two
required properties: first, that the $\sqrt{n}$-score at the true parameter,
$\sqrt{n}\nabla \Loss(\bs)$, should be asymptotically normal via a central
limit theorem argument, while the Taylor expansion error
$\DTay$ is vanishing at some appropriately fast rate; and second, that the term $\nabla^2\Loss(\bs)$
should converge in probability to some fixed and invertible matrix
(specifically, to its expectation).

In high dimensions, however, the above analysis fails. If $b\in\R^p$
where the dimension $p$ grows faster than the sample size $n$, then
$\nabla^2\Loss(\bs)$ will likely not converge in probability, and in
general will not even be invertible. We can instead frame the argument
in terms of a low-dimensional parameter of interest combined with a
high-dimensional nuisance parameter. We write
$b = (a^\top,c^\top)^\top$, where $a\in\R^k$ is the low-dimensional parameter of interests,
while $c\in\R^{p-k}$ is the high-dimensional nuisance parameter. For
example, if we are working in a regression model, where the loss takes
the form $\Loss(b) = \sum_i \tilde{\ell}(y_i; x_i^\top b)$ for some loss
function $\tilde{\ell}$ (e.g.,~squared loss for a linear regression), then we
might decompose the high-dimensional parameter vector as
$b=(a^\top,c^\top)^\top$ to separate the coefficients on $k$
features of interest (without loss of generality, the first $k$
coordinates of the feature vectors $X_i$) and the remaining $p-k$
features, which we think of as potential confounders that need to be
controlled for in the regression.

Suppose that our estimate of the low-dimensional parameter vector of interest, $a$,
is obtained by solving
\[
  \ah = \argmin_a \Loss(a,\ct),
\]
where $\ct$ is some preliminary estimator of $c$. For example, in a
high-dimensional regression problem, we may run an $\ell_1$-penalized
regression first to obtain an initial sparse estimate of the
parameters. Once an initial estimate is obtained, we can 
refit the low-dimensional vector $a$ without a
penalty to remove the shrinkage bias. 
In this setting, we have
\begin{multline}\label{eqn:ac_zero_score}
0 = \nabla_a\Loss(\ah,\ct) = \nabla_a \Loss(\as,\cs) + \nabla^2_{aa}\Loss(\as,\cs)\cdot(\ah - \as) + \\\nabla^2_{ac}\Loss(\as,\cs)\cdot(\ct - \cs) + \Delta_{\text{Taylor2}}.\end{multline}
where 
\begin{multline}
 \Delta_{\text{Taylor2}}\\
 = \rbr{\int_0^1 \nabla^2_{aa}\Loss\left((1-t)\as + t\ah,(1-t)\cs + t\ct\right)dt - \nabla^2_{aa}\Loss(\as,\cs) }\cdot(\ah - \as)\\
  +\rbr{\int_0^1 \nabla^2_{ac}\Loss\left((1-t)\as + t\ah,(1-t)\cs + t\ct\right)dt - \nabla^2_{ac}\Loss(\as,\cs) }\cdot(\ct - \cs).
\end{multline}
Therefore,
\begin{equation*}\ah = \as + \Big(-\nabla^2_{aa}\Loss(\as,\cs)\Big)^{-1} \Big( \nabla_a \Loss(\as,\cs)+\nabla^2_{ac}\Loss(\as,\cs)(\ct - \cs)  +\Delta_{\text{Taylor2}}\Big).\end{equation*}
Let 
$S = (S_a^\top,S_c^\top)^\top = (\nabla_a \Loss^\top,\nabla_c\Loss^\top)^\top$ 
denote the score vector
and the negative Hessian matrix is
\[
  H = \left(\begin{array}{cc}H_{aa}&H_{ac}\\H_{ca} & H_{cc}\end{array}\right) = -\left(\begin{array}{cc}\nabla^2_{aa}\Loss&\nabla^2_{ac}\Loss\\\nabla^2_{ca}\Loss & \nabla^2_{cc}\Loss\end{array}\right).
\]
With this notation, we have
\[\ah = \as + \Big(H_{aa}(\as,\cs)\Big)^{-1}\cdot \Big(S_a(\as,\cs) - H_{ac}(\as,\cs)\cdot(\ct - \cs)+\Delta_{\text{Taylor2}}\Big).\]
To assure the asymptotic normality of the error $\ah-\as$,
we need to handle the following four terms:
\begin{itemize}
\item Asymptotic normality of $\sqrt{n} S_a(\as,\cs)$, which will hold by a central limit theorem argument as before;
\item Convergence in probability of $H_{aa}(\as,\cs)$ to a fixed invertible matrix, which holds since $a\in\R^k$ is low-dimensional;
\item Some control on the distribution of the term $H_{ac}(\as,\cs)\cdot (\ct - \cs)$;
\item Sufficiently small bound on $\Delta_{\text{Taylor2}}$, 
which will hold as long as we assume that $(\ah,\ct)$ is sufficiently close to $(\as,\cs)$.
\end{itemize}
The third term, $H_{ac}(\as,\cs)\cdot (\ct - \cs)$, is the main
challenge --- since $c$ is high-dimensional, in general it will not be
possible to explicitly characterize the distribution of the error
$\ct-\cs$ in its estimate. Therefore, we note that a naive 
refitting does not result in an asymptotically normal estimator 
and a different strategy is needed for high-dimensional inference.

One strategy to solve this problem is to modify the score method.
Specifically, we want the term
$H_{ac}(\as,\cs)\cdot (\ct - \cs)$ to vanish
at a sufficiently fast rate, 
so that it is smaller than the asymptotically normal term
$S_a(\as,\cs)$. The decorrelated score method, described next,
provides such a result.

\subsection{The decorrelated score method}\label{sec:2.3}

When $\ah$ is defined as the minimizer of the objective function
at some fixed estimator $\ct$ for the nuisance parameter,
$\ah = \argmin_a \Loss(a,\ct)$, we can equivalently obtain
$\ah$ as the solution to the score equation
$0 = \nabla_a\Loss(a,\ct)$. To decorrelate the score equations, 
we will instead define $\ah$ as the solution to $0 = \nabla_a\Loss(a,\ct) - V^\top  \nabla_c\Loss(a,\ct)$,
where $V \in \R^{(p-k) \times k}$ is a carefully chosen matrix,
whose choice will be discusses in detail shortly.
The Taylor expansion around the true parameter then gives us
\begin{multline} \label{eqn:ac_zero_score:modified} 
0 = S_a(\ah,\ct) - V^\top  S_c(\ah,\ct) = S_a(\as,\cs) - V^\top S_c(\as,\cs)\\
- \Big(H_{aa}(\as,\cs) - V^\top H_{ca}(\as,\cs)\Big)\cdot (\ah - \as)\\
- \Big(H_{ac}(\as,\cs) - V^\top H_{cc}(\as,\cs)\Big)\cdot (\ct - \cs) + {\rm Rem}.
\end{multline}
Solving for $\ah$, we then obtain
\begin{multline}\label{eqn:solve_decorr}
  \ah = \as +
  \Big(\underbrace{H_{aa}(\as,\cs) - V^\top H_{ca}(\as,\cs)}_{\text{Term 1}}\Big)^{-1}\cdot\\
  \bigg(\underbrace{\Big( S_a(\as,\cs) - V^\top S_c (\as,\cs)\Big)}_{\text{Term 2}} - \\
  \underbrace{\Big(H_{ac}(\as,\cs) - V^\top H_{cc}(\as,\cs)\Big)\cdot (\ct - \cs)}_{\text{Term 3}}  + {\rm Rem} \bigg).
\end{multline}
In order to show that $\ah$ is asymptotically normal,
we would like to show that Term 1 converges in probability to a
fixed (and invertible) matrix; Term 2 converges to a normal
distribution via a central limit theorem argument; and Term 3 is
vanishing (relative to Term 2). The role of the matrix $V$ is
precisely to make Term 3 of smaller order 
compared to Term 2. Specifically, the matrix $V$ is chosen so that
$\nabla^2_{ac}\Loss(\as,\cs) \approx V^\top \nabla^2_{cc}\Loss(\as,\cs)$,
enabling us to show that Term 3 is vanishing without obtaining a
limiting distribution for the high-dimensional estimator $\ct$.  In
general, the matrix $V$ cannot be known in advance and is therefore
data-dependent rather than fixed. However,
in applications we will have that $V$ converges to 
some fixed matrix sufficiently fast and all the
statements above still hold. 

Finding the roots of the score equation may be numerically
difficult. We present two methods that can be used in order to obtain
$\ah$ that approximately satisfies the score equation next.

The first method is the one-step correction method. Define
\[
  W = \left(\begin{array}{cc}I_k \\ -V \end{array}\right) \cdot \Big(H_{aa}(\as,\cs) - V^\top H_{ca}(\as,\cs)\Big)^{-\top}.
\]
Expanding $W^\top S(\at,\ct)$ at $(\as,\cs)$ and reorganizing the terms, we obtain
\[
  \at + W^\top S(\at,\ct) = \as + W^\top S(\as,\cs) - 
  W^\top H_{\cdot c}(\as,\cs) \cdot (\ct - \cs)
  + {\rm Rem},
\]
where $H_{\cdot c} = \left(\begin{array}{cc} H_{ac} \\ H_{cc} \end{array}\right)$
and $\bt = (\at^\top, \ct^\top)^\top$ is a preliminary, consistent estimator of $\bs$.
Note that the form of the equation above is the same
as in  \eqref{eqn:solve_decorr}.
This motivates
us to define the one-step corrected estimator
\[
  \ac^{OS} = \at + W^\top S(\at,\ct).
\]
Similar to the earlier discussion after \eqref{eqn:solve_decorr}, 
the normality of $\ac^{OS}$ will follow if we
choose the matrix $W$ so that $W^\top H_{\cdot c} \approx 0_{k,p-k}$ and $W$ 
itself converges to some fixed matrix sufficiently fast.

The second method for constructing $\ah$ relies on the reparametrization
of the loss function. In the method sketched above, 
$\ah$ is defined as the minimizer of the objective function  
at a fixed preliminary estimate $\ct$ of the nuisance parameter, 
i.e.,~$\ah = \argmin_a \Loss(a,\ct)$. We saw above that the
bottleneck in this analysis is the nonzero off-diagonal block of the Hessian matrix, 
$H_{ac}(\as,\cs)$. To avoid the problematic term in the Taylor expansion,
we can reparametrize the loss in such a 
way that the new off-diagonal block will become close to zero. 
Specifically, consider defining $\ah$ as the solution to a different optimization problem,
\begin{equation}\label{eqn:ah_decorr}
\ah = \argmin_a \Loss\big(a, \ct - V(a - \at)\big),\end{equation}
where $\at,\ct$ are preliminary estimates of $\as,\cs$.

To better understand the approach in \eqref{eqn:ah_decorr}, 
consider again a regression setting where the distribution 
of each response variable $y_i$ is modeled as a function of 
$x_i^\top b  = (x_{i,A}^\top,x_{i,A^c}^\top) (a^\top,c^\top)^\top$, 
where the subset $A\subset\{1,\dots,p\}$ indexes the $k$ features
of interest corresponding to the subvector $a$ of the regression coefficients.
In this setting, the negative Hessian matrix $H_{ac}(\as,\cs)$ 
will be nonzero whenever features in $A$ are correlated with 
features in $A^c$; thus, to set this block of the Hessian matrix 
to be (close to) zero, we can think of modifying the features of 
interest in the set $A$ by regressing out the confounding features in $A^c$. 
Specifically, let $v_j\in\R^{p-k}$ be the coefficient vector when 
regressing the feature $j\in A$ on all features in $A^c$. Then 
\[x_i^\top b  =(x_{i,A}^\top,x_{i,A^c}^\top) (a^\top,c^\top)^\top =  \big(x_{i,A} - V^\top x_{i,A^c}\big)^\top a + x_{i,A^c}  ^\top\big(c + V a\big),\]
where $V\in\R^{(p-k)\times k}$ is the matrix with columns $v_j$.
Note that, in this rearranged expression, 
the features of interest have been modified to be approximately 
orthogonal to, or approximately independent from, the nuisance features. 
Suppose we take $\ct + V \at$ as the preliminary estimate of
the coefficients $c+V a$ on the confounding features in this new model.
If we then re-estimate the parameter vector of interest $a$, 
obtaining a new estimate $\ah$, then the final fitted regression is given by
\[\big(x_{i,A} - V^\top x_{i,A^c}\big)^\top \ah + x_{i,A^c}^\top \big(\ct + V \at\big) = (x_{i,A},x_{i,A^c})^\top \big(\ah,\ct  - V (\ah - \at)\big),\]
thus motivating the form of the optimization problem given above in~\eqref{eqn:ah_decorr}.

Defining $\ah$ as the solution to the
decorrelated optimization problem~\eqref{eqn:ah_decorr}, 
the Taylor expansion then gives us
\begin{align*}
0 &=\nabla_a \Loss\big(a, \ct - V (a-\at)\big)\Big|_{a = \ah}\\
&=S_a(\ah,\ct - V(\ah - \at)) - V^\top  S_c(\ah,\ct - V(\ah - \at))\\
&= S_a(\as,\cs) - V^\top S_c(\as,\cs)
- \Big(H_{aa}(\as,\cs) - V^\top H_{ca}(\as,\cs)\Big)\cdot (\ah - \as)\\
&\hspace{.5in}-\Big(H_{ac}(\as,\cs) - V^\top H_{cc}(\as,\cs)\Big)\cdot (\ct  - V(\ah - \at)- \cs) + {\rm Rem},
\end{align*}
where ${\rm Rem}$ is redefined appropriately as 
the error term in this new expansion. Solving for $\ah$, we then obtain
\begin{multline}\label{eqn:solve_reparam}
\ah = \as + \Big(\underbrace{H_{aa}(\as,\cs) - V^\top H_{ca}(\as,\cs)}_{\text{Term 1}}\Big)^{-1}\cdot \bigg(\underbrace{\Big( S_a(\as,\cs) - V^\top S_c(\as,\cs)\Big)}_{\text{Term 2}}\\ - \underbrace{\Big(H_{ac}(\as,\cs) - V^\top H_{cc}(\as,\cs)\Big)\cdot (\ct  - V(\ah - \at)- \cs)}_{\text{Term 3}} + {\rm Rem}\bigg).\end{multline}
Therefore, $\ah$ is going to be asymptotically normal if 
Term 1 converges in probability to a fixed (and invertible) matrix;
$\sqrt{n}\cdot$Term 2 converges to a mean-zero normal distribution via a central limit theorem argument; 
and Term 3 and the remaining error ${\rm Rem}$ are vanishing (relative to Term 2).
As before, the role of the matrix $V$ is in controlling Term 3.
Specifically, the matrix $V$ is chosen so that 
$H_{ac}(\as,\cs) \approx V^\top H_{cc}(\as,\cs)$, 
enabling us to show that Term 3 is vanishing without obtaining a
limiting distribution for the high-dimensional initial estimates $\at,\ct$.
In general, the matrix $V$ cannot be known in advance 
and is therefore data-dependent rather than fixed, but in our analysis we will see 
that as long as $V$ itself is sufficiently close to some fixed matrix, all the statements above will still hold.

\subsection{Non-differentiable loss in quantile regression}\label{sec:2.4}

When the loss function $\Loss$ is non-differentiable, 
which is the case in quantile regression, 
approaches based on the decorrelated score method 
cannot be directly applied. However, 
a simple modification allows us to proceed in a similar
way as before. Assuming that the loss is nondifferentiable
and convex, we let $S(a,c)$ denote the subdifferential of the loss.
While $S(\cdot)$ might be highly nondifferentiable, 
its expected value is smooth in many problems.
Therefore, we can compute the Hessian as 
the gradient of the expected value of $S(\cdot)$. 
In particular, we define the expected score function $\ES(a,c)$ 
as the expectation of the score $S(\cdot)$
at any {\em fixed} parameter choice $(a,c)$.
Here it is important to note that, 
for a random parameter vector $(\ah,\ct)$, 
the expected score function $\ES(\ah,\ct)$ is {\em not} equal to 
$\EE{S(\ah,\ct)}$, since this second 
quantity would evaluate its expectation 
with respect to the random values of $\ah$ and $\ct$ as well.
With the expected score function defined, we let 
\[H(a,c) = -\nabla \ES(a,c)\]
be the negative gradient of the expected score.

We specialize the discussion so for to the quantile regression problem at hand.
We will base the inference procedures 
on
the local linear formulation of the estimation problem 
for the varying-coefficient quantile regression model \eqref{eqn:VCQR_bh}.
Suppose $A=\{1,\dots,k\}$ is the index set for the parameters 
of interest. Let $Y \in \R^n$ be the response 
and $U \in \R^n$ be the index for the varying coefficient. 
The matrix of input variables is denoted as 
$X = (X_A,X_{A^c}) \in \R^{n \times p}$, 
where $X_A \in \R^{n\times k}$ represents 
the features of interest and $X_{A^c} \in \R^{n\times (p-k)}$ the other features. 
Let 
\[\Gamma(u) = \left(X_A, X_{A^c},\diag (U-u)X_A, \diag (U-u)X_{A^c} \right) \in \R^{n \times 2p},\] 
and 
$\Gamma_i^\top(u) = (x_{i, A}^\top, x_{i, A^c}^\top, (u_i - u)\cdot x_{i, A}^\top,  (u_i - u)\cdot x_{i, A^c}^\top)$ represents the $i$th row vector of $\Gamma(u)$. 
The score function for quantile regression is given as 
\begin{multline*}
S(a_0, a_1, c_0, c_1) = 
\sum_{i\in[n]} w_i \cdot \Gamma_i(u) \cdot \Psi_\tau
\left(
y_i - x_{i,A}^{\top} a_0 - x_{i, A^c}^\top c_0 
\right.
\\
\left.
- (u_i - u)\cdot x_{i, A}^{\top} a_1 -  (u_i-u)\cdot x_{i, A^c}^\top c_1\right),
\end{multline*}
where $\Psi_\tau(u) = \tau - \ind(u < 0)$.
Let $b_0 = (a_0^\top, c_0^\top)^\top \in \R^{p}$, $b_1 = (a_1^\top, c_1^\top)^\top \in \R^p$,
and $b = (b_0^\top,b_1^\top)^\top \in \R^{2p}$. 
Then the above score function can be
written as 
\begin{equation}\label{eqn:score}
S(b) = 
\sum_{i\in[n]} w_i \cdot \Gamma_i(u) \cdot \Psi_\tau\rbr{y_i - \Gamma_i^\top(u)\cdot b}\;.
\end{equation}
Let $\bs = \bs(\tau, u)$ be defined as a solution to 
$0 = \EE{S(b)}$ when $h\rightarrow 0$. 
Let $\qt_i (\tau, u) = \Gamma_i^\top(u)\bs(\tau, u)$ be a
local linear approximation to $q_i(\tau)=q(x_i;\tau,u_i)$.
Since $(\tau, u)$ is fixed, 
we write $\Gamma_i(u) = \Gamma_i$, $\qt_i (\tau, u) = \qt_i$ and $q_i (\tau) = q_i$
for notational simplicity. 
Finally, we use 
$\Delta_i = \Delta_i(\tau, u) = \tilde{q}_i - q_i$
to denote the approximation error from using the local linear 
model for the conditional quantile.

An approximate negative Hessian corresponding to the expected score function is given as 
\begin{equation}\label{eqn:Hstar}
H^{\star} = H(\bs; \tau, u) = \sum_{i\in[n]} w_i \cdot f_i(q_i + \Delta_i) \cdot \Gamma_i\Gamma_i^\top.
\end{equation}
Let $\Vs \in \R^{2k \times 2p}$ be the rows related 
to $X_A$, $X_A(U-u)$ of an approximate inverse of $H$ such that 
\[ 
\norm{\Vs H^\star - E_a}_{\infty,F} \leq \lambda^{\star},
\] 
where 
\[
\norm{\Vs}_{\infty,F} =\sup_{i \in [k], j \in [p]} \norm{\Vs_{(i,i+k),(j,j+p)}}_F,
\]
$E_a = (e_1, \cdots, e_{2k})^{\top} \in \R^{2k\times 2p}$, and 
$\lambda^{\star}$ is a  parameter 
that will be precisely given in Section \ref{main}.

With these preliminaries,
we define the one-step correction estimator $\ac^{OS}$ as $\ac^{OS} = \ah - S_d (\bh, \Vh)$, where \begin{equation}\label{eqn:S_d}
S_d(b, V) := \sum_{i \in [n]} S_{di}(b,V) = \sum_{i \in [n]} -w_iV\Gamma_i\Psi_{\tau} (y_i - \Gamma_i^{\top}b),\end{equation} 
and $\Vh, ~\bh $ are plug-in estimators of $\Vs,~ \bs$ to be defined later, and $\ah = \bh_{1:2k}$.

%%% Local Variables:
%%% TeX-master: "QR_score.tex"
%%% End:

\section{Algorithm} \label{algorithm}

We provide computational details for the three proposed 
estimators. The first estimator is  based on finding the root of the decorrelated-score;
the second estimator is based on the one-step correction; and
the third estimator is based on the reparametrization of the loss function. 
As discussed in the previous section, all these estimators are 
asymptotically equivalent.
Estimation proceeds in three steps with the first
two steps being the same for all three estimators.
In the first step we obtain a pilot estimator of $\bs$,
while in the second step we obtain $\hat V$.
We provide details next.

\noindent {\bf Step 1.} Obtain 
the initial estimator $\bh^{\rm ini}$ by minimizing the 
optimization program \eqref{eqn:VCQR_bh}.
The kernel weights are given as
$w_i = (nh)^{-1}K\rbr{h^{-1}{U_i-u}}$,
while the penalty parameter $\lambda_b$ is defined 
in a data dependent fashion
as
\begin{equation}
\lambda_b =c_b\sqrt{\tau(1-\tau)\log (nhp)}\cdot \rbr{\max_{j \in [p]}\EEn{w_i^2x_{ij}^2}}^{1/2},
\end{equation}
where $\EEn{z_i}$ denotes the summation,
$\EEn{z_i} = \sum_{i \in [n]} z_i$,
and $c_b$ is a data independent constant. 
We subsequently threshold elements of $\bh^{\rm ini}$
to obtain $\bh$ with
\[
\bh_j = 
\left\{
\begin{array}{ll}
\bh^{\rm ini}_j \cdot \One{\rbr{ \bh^{\rm ini}_j }^2+\rbr{\bh^{\rm ini }_{j+p}}^2 > \lambda_b^2}, & j = 1, \ldots, p,\\
\bh^{\rm ini}_j \cdot \One{\rbr{ \bh^{\rm ini 2}_j}^2+\rbr{\bh^{\rm ini }_{j-p}}^2 > \lambda_b^2}, & j = p+1, \ldots, 2p,
\end{array}
\right.
\]
to ensure the sparsity of
the estimator for Theorem \ref{thm:b:consistent&sparse}.

\noindent {\bf Step 2.}
Obtain $\Vh$ by
\begin{equation}
 \label{eq:lasso_fixed:opt}
 \Vh = \argmin_{V \in \R^{2k\times 2p}} \cbr{\text{trace}\rbr{\frac{1}{2}V\Hh V^{\top} - E_aV^{\top} } + \lambda_V \norm{V}_{1,F}},
\end{equation}
where $\lambda_V = n^{-1}c_v\sqrt{nh}\Phi^{-1}\rbr{1-\frac{0.05}{2nhp}}$
with $c_v$ being a data independent constant; 
and $\Hh = \sum_{i \in [n]}w_i\hat{f}_i\Gamma_i \Gamma_i^{\top}$
with $w_i$ being the kernel weight as defined in Step 1 
and $\hat{f}_i$ is computed with a data adaptive procedure as
\[\hat{f}_i = \frac{\One{|y_i - \Gamma_i^{\top}\bh| \leq h_f}}{2h_f}\]
with
\[
h_f=(\Phi^{-1}(\tau+h_p)-\Phi^{-1}(\tau-h_p))
\min\cbr{\sqrt{\text{Var}(\hat{e})}, \frac{Q_{0.75}(\hat{e})-Q_{0.25}(\hat{e})}{1.34}},
\]
where
\begin{gather*}
\hat{e}_i=y_i-\Gamma_i^{\top}\bh, \quad
Q_{\alpha}(\hat{e})=\inf_q \cbr{q: \frac{\sum_i w_i\One{\hat{e}_i\leq q}}{\sum_i w_i}\geq \alpha},\\
\text{Var}(\hat{e})=\frac{\sum_i w_i \rbr{\hat{e}_i-\frac{\sum_j w_j\hat{e}_j}{\sum_j w_j}}^2}{\sum_i w_i},
\end{gather*}
and
\[
h_p=n^{-1/3}\{\Phi^{-1}(0.975)\}^{2/3}\cbr{\frac{1.5[\phi(\Phi^{-1}(\tau))]^2}{2[\Phi^{-1}(\tau)]^2+1}}^{1/3}
\]
is  the Powell bandwidth defined in \citet[][Section 3.4.4]{Koenker2005Quantile}.

\noindent {\bf Step 3.}
Here we obtain our final estimator $\check{a}$ using one of the three procedures. 

\begin{enumerate}
    \item Finding the root of the decorrelated score. We would like to construct
    $\check{a}$ by solving for $\sum_i \widetilde{S}_i(a) = 0$ where 
    $\widetilde{S}_i(a)= S_{di}((a^\top, \ch^\top)^\top, \Vh)$
    with $S_{di}$ defined in \eqref{eqn:S_d}.
    Because $\sum_i \widetilde{S}_i(a)$ is not continuous,
    we can approximately solve the equation by 
    \begin{equation}\label{eqn:DS est} 
    \ac^{DS}=\argmin [\sum_i \widetilde{S}_i(a)]^\top[\sum_i \widetilde{S}_i(a)\widetilde{S}_i^\top(a)]^{-1} [\sum_i \widetilde{S}_i(a)].
    \end{equation}
    Minimizing the above problem is not computationally simple. The following two strategies
    might be preferred, as discussed in Section \ref{sec:2.3}.
    
    \item The one step correction estimator. We compute the estimator as
    \begin{equation}\label{eqn:one step est}
    \ac^{OS} = \ah + \sum_{i \in [n]} w_i  \Vh\Gamma_i^\top \Psi_{\tau} (y_i - \Gamma_i^{\top}\bh).
    \end{equation} 
 
    \item The reparametrization estimator. We first obtain $\tilde{\Gamma}_i, \tilde{y}_i$ as
    $\tilde{\Gamma}_i=\Gamma_{i, A}- \Vh_2\Gamma_{i, A^c}$ and 
    $\tilde{y}_i=y_i-\Gamma_{i, A^c}^\top(\ch +\Vh_2^\top\ah)$, 
    where $\Vh_2 = \Vh_{11}^{-1} \Vh_{12}$, $\Vh_{11} = \Vh_{1:2k} \in \R^{2k \times 2k}$ 
    and $\Vh_{12} = \Vh_{(2k+1):2p} \in \R^{2k \times 2(p-k)}$. Then the estimators is 
    computed as 
    \begin{equation}\label{eqn:RP est}
    \ac^{RP}=\argmin_{a} \sum_i w_i \cdot \rho_{\tau} \rbr{\tilde{y}_i-\tilde{\Gamma}_i^\top a}.
    \end{equation}
    
\end{enumerate}

With the estimator $\ac$, being $\ac^{DS}$, $\ac^{OS}$, or $\ac^{PR}$,
we can perform statistical inference about 
the parameter of interests, $a$. For any one of the three estimamors, 
we have that 
\[
\widehat{\Sigma}_a^{-1/2}(\ac -a) \sim \mathcal{N}(\mathbf{0}, I_k),
\]
where the covariance matrix is computed as 
\[ 
\widehat{\Sigma}_a = nh\sum_i w_i^2 \Vh\Gamma_i^\top \Psi_{\tau} ^2(y_i - \Gamma_i^{\top}\bh)  \Gamma_i \Vh^{\top}.
\]

We end this section with some remarks on the computation.
The kernel weights we chose in our simulation studies 
are given as 
\[w_i= (nh)^{-1}\One{|U_i-u|/h<0.5}\]
with $h=c_h n^{-1/3}$
and $c_h = 4$. However, we note that any
kernel function 
that
satisfies Assumption \ref{assumption:kernel} presented later in Section \ref{assumptions}
can be used. 
Many frequently used kernels, such as 
the Gaussian kernel and box kernel, satisfy this assumption. 
We also set $c_b=0.4$ and
$c_v = 0.02$  in our numerical studies. 
The performance of the algorithm is not very sensitive to the choice of these parameters.
Both in Step 1 and Step 2 of the algorithm, one can perform
optional refitting of the selected coefficients to
improve finite sample performance. For example, in Step 1,
let $S \in \R^s$
be the support of the covariates in $\Gamma$ corresponding to $\bh$, 
and $\Gamma_{i, S}$ and $b_{S}$ are the corresponding entries in $\Gamma_i$ and $b$, 
then \[
\bh_{post} = \left\{ \argmin_{b \in \R^{p}} \sum_{i=1}^n w_i \cdot \rho_{\tau} \rbr{y_i-\Gamma_{i, S}^\top b_S}: b_j=0 ~ \forall j\in S^c\right\}
\] 
can be used to replace $\bh$.

%%% Local Variables:
%%% TeX-master: "QR_score.tex"
%%% End:

\section{Main results}\label{main}

In this section, we present our main results.
We start by detailing the assumptions in Section~\ref{assumptions}.
Results on estimation consistency are presented in Section~\ref{sec:con}.
Finally, we give results on the asymptotic normality of the estimator in 
Section~\ref{sec:norm}.

\subsection{Assumptions}\label{assumptions}

We state the assumptions needed to establish our results.

\begin{assumption}
\label{assumption:kernel} \textnormal{(Kernel assumptions)} The kernel
function $K(\cdot)$ satisfies 
\begin{gather*}
K(t) \leq \nu_0 < \infty \quad \text{for all $t$,} \quad
\int K(u)du = \nu_1 < \infty, \\
\int K^2(u)du = \nu_2 < \infty,\quad
\int K(u)u^2du = \mu_2 < \infty,\\ \int K(u)u^4du = \mu_4 < \infty.
\end{gather*}
\end{assumption}

The kernel is chosen by a statistician, 
so the above assumption does not put restrictions on the data generating process. A number of
standard kernels such as the Gaussian kernel, box kernel, and Epanechnikov kernel,
all satisfy the above assumption. 

\begin{assumption}
\label{assumption:u} \textnormal{(Assumptions on $U$)}
We assume $U$ has bounded support.
Without loss of generality, we assume $U \in [0,1]$.
Let $f_U(u)$ be the density of $U$.
There exists $\bar{f}$ such that $f_U(u)\leq \bar{f}$. 
\end{assumption}

From Assumptions \ref{assumption:kernel} and \ref{assumption:u}, 
the kernel weights $w_i$'s satisfy the following
with high probability
\begin{gather*}\label{eqn:w_assumptions}
    \norm{w_i}_{\infty} \leq B_w = \frac{\nu_0}{nh} \leq \frac{B_K}{nh},
    \quad
    \sum_i w_i \leq B_K, \\
    \sum_i w_i \rbr{\frac{u_i - u}{h}}^2 \leq B_K,
    \text{ and } 
    \sum_i w_i \rbr{\frac{u_i - u}{h}}^4 \leq B_K 
\end{gather*}
for some constant $B_K > 0$.

\begin{assumption}\textnormal{(Assumptions on the distribution of $Y$)}
\label{assumption:density}
Let $f_i(y)$ be the conditional density of $Y_i$ given $X_i = x_i, U_i = u_i$.
We assume that there exist constants $\underline{f}, \bar{f}, \bar f'$
such that  
  \[0<\underline{f} \leq f_i(y) \leq \bar{f}, \quad \textnormal{and} \quad |f_i'(y)| \leq \bar f' \quad \text{for all}~y.\] 
\textnormal{This type of assumption on the conditional distribution of $Y$ is commonly used in the literature on quantile regression, for example, see \cite{Belloni2016Quantile}.}
\end{assumption}

%let $\beta^{\star}(u, \tau) = \argmin_{\beta: |S(\beta)|\leq s } \EEst{\rho_{\tau} \left(y - X\beta\right)}{U=u}$.

\begin{assumption}\textnormal{(Approximate linear sparsity and smoothness of $q(x;\tau,u)$)} \label{assumption:qr}
  Assume there exists a smooth and sparse $\beta^{\star}(\tau, u)$ such that:
  \begin{itemize}
      \item  $u \mapsto \beta^{\star}(\tau, u)$ is differentiable for all $\tau \in [\ve, 1-\ve]$ and 
    \[
    \norm{\beta^{\star}(\tau, u') - \beta^{\star}(\tau, u) - (u' - u)\cdot \nabla_u \beta^{\star}(\tau, u) }_2 \leq B_\beta(u'-u)^2;
    \]
    \item the supports of $\beta^{\star}(u, \tau)$ and
    $\partial_u \beta^{\star}(\tau, u)$ are sparse;
    that is, for the sets
    \[
S := \cbr{j \in [p] \mid \beta_j^{\star}(\tau, u) \neq 0}
\]
and
\[
S' := \cbr{j \in [p] \mid \beta^{\star}_j(\tau, u) \neq 0 \text{ or }  \partial_u \beta^{\star}_j(\tau, u) \neq 0},
\]
we have $s := |S| \ll n$ and $ |S'|\leq s_1 := c_1 s$ for some constant $c_1$.
  \end{itemize}
We assume 
that the quantile function $q_i = q(x_i; \tau, u_i)$ can be well approximated by
a linear function $x_i^\top \beta^{\star}(\tau, u_i)$; specifically, 
   \[   \PPst{Y \leq x_i^\top \beta^{\star}(\tau, u_i)}{X=x_i, U = u_i} = \tau + R_{i},\]
    where 
    \[
    \sqrt{\sum_i w_iR_{i}^2} = \epsilon_R = O\rbr{\sqrt{\frac{\log (np)}{nh}} }.
    \]
\end{assumption}

\noindent
This assumption requires that the
conditional quantiles of $Y$ approximately 
follow a linear varying-coefficient model and 
the approximation error is vanishing as 
$n \rightarrow \infty$. In addition, the varying-time
coefficient $\beta^\star(\tau, u)$ is H\"{o}lder smooth,
sparse and has sparse first derivatives.

For the case when $u$ and $\tau$ are fixed, we will write $\beta^{\star} = \beta^{\star}(\tau, u)$. Let $\bs =\bs(\tau,u) = \rbr{\beta^{\star}(\tau,u)^\top, \nabla_u^\top \beta^{\star}(\tau,u)}^\top $
and write $\fq = q(x_i;\tau,u_i)$, $ \fql = \Gamma_i^\top \bs$ as a local linear approximation to $\fq$. Let 
\[H =\sum_{i \in [n]} w_i \cdot f_i(\fq) \cdot \Gamma_i\Gamma_i^\top \text{ and } H^{\star} =\sum_{i \in [n]} w_i \cdot f_i(\fql) \cdot \Gamma_i\Gamma_i^\top.\]

\begin{assumption} 
\label{assumption: lam}
\textnormal{(Assumptions on the Hessian)}
Let $\Vs$ be the rows related to $[X_A, X_A (U-u)]$ of an approximate inverse of $H^\star$. 
We assume that 
\[\norm{H^\star\Vs - E_a}_{\infty, F} \leq \lambda^{\star}=O\rbr{ B_V \sqrt{ \frac{\log p}{nhh_f} }},
\qquad
\norm{\Vs}_{0,F} \leq s_2 = c_2 s,\] 
where $\norm{V}_{0,F} := |
\{(i,j):i \in [k], j \in [p], \Vs_{(i,i+k),(j,j+p)}\neq 0\}
|$, 
$\norm{V}_{\infty,F} := \max_{i \in [k], j \in [p]} \norm{\Vs_{(i,i+k),(j,j+p)}}_F$ 
and 
$\max_{i \in [n]}\|\Vs \Gamma_i\|_2 = B_V \asymp \log p$.
Furthermore, we have 
\[\text{$(nh)^{-1}sB_V^2\log p = o(1)$ ~\textnormal{and}~ $\log(B_V^2h_fh) = o(\log p)$.}\]
\end{assumption}

Assumption \ref{assumption: lam} holds when $X_A$ follows a 
multivariate approximately sparse linear model
with respect to $X_{A^c}$, where we require the coefficients
to be approximately linear, sparse and smooth (see Appendix \ref{app:c}).
For example, when the distribution of $X$ does not depend
on $U$ and the exact sparse linear model holds, 
Assumption \ref{assumption: lam} obviously holds. 

\begin{assumption}
\label{assumption:X}
  \textnormal{(Assumptions on X)} 
  We make the following assumptions on the covariate $X$: 
  \begin{itemize}
 \item Boundedness: there exists a constant $B_X$ such that with high probability,
 
   \[\max_i \norm{x_i}_{\infty} \leq B_X 
   \text{ and }
   \max_{j \in [p]} \sum_i w_i^2 x_{ij}^2 \leq \frac{B_X^2 B_K^2}{nh}.
   \]
  
\item Restricted eigenvalues: Consider the following cones
\begin{align*}\mathbb{C}(s_1) &= \{\theta:\norm{\theta}_0 \leq s_1 ~\text{and}~ \norm{\theta}_2 = 1 \}, ~\text{and} \\
\mathbb{C}(S_2) &= \{\Theta \in \R^{2k\times 2p}: \norm{\Theta_{S_2^c}}_{1,F} \leq 6 \norm{\Theta_{S_2}}_{1,F} ~\text{and}~\norm{\Theta}_{1,F} = 1\},\end{align*} 
where $S_2 = \left\{(i,j): \norm{\Vs_{(i,i+k),(j,j+p)}}_F > 0 \right\}$ is the support of $\Vs$
By Assumption \ref{assumption: lam}, $|S_2|\leq s_2$. 
We assume there exist $0 < \kappa_-^2 \leq \kappa_+^2 < \infty$
such that
\begin{align}\label{eqn:growth}
\kappa_-^2 &\leq \sum_{i \in [n]}w_i (\Gamma_i^\top \theta)^2 \leq \kappa_+^2 \quad \textnormal{ for all } \theta \in \mathbb{C}(s_1) ~\textnormal{and} \\
\kappa_{-}^2 &\leq \sum_{i \in [n]}  w_i\textnormal{trace}(\Theta^\top\Gamma_i\Gamma_i^\top \Theta) \leq \kappa_{+}^2 \quad \textnormal{ for all } \Theta \in \mathbb{C}(S_2).\end{align}

\item For some constant $\kappa_q > 0$,
\begin{equation}\label{eqn:growth2}
    \inf_{
\substack{\norm{\delta}_{1,2} \leq \frac{7|S'| \cdot \sqrt{\log p}}{\kappa_- \sqrt{nh}} \\ \EEn{w_i(\Gamma_i^\top\delta)^2} = \frac{|S'|\log p}{nh} }
}
      \frac{\rbr{ \underline{f} \cdot \sum_{i\in [n]}w_i\cdot\rbr{\Gamma_i^\top \delta}^2}^{3/2}}{\bar{f}' \cdot \sum_{i \in [n]}w_i\cdot\rbr{\Gamma_i^\top \delta}^3} \geq \kappa_q.
  \end{equation} 
  
\end{itemize}
\end{assumption}

\noindent The assumptions 
on the design $X$ are mild and commonly used in the literature on high-dimensional 
estimation and inference.
For example, boundedness and restricted eigenvalue condition 
was used in \cite{negahban2010unified}. 
The condition \eqref{eqn:growth2} is a mild growth condition, 
which is satisfied for many design matrices $X$, 
see \cite{Belloni2016Quantile} and \cite{Belloni2011penalized}.

Finally, we need the following growth condition.
\begin{assumption}\label{assumption:growth}
\textnormal{(Growth conditions)}
We assume 
\[
h \asymp O(n^{-1/3}), \quad
h_f \asymp O(n^{-1/3}), \quad\text{and}\quad
\rbr{nhh_f}^{-1/2}{s \log p \log (np)} = o(1).
\]
\end{assumption}

With these assumptions, we are ready to present our main results next.

\subsection{Consistency and sparsity results of the initial estimators}
\label{sec:con}

We establish
the asymptotic properties of the initial estimators in 
Step 1 and Step 2 from Section \ref{algorithm}.

\begin{theorem}
  \label{thm:b:consistent&sparse}
 Under Assumptions \ref{assumption:kernel}, \ref{assumption:u}, \ref{assumption:density}, \ref{assumption:qr}, \ref{assumption:X} and \ref{assumption:growth},  the estimator $\bh$ from Step 1 in Section \ref{algorithm} satisfies
\begin{align}
\EEn{w_i \cdot \rbr{\Gamma_i^\top \rbr{\bh - \bs}}^2} 
& \leq O_p\left(\frac{s\log (np)}{nh}\right), \label{eqn:thm_b_1} \\ 
\norm{\bh-\bs}_{1,2} & \leq O_p\rbr{s\sqrt{\frac{\log (np)}{nh}}}, ~\textnormal{and} \label{eqn:thm_b_2} \\
\norm{\bh}_{0,2} &\leq O_p(s),  \label{eqn:thm_b_3}
\end{align} 
where $\EEn{z_i} := \sum_{i \in [n]} z_i$, $\norm{b}_{1,2} := \sum_{i= 1}^p \sqrt{b_i^2+b_{i+p}^2} $, and $\norm{b}_{0,2}$ is defined as the $\ell_0$ norm of the vector $\left(\sqrt{b_1^2+b_{1+p}^2},\cdots, \sqrt{b_p^2+b_{2p}^2}\right)$.
\end{theorem}
 
Theorem \ref{thm:b:consistent&sparse} gives us convergence results 
regarding the $\ell_{1,2}$-penalized quantile regression estimator 
from \eqref{eqn:VCQR_bh}. In particular, \eqref{eqn:thm_b_1} gives
the rate of convergence rate of the prediction $\ell_2$-norm, 
\eqref{eqn:thm_b_2} gives the $\ell_{1,2}$-norm of the error, 
and \eqref{eqn:thm_b_3} gives the sparsity of $\bh$ in Step 1 of Section \ref{algorithm}. 
Both Theorem \ref{lem:qr_fixed:lasso:rate_v} and 
Theorem \ref{thm:normality} rely on these conditions. 
The extra growth condition in Assumption \ref{assumption:X} 
is mild. Specifically, with the penalty parameter
$\lambda_b \asymp O\rbr{\sqrt{\frac{\log p}{nh}}}$
the assumption is satisfied. The sparsity here is achieved
by truncating the small components in $\bh^{\rm ini}$ 
to zero, while maintaining the same rate of convergence.    
 
 \begin{theorem}
  \label{lem:qr_fixed:lasso:rate_v}
 Suppose the assumptions for Theorem \ref{thm:b:consistent&sparse} hold 
 and the estimator $\bh$ obtained in Step  1
 satisfies \eqref{eqn:thm_b_1}, \eqref{eqn:thm_b_2}, and \eqref{eqn:thm_b_3}.
 Furthermore, suppose that $\lambda_V \geq 2\lambda^{\star}$
 and Assumption \ref{assumption: lam} holds for $\lambda^{\star}$. 
 Then $\Vh$ from Step 2 satisfies 
\begin{align}
     \norm{\Vh -\Vs}_F &\leq O_p\rbr{B_V\sqrt{\frac{s\log (np)}{nhh_f}}}
     \quad \text{and} \label{eqn:thm_v_1} \\
     \norm{\Vh -\Vs}_{1,F} &\leq O_p\rbr{sB_V\sqrt{\frac{\log(np)}{nhh_f}}},
     \label{eqn:thm_v_2}
 \end{align} 
 where $\norm{V}_{1,F} := \sum_{i \in [k], j \in [p]} \norm{V_{(i,i+k),(j,j+p)}}_F$.
\end{theorem}

Theorem \ref{lem:qr_fixed:lasso:rate_v} gives the convergence rate of 
the $\ell_{1,2}$-norm and prediction $\ell_2$-norm of $\Vh$ in Step 2
of Section \ref{algorithm}. Because the $\widehat{H}$ in 
the objective function relies on the estimator $\bh$, 
both the convergence and sparsity results 
from Theorem \ref{thm:b:consistent&sparse} are needed.
 
\subsection{Normality result of the final estimators}
\label{sec:norm}

We state the asymptotic normality result for the one step estimator.

\begin{theorem}
  \label{thm:normality}
\textbf{(Normality for the one-step estimator)}
Assume that Assumptions \ref{assumption:kernel}~--~\ref{assumption:growth} hold and \eqref{eqn:thm_b_1}, \eqref{eqn:thm_b_2}, \eqref{eqn:thm_b_3}, \eqref{eqn:thm_v_1}, \eqref{eqn:thm_v_2} hold. Then the one step estimator defined in (\ref{eqn:one step est}) satisfies 
\begin{align*}
\sqrt{nh}\widehat{\Sigma}_a^{-1/2}(\check{a}^{OS}-a^{\star})\rightarrow_d \mathcal{N}(0, I_{2k}),
\end{align*}
where the covariance matrix $\widehat{\Sigma}_a$
is estimated as 
either 
\begin{align}
\widehat{\Sigma}_a 
&=nh\Vh\left\{\sum_i w_i^2  \Gamma_i\Psi_{\tau}(y_i-\Gamma_i^\top \bh)\Psi_{\tau}(y_i-\Gamma_i^\top\bh)^\top \Gamma_i^\top\right\} \Vh^{\top} \label{eqn:cov1} \\
\intertext{or}
\widehat{\Sigma}_a &=\tau(1-\tau)\nu_2 \Vh\left\{\sum_j w_j\Gamma_j\Gamma_j^\top\right\}\Vh^\top.
\label{eqn:cov2}
\end{align}
\end{theorem}

Theorem \ref{thm:normality} tells us that the one 
step estimator is $\sqrt{nh}$-consistent. 
The covariance 
\eqref{eqn:cov2} is the expected version of \eqref{eqn:cov1},
where \eqref{eqn:cov1} comes from the central limit theorem. 
The estimators from the decorrelated score ($\ac^{DS}$) and 
reparameterization ($\ac^{RP}$) 
are both asymptotically equivalent to $\ac^{OS}$; 
the detailed proof is in the appendix.

%%% Local Variables:
%%% TeX-master: "QR_score.tex"
%%% End:

\section{Numerical studies}

Through an empirical study,
we investigate the finite sample performance 
of our confidence interval construction approach
and show that it works under high-dimensional 
settings and is robust to different error distributions.

For each individual, the data is generated 
independently and identically distributed from the 
following distribution of $\{U,X_1, X_{-1},\epsilon,Y\}$. 
First, we generate the index variables as
$U\sim \text{Unif} [0,2]$ and
the confounding variables as
\[
X_{-1} \mid U\sim \mathcal{N}(\mu(U),\Sigma(U)),
\]
where $\mu_j(U)=a_0 \cdot j \cdot (U^{a_1}-1)$ and 
$\Sigma(U)$ is an autoregressive (AR) covariance
with elements $\Sigma(U)_{i,j} = \rho(U)^{|i-j|}$ and
the parameter 
$\rho(U)=\rho^{1+b_0 (U^{b_1}-1)}$. Note that 
when $a_1=b_1=0$ we have 
a model where the nuisance covariates $X_{-1} \mid U$ are homogeneous
and do not depend on the index variable $U$.
We then generate $X_1$ and $Y$.
Let $\nu \in \R^{p-1}$ 
with 
$\nu_{j-1} = 1/j^2$, $j = 2,\cdots, p$,
and $\beta = (\frac{1}{2}, c_y \nu^{\top})^{\top}$. 
Then 
\begin{align*} X_1 &=X_{-1} (c_x \nu) + \epsilon_{x}, ~\text{where}~ \epsilon_x\sim \mathcal{N}(0,1) ~\text{and is independent 
of}~ (X_{-1},U),\\
Y &= X \beta(U) + \epsilon, \quad \text{where}\quad \beta(U) = \beta(c_0U^{c_1} + 1-c_0),\\
\text{and}~
& \epsilon \mid X,U\sim \sigma_e(U) \cdot F_e \cdot\sqrt{(2-\gamma+\gamma \cdot X_1^2)/2}.
\end{align*}
Note that $\epsilon$ is allowed to depend on $X$. 
In particular, $\gamma=0$ leads to a homogeneous setting 
and $\gamma=1$ leads to a heterogeneous setting. 
Here $\sigma_e(U)=\sigma_e(1+d_0 (U^{d_1}-1))$ 
and $F_e$ is either the standard Gaussian or 
$t$ distribution with 3 degrees of freedom ($t(3)$). 

Specially, this data generation process leads to the following quantiles:
\[ 
q(x;\tau,u)  = x\beta(u)+\sigma_e(u)\cdot\sqrt{(2-\gamma+\gamma \cdot x_1^2)/2}\cdot q_e(\tau),
\] 
where $q_e(\tau)$ is the $\tau$-th quantile of the distribution $F_e$.
In this simulation, we are interested in the inference for
$\beta_1(\tau,u) \in \R^1$ at the point $(\tau,u) = (0.5,1)$. 
At this point, $q(x;0.5,1)=x\beta(1)$ satisfy Assumption \ref{assumption: lam}. 
 
The coefficients $c_x$ and $c_y$ are used to control the $R^2$ in different regression equations. 
We use $R^2_y$ to denote the  $R^2$ in
the equation
$Y-X_1\beta_1(U) = X_{-1} \beta_{-1}(U) + \epsilon$,
while $R^2_x$ denotes the  $R^2$ in
the equation $X_1 = X_{-1} (c_x \nu) + \epsilon_{x}$.
We vary the parameters and choose $c_y$, $c_x$ to
form different combinations of $(R^2_y, R^2_x)$. 
Details can be found in Appendix \ref{app:sim}. 

We evaluate the performance of our algorithms described in 
Section \ref{algorithm} (DS \eqref{eqn:DS est}, OS \eqref{eqn:one step est} and RP \eqref{eqn:RP est}) and
compare them with the Oracle and the Naive methods. 
For the oracle method, we assume that the true (low dimensional) set of predictors is known 
in advance and our inference is based on the kernel weighted quantile
regression on the true set of variables. For the Naive method, 
we fit the kernel weighted penalized regression as in Step 1. 
Then we fit the post-regularized regression and do the inference treating 
the set $\hat S = \cbr{ j : \hat{\beta}_j\neq 0}$ as fixed. We compare their 
performance from $M=100$ simulations in terms of the bias, 
empirical standard deviation (SD), 
the expected estimated standard error  (ESE), and 
coverage rate for the 95\% nominal confidence intervals (CR).

\begin{table} 
\begin{center}
\begin{tabular}{c|cccccc}
\hline
$\epsilon$ distribution &$\gamma$&Method & Bias & SD & ESE & CR\\
\hline
\multirow{10}{*}{Normal}&\multirow{5}{*}{0}&One Step & 0.022 & 0.063 & 0.076 & 0.97\\
&&Decorrelated score& 0.057 & 0.139 & 0.124 & 0.92\\
&&Reparameterization & 0.045 & 0.055 & 0.076 & 0.97\\
&&Naive & 0.274 & 0.047 & 0.045 & 0.41\\
&&Oracle & -0.004 & 0.062 & 0.067 & 0.95\\
\cline{2-7}
&\multirow{5}{*}{1}&One Step & 0.093 & 0.149 & 0.190 & 0.96\\
& &Decorrelated score& 0.174 & 0.134 & 0.297& 0.99\\
& &Reparameterization & 0.148 & 0.142 & 0.190 & 0.97\\
& &Naive & 0.246& 0.120 & 0.150 & 0.62\\
& &Oracle & -0.034 & 0.133 & 0.140 & 0.94\\
\hline
\multirow{10}{*}{$t(3)$}&\multirow{5}{*}{0}&One Step & 0.022 & 0.083 & 0.084 & 0.96\\
& &Decorrelated score& 0.077 & 0.112 & 0.145& 0.95\\
& &Reparameterization & 0.054 & 0.071 & 0.084 & 0.93\\
& &Naive & 0.249 & 0.075 & 0.057 & 0.10\\
& &Oracle & -0.020 & 0.042 & 0.051 & 0.95\\
\cline{2-7}
&\multirow{5}{*}{1}&One Step & 0.090 & 0.193 & 0.212 & 0.95\\
& &Decorrelated score& 0.185 & 0.163 & 0.313& 0.98\\
& &Reparameterization & 0.167 & 0.173 & 0.212 & 0.94\\
& &Naive & 0.234& 0.156 & 0.186 & 0.73\\
& &Oracle & -0.01 & 0.126 & 0.164 & 0.98\\
\hline
\end{tabular}
\end{center}
\caption{Simulation results for the correlation setting $(R^2_x, R^2_y) = (0.7,0.3)$.}
\label{table1}
\end{table}

\begin{figure}\label{fig:sim-CR}
\includegraphics[scale=0.47]{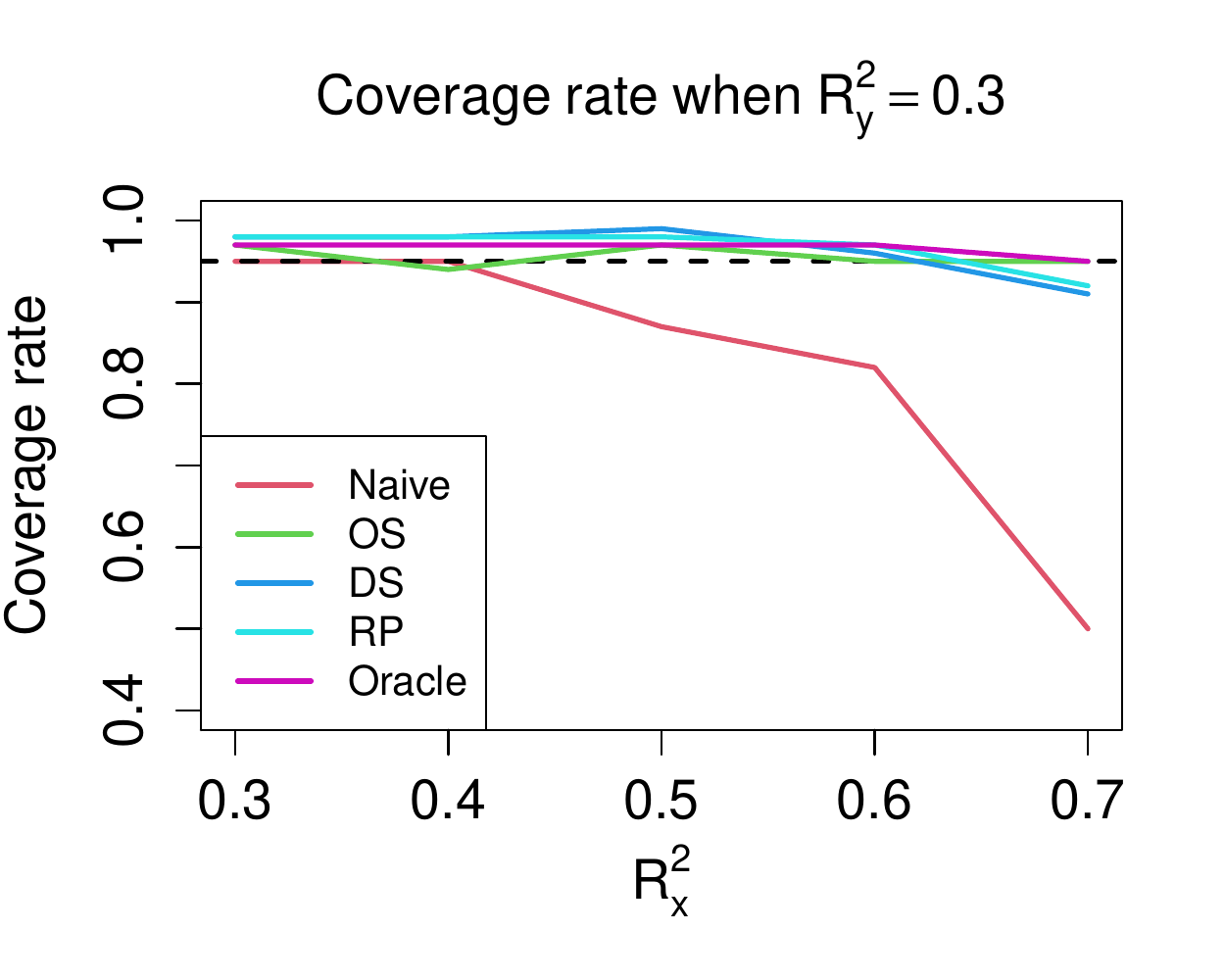}
\includegraphics[scale=0.47]{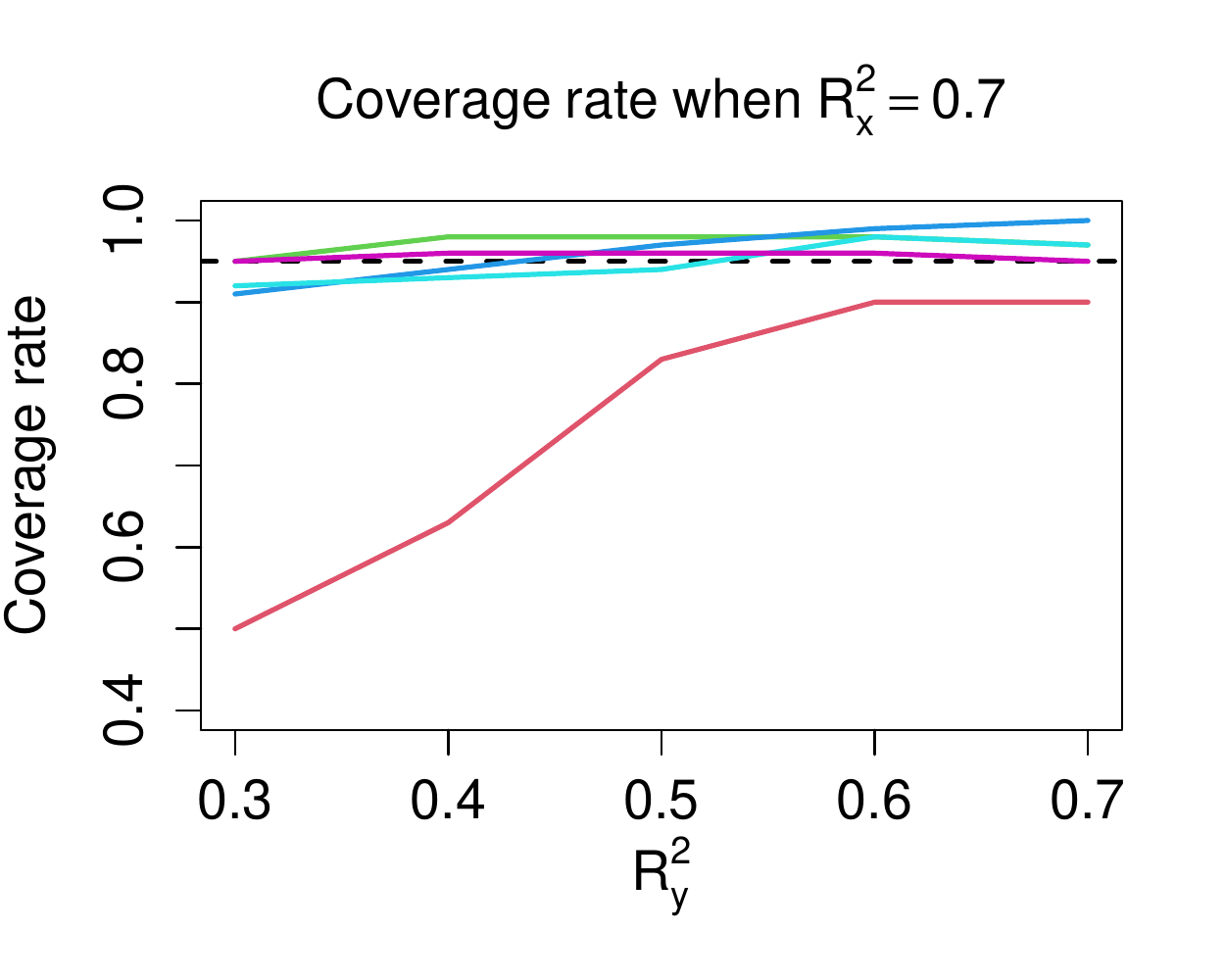}
\caption{Left: CR of different methods with $R^2_y$ fixed at 0.3 and 
changing $R^2_x$. Right: CR of different methods with $R^2_x$ fixed at 0.7 and 
changing $R^2_y$. In the simulation, the error term $\epsilon$ 
is normally distributed and $\gamma=0$. }
\label{fig:simulation}
\end{figure}

The simulation results for two settings
with normally distributed and $t(3)$-distributed $\epsilon$'s are listed in 
Table \ref{table1}. Additional simulation results are presented in 
Table \ref{table2} in Appendix \ref{app:sim}. 
From the simulation, the oracle method consistently 
produce confidence intervals with coverage rate close to
the nominal value $95\%$ in all simulation settings. 
The Naive estimator has some significant bias when $R^2_y$ 
is small and $R^2_x$ is large. Furthermore, without 
any correction, the confidence intervals tend to have significantly 
lower coverage than the nominal value.
The OS, DS, and PR estimators have relatively low bias 
compared to the naive method in all settings; also,
their coverage rates are closer to the nominal value than the naive method. 

We plot the trend of coverage rate
for all methods with the change of $R^2_y$ and $R^2_x$ in 
Figure \ref{fig:simulation} to better understand the performance
of the different methods as the data generating distribution changes.
We find that the confidence intervals from the
Naive method significantly undercover when the 
response $Y$ has low correlation with the covariates $X$ and
$X_1$ has high correlation with the confounding variables $X_{-1}$.
On the other hand, the proposed methods provide
satisfactory coverage across all settings. 

Regarding the widths of the confidence intervals, 
the Naive method underestimates the standard error 
in some data settings, resulting in low coverage rates. 
The OS, DS, and PR methods provide CI's with the correct coverage rate
and the widths of the CI's are slightly larger than those of the Oracle method. 
Among the three proposed methods, the OS method has the
best finite sample performance in terms of stability and computational cost.

%%% Local Variables:
%%% TeX-master: "QR_score.tex"
%%% End:

\section{Real data example} 

%\subsection*{Example 2}
\begin{figure}
\includegraphics[scale=0.46]{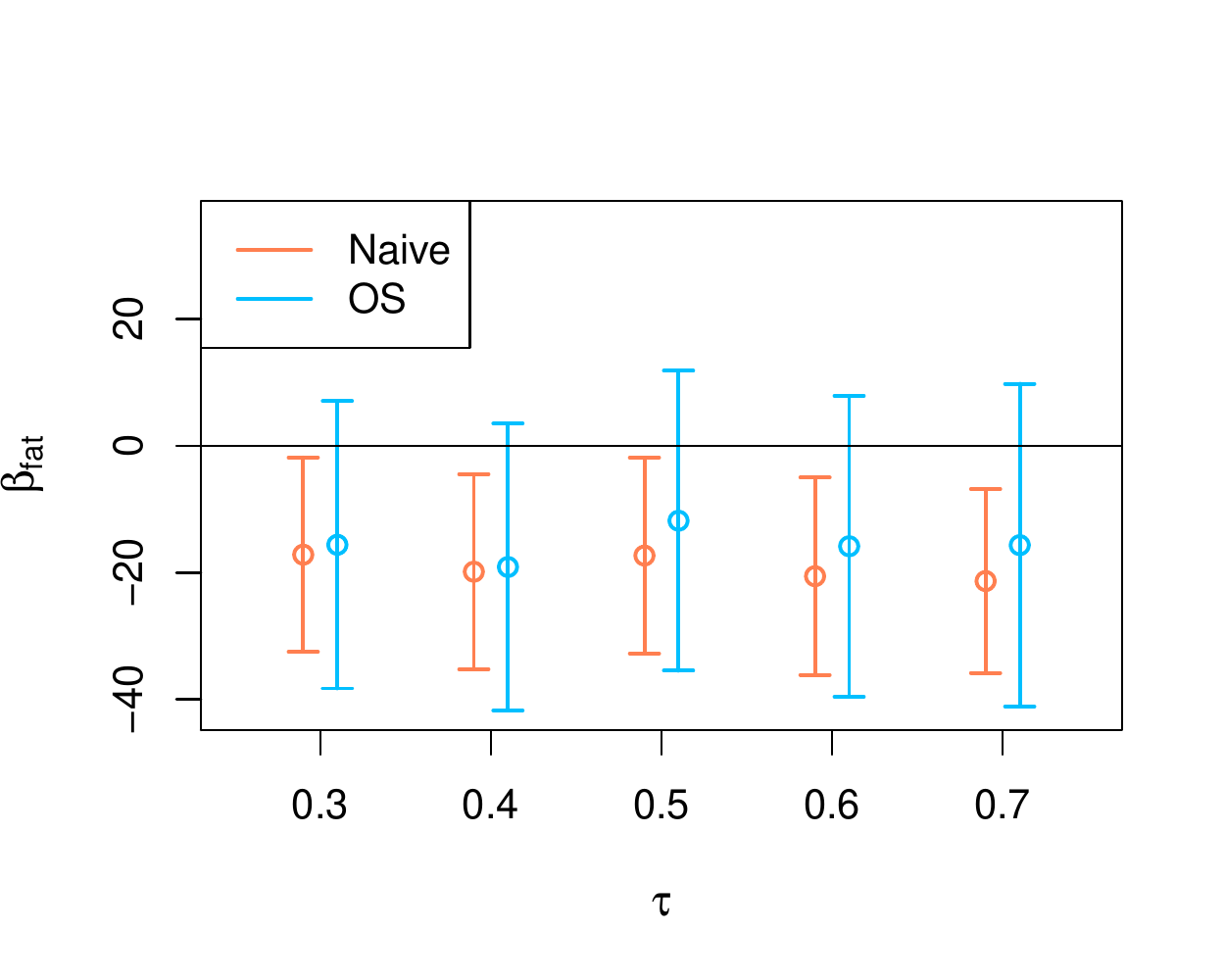}
\includegraphics[scale=0.46]{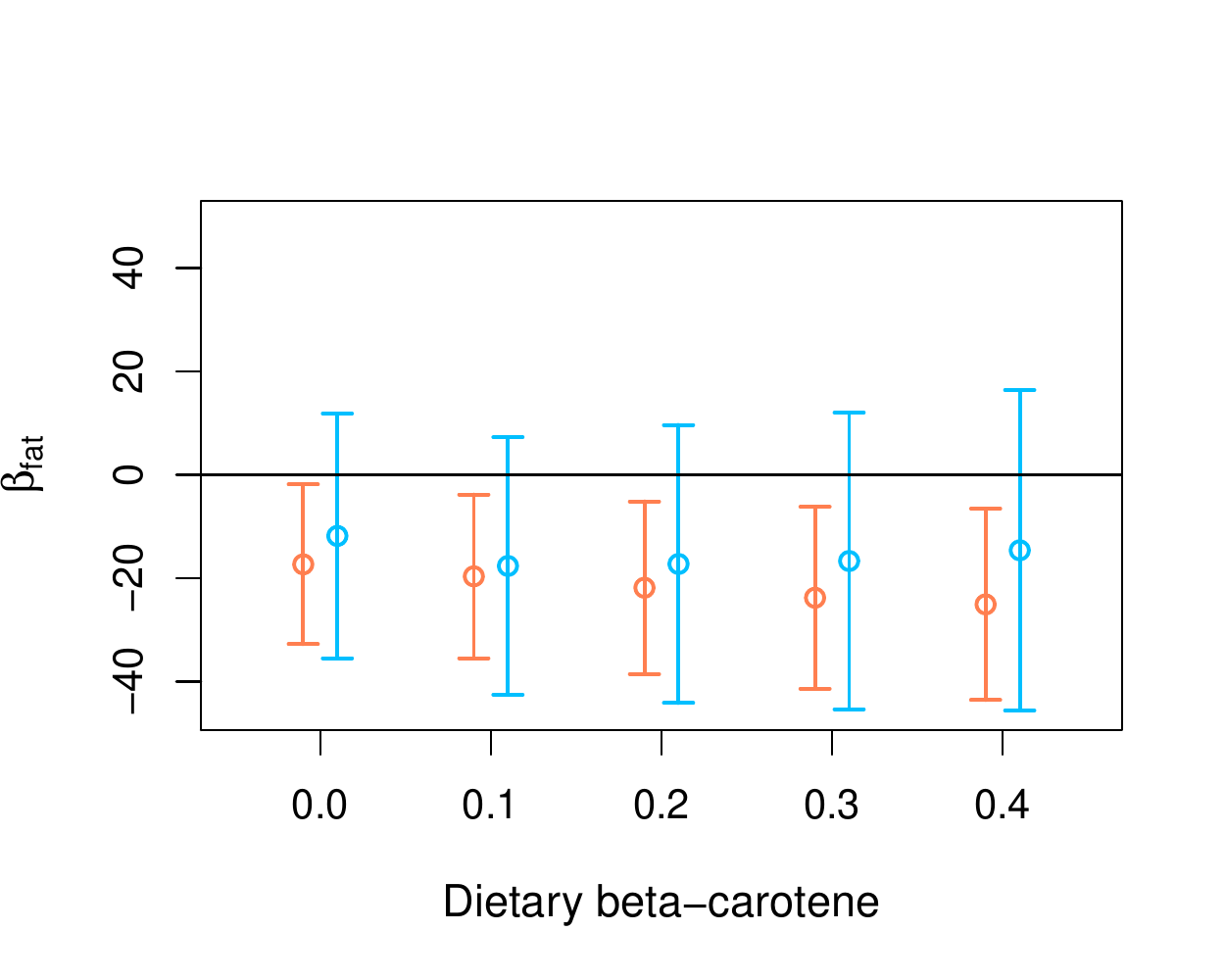}
\caption{Inference for fat. Left: fixing scaled dietary beta-carotene level at 0, $95\%$ confidence intervals for different $\tau$. Right: fixing $\tau =0.5$, $95\%$ confidence intervals for different beta-carotene levels. }
\label{fig:eg2-fat}
\end{figure}

\begin{figure}
\includegraphics[scale=0.46]{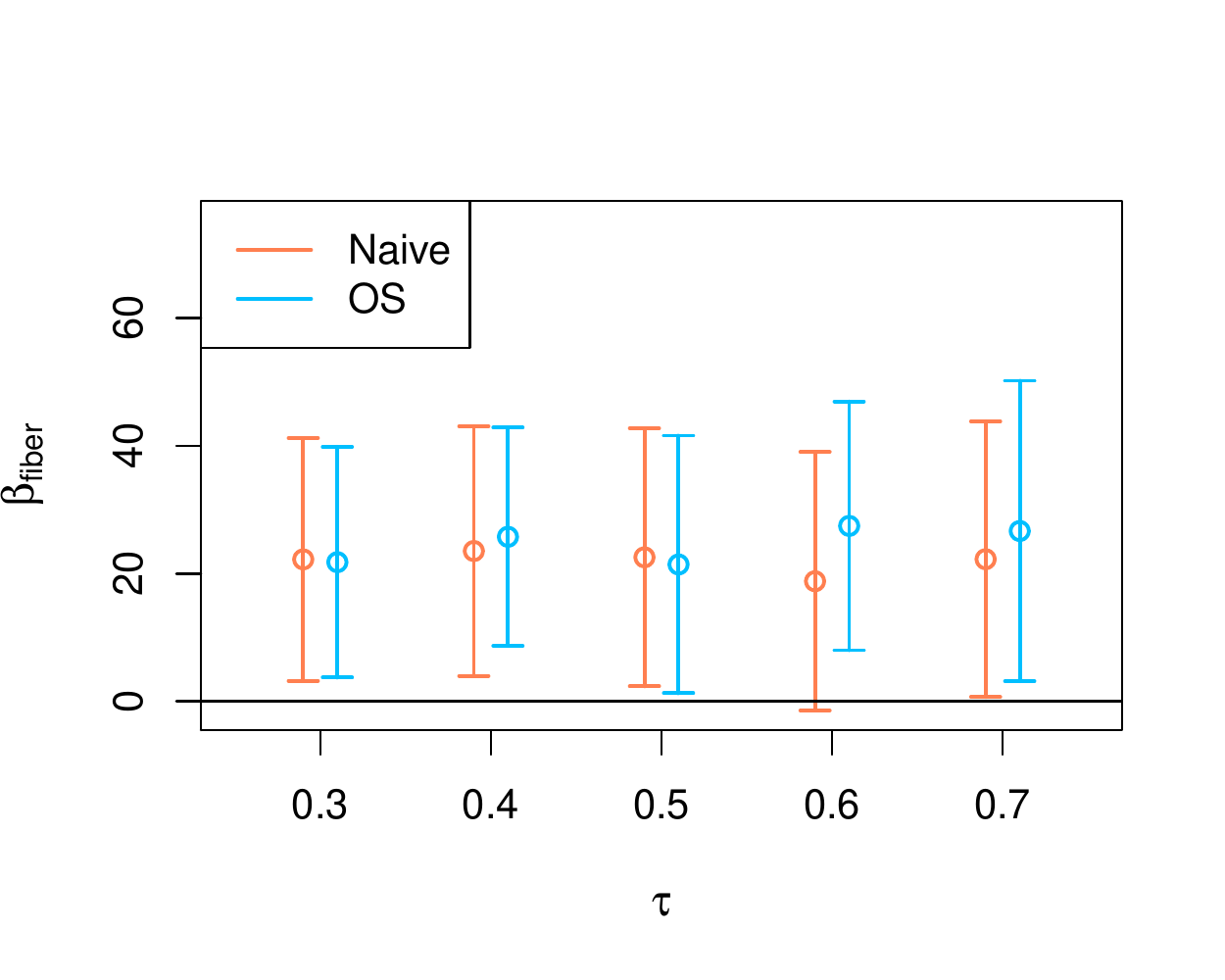}
\includegraphics[scale=0.46]{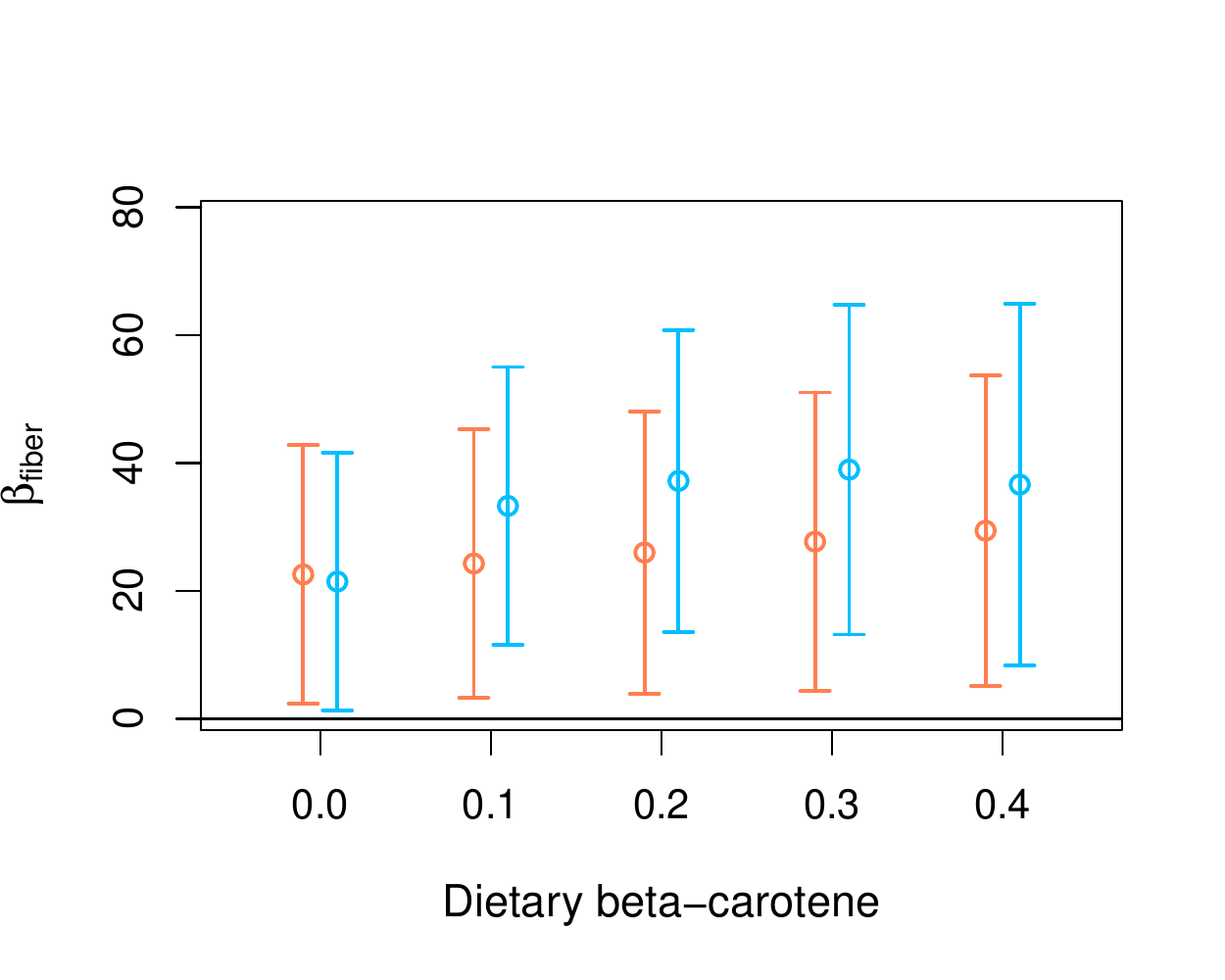}
\caption{Inference for fiber. Left: fixing scaled dietary beta-carotene level at 0, $95\%$ confidence intervals for different $\tau$. Right: fixing $\tau =0.5$, $95\%$ confidence intervals for different beta-carotene levels.}
\label{fig:eg2-fiber}
\end{figure}

As an illustration of our method, we apply our methods to analyze the 
plasma beta-carotene level data set collected by a cross-sectional study \citep{Nierenberg1989Determinants}. 
This dataset consists of 315 observations on 14 variables. Our interest is to study the
relationship between the plasma beta-carotene level and the following variables: age, sex, 
smoking status, quetelet (BMI), vitamin use, number of calories consumed per day, grams of
fiber consumed per day, number of alcoholic drinks consumed per week, cholesterol consumed
per day, dietary beta-carotene consumed per day and dietary retinol consumed per day. 

We fit our varying coefficient model by using dietary beta-carotene consumption
as the index $U$. We replace all categorical variables with dummy variables and 
standardize all variables. Then we include all two-way interactions in our model, 
so we have 116 confounding variables in total. We take the plasma beta-carotene level
as the outcome $Y$, the fat intake (in grams) or the fiber intake (in grams) as the treatment
effect $X_A$ respectively, and the remaining variables as the confounding variables. 
We use our model to make inference on $\beta(\tau, u)$ at different beta-carotene
consumption level $u$ and different quantiles $\tau$. 

Our results are shown in Figures \ref{fig:eg2-fat} and \ref{fig:eg2-fiber}. The Naive method 
is shown in red and we compare it with the one-step correction (OS) method. 
From Figure \ref{fig:eg2-fat}, the result of the naive method suggests that
the fat intake is significantly negatively correlated with the plasma beta-carotene level;
however, the OS method suggests that this negative effect is not significant. For fiber,
the Naive method underestimated the positive effect of fiber intake on the plasma beta-carotene level, 
whereas the OS method showed that this positive relationship is significant.
Furthermore, from Figure \ref{fig:eg2-fiber} (right plot), we can see an increasing trend 
of the effect of fiber intake with the increasing level of dietary beta-carotene.  

\section{Discussion}

We studied high-dimensional quantile regression
model with varying coefficients that allows us to
capture non-stationary effects of the input variables across time. 
Despite the importance in practical 
applications, no valid statistical inferential tools were previously available
for this problem. We addressed this issue by developing
new tools for statistical inference, allowing us to construct
valid confidence bands and honest tests for nonparametric coefficient
functions of time and quantile.
Performing statistical inference in this regime is
challenging due to the usage of model selection techniques in
estimation. Our inferential results do not rely on correct 
model selection and are valid for a range of data generating 
procedures, where one cannot expect for perfect model recovery.
The statistical framework allows us to construct a confidence interval at
a fixed point in time and a fixed quantile based on a normal
approximation, as well as a uniform confidence band for the
nonparametric coefficient function based on a Gaussian process
approximation. We perform numerical simulations to demonstrate the
finite sample performance of our method.
In addition, we also illustrate the performance of the methods through an 
application to a real data example.

\section*{Acknowledgment}

We thank Rina Foygel Barber for numerous suggestions and detailed
advice, as well as careful reading of various versions of the manuscript. 
This work is partially supported by the William S.~Fishman 
Faculty Research Fund at the University of Chicago
Booth School of Business. This work was completed in
part with resources supported by the University of Chicago
Research Computing Center.

%%% Local Variables:
%%% TeX-master: "QR_score.tex"
%%% End:
%%% Local Variables:
%%% TeX-master: "QR_score.tex"
%%% End:

\appendix

{
\bibliographystyle{my-plainnat}
%\bibpunct{(}{)}{,}{a}{,}{,}
\bibliography{paper}

\begin{thebibliography}{59}
\providecommand{\natexlab}[1]{#1}
\providecommand{\url}[1]{\texttt{#1}}
\expandafter\ifx\csname urlstyle\endcsname\relax
  \providecommand{\doi}[1]{doi: #1}\else
  \providecommand{\doi}{doi: \begingroup \urlstyle{rm}\Url}\fi

\bibitem[Barber and Kolar(2018)]{Barber2015ROCKET}
R.~F. Barber and M.~Kolar.
\newblock Rocket: Robust confidence intervals via kendall's tau for
  transelliptical graphical models.
\newblock \emph{Ann. Statist.}, 46\penalty0 (6B):\penalty0 3422--3450, 2018.

\bibitem[Belilovsky et~al.(2016)Belilovsky, Varoquaux, and
  Blaschko]{Belilovsky2016Testing}
E.~Belilovsky, G.~Varoquaux, and M.~B. Blaschko.
\newblock Testing for differences in gaussian graphical models: Applications to
  brain connectivity.
\newblock In D.~D. Lee, M.~Sugiyama, U.~V. Luxburg, I.~Guyon, and R.~Garnett,
  editors, \emph{Advances in Neural Information Processing Systems 29}, pages
  595--603. Curran Associates, Inc., 2016.

\bibitem[Belloni and Chernozhukov(2011)]{Belloni2011penalized}
A.~Belloni and V.~Chernozhukov.
\newblock $\ell_1$-penalized quantile regression in high-dimensional sparse
  models.
\newblock \emph{Ann. Stat.}, 39\penalty0 (1):\penalty0 82--130, 2011.

\bibitem[Belloni and Chernozhukov(2013)]{Belloni2013Least}
A.~Belloni and V.~Chernozhukov.
\newblock Least squares after model selection in high-dimensional sparse
  models.
\newblock \emph{Bernoulli}, 19\penalty0 (2):\penalty0 521--547, 2013.

\bibitem[Belloni et~al.(2013{\natexlab{a}})Belloni, Chernozhukov, and
  Hansen]{Belloni2012Inference}
A.~Belloni, V.~Chernozhukov, and C.~B. Hansen.
\newblock Inference on treatment effects after selection amongst
  high-dimensional controls.
\newblock \emph{Rev. Econ. Stud.}, 81\penalty0 (2):\penalty0 608--650,
  2013{\natexlab{a}}.

\bibitem[Belloni et~al.(2013{\natexlab{b}})Belloni, Chernozhukov, and
  Kato]{Belloni2013Robust}
A.~Belloni, V.~Chernozhukov, and K.~Kato.
\newblock Valid post-selection inference in high-dimensional approximately
  sparse quantile regression models.
\newblock \emph{arXiv preprint arXiv:1312.7186}, 2013{\natexlab{b}}, \href
  {http://arxiv.org/abs/1312.7186} {\ttfamily arXiv:1312.7186}.

\bibitem[Belloni et~al.(2015)Belloni, Chernozhukov, and
  Kato]{Belloni2013Uniform}
A.~Belloni, V.~Chernozhukov, and K.~Kato.
\newblock Uniform post-selection inference for least absolute deviation
  regression and other {Z}-estimation problems.
\newblock \emph{Biometrika}, 102\penalty0 (1):\penalty0 77--94, 2015.

\bibitem[Belloni et~al.(2016{\natexlab{a}})Belloni, Chen, and
  Chernozhukov]{Belloni2016Quantile}
A.~Belloni, M.~Chen, and V.~Chernozhukov.
\newblock Quantile graphical models: Prediction and conditional independence
  with applications to financial risk management.
\newblock \emph{ArXiv e-prints, arXiv:1607.00286}, 2016{\natexlab{a}}, \href
  {http://arxiv.org/abs/1607.00286} {\ttfamily arXiv:1607.00286}.

\bibitem[Belloni et~al.(2016{\natexlab{b}})Belloni, Chernozhukov, and
  Wei]{Belloni2013Honest}
A.~Belloni, V.~Chernozhukov, and Y.~Wei.
\newblock Post-selection inference for generalized linear models with many
  controls.
\newblock \emph{J. Bus. Econom. Statist.}, 34\penalty0 (4):\penalty0 606--619,
  2016{\natexlab{b}}.

\bibitem[Boucheron et~al.(2013)Boucheron, Lugosi, and
  Massart]{Boucheron2013Concentration}
S.~Boucheron, G.~Lugosi, and P.~Massart.
\newblock \emph{Concentration inequalities}.
\newblock Oxford University Press, Oxford, 2013.
\newblock A nonasymptotic theory of independence, With a foreword by Michel
  Ledoux.

\bibitem[Bradic and Kolar(2017)]{Bradic2017Uniform}
J.~Bradic and M.~Kolar.
\newblock Uniform inference for high-dimensional quantile regression: linear
  functionals and regression rank scores.
\newblock \emph{arXiv preprint arXiv:1702.06209}, 2017.

\bibitem[de~la Pe\~{n}a et~al.(2009)de~la Pe\~{n}a, Lai, and
  Shao]{Pena2009Self}
V.~H. de~la Pe\~{n}a, T.~L. Lai, and Q.-M. Shao.
\newblock \emph{Self-normalized processes}.
\newblock Probability and its Applications (New York). Springer-Verlag, Berlin,
  2009.
\newblock Limit theory and statistical applications.

\bibitem[Fan and Zhang(2000)]{Fan2000Simultaneous}
J.~Fan and W.~Zhang.
\newblock Simultaneous confidence bands and hypothesis testing in
  varying-coefficient models.
\newblock \emph{Scand. J. Stat.}, 27\penalty0 (4):\penalty0 715--731, 2000.

\bibitem[Farrell(2015)]{Farrell2013Robust}
M.~H. Farrell.
\newblock Robust inference on average treatment effects with possibly more
  covariates than observations.
\newblock \emph{Journal of Econometrics}, 189\penalty0 (1):\penalty0 1--23,
  2015.

\bibitem[Gin\'{e} and Guillou(2001)]{Gine2001consistency}
E.~Gin\'{e} and A.~Guillou.
\newblock On consistency of kernel density estimators for randomly censored
  data: rates holding uniformly over adaptive intervals.
\newblock \emph{Ann. Inst. H. Poincar\'{e} Probab. Statist.}, 37\penalty0
  (4):\penalty0 503--522, 2001.

\bibitem[Hastie and Tibshirani(1993)]{Hastie1993Varying}
T.~J. Hastie and R.~J. Tibshirani.
\newblock Varying-coefficient models.
\newblock \emph{J. R. Stat. Soc. B}, 55\penalty0 (4):\penalty0 757--796, 1993.

\bibitem[Hoover et~al.(1998)Hoover, Rice, Wu, and
  Yang]{Hoover1998Nonparametric}
D.~R. Hoover, J.~A. Rice, C.~O. Wu, and L.-P. Yang.
\newblock Nonparametric smoothing estimates of time-varying coefficient models
  with longitudinal data.
\newblock \emph{Biometrika}, 85\penalty0 (4):\penalty0 809--822, 1998.

\bibitem[Huang et~al.(2004)Huang, Wu, and Zhou]{Huang2004Polynomial}
J.~Z. Huang, C.~O. Wu, and L.~Zhou.
\newblock Polynomial spline estimation and inference for varying coefficient
  models with longitudinal data.
\newblock \emph{Stat. Sinica}, 14\penalty0 (3):\penalty0 763--788, 2004.

\bibitem[Jankov\'{a} and van~de Geer(2015)]{Jankova2014Confidence}
J.~Jankov\'{a} and S.~van~de Geer.
\newblock Confidence intervals for high-dimensional inverse covariance
  estimation.
\newblock \emph{Electron. J. Stat.}, 9\penalty0 (1):\penalty0 1205--1229, 2015.

\bibitem[Jankov\'a and van~de Geer(2017)]{Jankova2017Honest}
J.~Jankov\'a and S.~A. van~de Geer.
\newblock Honest confidence regions and optimality in high-dimensional
  precision matrix estimation.
\newblock \emph{TEST}, 26\penalty0 (1):\penalty0 143--162, 2017.

\bibitem[Javanmard and Montanari(2013{\natexlab{a}})]{Javanmard2013Hypothesis}
A.~Javanmard and A.~Montanari.
\newblock Hypothesis testing in high-dimensional regression under the gaussian
  random design model: Asymptotic theory.
\newblock \emph{arXiv preprint arXiv:1301.4240}, 2013{\natexlab{a}}.

\bibitem[Javanmard and Montanari(2013{\natexlab{b}})]{Javanmard2013Nearly}
A.~Javanmard and A.~Montanari.
\newblock Nearly optimal sample size in hypothesis testing for high-dimensional
  regression.
\newblock \emph{arXiv preprint arXiv:1311.0274}, 2013{\natexlab{b}}, \href
  {http://arxiv.org/abs/1311.0274} {\ttfamily arXiv:1311.0274}.

\bibitem[Kai et~al.(2011)Kai, Li, and Zhou]{Kai2011New}
B.~Kai, R.~Li, and H.~H. Zhou.
\newblock New efficient estimation and variable selection methods for
  semiparametric varying-coefficient partially linear models.
\newblock \emph{Ann. Stat.}, 39\penalty0 (1):\penalty0 305--332, 2011.

\bibitem[Kim et~al.()Kim, Liu, and Kolar]{Kim2019Two}
B.~Kim, S.~Liu, and M.~Kolar.
\newblock Two-sample inference for high-dimensional markov networks.
\newblock \emph{Journal of the Royal Statistical Society: Series B (Statistical
  Methodology)}, n/a\penalty0 (n/a), \href
  {http://arxiv.org/abs/https://rss.onlinelibrary.wiley.com/doi/pdf/10.1111/rssb.12446}
  {\ttfamily
  arXiv:https://rss.onlinelibrary.wiley.com/doi/pdf/10.1111/rssb.12446}.

\bibitem[Kim(2007)]{Kim2007Quantile}
M.-O. Kim.
\newblock Quantile regression with varying coefficients.
\newblock \emph{Ann. Statist.}, 35\penalty0 (1):\penalty0 92--108, 2007.

\bibitem[Koenker(1984)]{Koenker1984note}
R.~Koenker.
\newblock A note on {$L$}-estimates for linear models.
\newblock \emph{Statist. Probab. Lett.}, 2\penalty0 (6):\penalty0 323--325,
  1984.

\bibitem[Koenker(2005)]{Koenker2005Quantile}
R.~Koenker.
\newblock \emph{Quantile regression}, volume~38 of \emph{Econometric Society
  Monographs}.
\newblock Cambridge University Press, Cambridge, 2005.

\bibitem[Koltchinskii and Yuan(2010)]{Koltchinskii2010Sparsity}
V.~Koltchinskii and M.~Yuan.
\newblock Sparsity in multiple kernel learning.
\newblock \emph{Ann. Statist.}, 38\penalty0 (6):\penalty0 3660--3695, 2010.

\bibitem[Kozbur(2013)]{Kozbur2013Inference}
D.~Kozbur.
\newblock Inference in additively separable models with a high dimensional
  component.
\newblock Job Market Paper, 2013.

\bibitem[Lee et~al.(2013)Lee, Sun, Sun, and Taylor]{Lee2013Exact}
J.~D. Lee, D.~L. Sun, Y.~Sun, and J.~E. Taylor.
\newblock Exact post-selection inference with the lasso.
\newblock \emph{ArXiv e-prints, arXiv:1311.6238}, 2013, \href
  {http://arxiv.org/abs/1311.6238} {\ttfamily arXiv:1311.6238}.

\bibitem[Liu(2017)]{Liu2017Structural}
W.~Liu.
\newblock Structural similarity and difference testing on multiple sparse
  {G}aussian graphical models.
\newblock \emph{Ann. Statist.}, 45\penalty0 (6):\penalty0 2680--2707, 2017.

\bibitem[Lockhart et~al.(2014)Lockhart, Taylor, Tibshirani, and
  Tibshirani]{Lockhart2013significance}
R.~Lockhart, J.~E. Taylor, R.~J. Tibshirani, and R.~J. Tibshirani.
\newblock A significance test for the lasso.
\newblock \emph{Ann. Stat.}, 42\penalty0 (2):\penalty0 413--468, 2014.

\bibitem[Lu et~al.(2018)Lu, Kolar, and Liu]{Lu2015Posta}
J.~Lu, M.~Kolar, and H.~Liu.
\newblock Post-regularization inference for time-varying nonparanormal
  graphical models.
\newblock \emph{Journal of Machine Learning Research}, 18\penalty0
  (203):\penalty0 1--78, 2018.

\bibitem[Lu et~al.(2020)Lu, Kolar, and Liu]{Lu2019Kernel}
J.~Lu, M.~Kolar, and H.~Liu.
\newblock Kernel meets sieve: Post-regularization confidence bands for sparse
  additive model.
\newblock \emph{Journal of the American Statistical Association}, 115\penalty0
  (532):\penalty0 2084--2099, 2020.

\bibitem[Meinshausen(2015)]{Meinshausen2013Assumption}
N.~Meinshausen.
\newblock Group bound: confidence intervals for groups of variables in sparse
  high dimensional regression without assumptions on the design.
\newblock \emph{J. R. Stat. Soc. Ser. B. Stat. Methodol.}, 77\penalty0
  (5):\penalty0 923--945, 2015.

\bibitem[Na and Kolar(2021)]{Na2018High}
S.~Na and M.~Kolar.
\newblock High-dimensional index volatility models via stein's identity.
\newblock \emph{Bernoulli}, 27\penalty0 (2):\penalty0 794--817, 2021.

\bibitem[Na et~al.(2019)Na, Yang, Wang, and Kolar]{Na2019High}
S.~Na, Z.~Yang, Z.~Wang, and M.~Kolar.
\newblock High-dimensional varying index coefficient models via stein's
  identity.
\newblock \emph{Journal of Machine Learning Research}, 20\penalty0
  (152):\penalty0 1--44, 2019.

\bibitem[Negahban et~al.(2012)Negahban, Ravikumar, Wainwright, and
  Yu]{negahban2010unified}
S.~N. Negahban, P.~Ravikumar, M.~J. Wainwright, and B.~Yu.
\newblock A unified framework for high-dimensional analysis of $ m $-estimators
  with decomposable regularizers.
\newblock \emph{Stat. Sci.}, 27\penalty0 (4):\penalty0 538--557, 2012.

\bibitem[Nierenberg et~al.(1989)Nierenberg, Stukel, Baron, Dain, and
  Greenberg]{Nierenberg1989Determinants}
D.~W. Nierenberg, T.~A. Stukel, J.~A. Baron, B.~J. Dain, and E.~R. Greenberg.
\newblock Determinants of plasma levels of beta-carotene and retinol.
\newblock 130\penalty0 (3):\penalty0 511--521, 1989.

\bibitem[Nolan and Pollard(1987)]{Nolan1987Uprocess}
D.~Nolan and D.~Pollard.
\newblock {$U$}-processes: rates of convergence.
\newblock \emph{Ann. Statist.}, 15\penalty0 (2):\penalty0 780--799, 1987.

\bibitem[Ren et~al.(2015)Ren, Sun, Zhang, and Zhou]{Ren2013Asymptotic}
Z.~Ren, T.~Sun, C.-H. Zhang, and H.~H. Zhou.
\newblock Asymptotic normality and optimalities in estimation of large
  {G}aussian graphical models.
\newblock \emph{Ann. Stat.}, 43\penalty0 (3):\penalty0 991--1026, 2015.

\bibitem[Sun and Zhang(2013)]{Sun2012Sparse}
T.~Sun and C.-H. Zhang.
\newblock Sparse matrix inversion with scaled lasso.
\newblock \emph{J. Mach. Learn. Res.}, 14:\penalty0 3385--3418, 2013.

\bibitem[Tang et~al.(2013)Tang, Song, Wang, and Zhu]{Tang2013Variable}
Y.~Tang, X.~Song, H.~J. Wang, and Z.~Zhu.
\newblock Variable selection in high-dimensional quantile varying coefficient
  models.
\newblock \emph{J. Multivariate Anal.}, 122:\penalty0 115--132, 2013.

\bibitem[Taylor et~al.(2014)Taylor, Lockhart, Tibshirani, and
  Tibshirani]{Taylor2014Post}
J.~E. Taylor, R.~Lockhart, R.~J. Tibshirani, and R.~J. Tibshirani.
\newblock Post-selection adaptive inference for least angle regression and the
  lasso.
\newblock \emph{arXiv preprint arXiv:1401.3889}, 2014, \href
  {http://arxiv.org/abs/1401.3889} {\ttfamily arXiv:1401.3889}.

\bibitem[Tibshirani(1996)]{tibshirani96regression}
R.~J. Tibshirani.
\newblock Regression shrinkage and selection via the lasso.
\newblock \emph{J. R. Stat. Soc. B}, 58\penalty0 (1):\penalty0 267--288, 1996.

\bibitem[van~de Geer and B\"{u}hlmann(2013)]{Geer2013penalized}
S.~A. van~de Geer and P.~B\"{u}hlmann.
\newblock {$\ell\sb 0$}-penalized maximum likelihood for sparse directed
  acyclic graphs.
\newblock \emph{Ann. Stat.}, 41\penalty0 (2):\penalty0 536--567, 2013.

\bibitem[van~de Geer et~al.(2014)van~de Geer, B\"{u}hlmann, Ritov, and
  Dezeure]{Geer2013asymptotically}
S.~A. van~de Geer, P.~B\"{u}hlmann, Y.~Ritov, and R.~Dezeure.
\newblock On asymptotically optimal confidence regions and tests for
  high-dimensional models.
\newblock \emph{Ann. Stat.}, 42\penalty0 (3):\penalty0 1166--1202, 2014.

\bibitem[{van der Vaart} and Wellner(1996)]{vanderVaart1996Weak}
A.~W. {van der Vaart} and J.~A. Wellner.
\newblock \emph{Weak Convergence and Empirical Processes: With Applications to
  Statistics}.
\newblock Springer, 1996.

\bibitem[Wang et~al.(2009)Wang, Zhu, and Zhou]{Wang2009Quantile}
H.~J. Wang, Z.~Zhu, and J.~Zhou.
\newblock Quantile regression in partially linear varying coefficient models.
\newblock \emph{Ann. Statist.}, 37\penalty0 (6B):\penalty0 3841--3866, 2009.

\bibitem[Wang and Kolar(2014)]{Wang2014Inference}
J.~Wang and M.~Kolar.
\newblock Inference for sparse conditional precision matrices.
\newblock \emph{ArXiv e-prints, arXiv:1412.7638}, 2014, \href
  {http://arxiv.org/abs/1412.7638} {\ttfamily arXiv:1412.7638}.

\bibitem[Wang and Kolar(2016)]{Wang2016Inference}
J.~Wang and M.~Kolar.
\newblock Inference for high-dimensional exponential family graphical models.
\newblock In A.~Gretton and C.~C. Robert, editors, \emph{Proceedings of the
  19th International Conference on Artificial Intelligence and Statistics},
  volume~51 of \emph{Proceedings of Machine Learning Research}, pages
  1042--1050, Cadiz, Spain, 2016. PMLR.

\bibitem[Wang et~al.(2020)Wang, Kolar, and Shojaie]{Wang2020Statistical}
X.~Wang, M.~Kolar, and A.~Shojaie.
\newblock Statistical inference for networks of high-dimensional point
  processes.
\newblock \emph{arXiv:2007.07448}, 2020, \href
  {http://arxiv.org/abs/2007.07448v1} {\ttfamily arXiv:2007.07448v1}.

\bibitem[Xia et~al.(2015)Xia, Cai, and Cai]{Xia2015Testing}
Y.~Xia, T.~Cai, and T.~T. Cai.
\newblock Testing differential networks with applications to the detection of
  gene-gene interactions.
\newblock \emph{Biometrika}, 102\penalty0 (2):\penalty0 247--266, 2015.

\bibitem[Yu et~al.(2016)Yu, Gupta, and Kolar]{Yu2016Statistical}
M.~Yu, V.~Gupta, and M.~Kolar.
\newblock Statistical inference for pairwise graphical models using score
  matching.
\newblock In \emph{Advances in Neural Information Processing Systems 29}.
  Curran Associates, Inc., 2016.

\bibitem[Yu et~al.(2020{\natexlab{a}})Yu, Gupta, and Kolar]{Yu2019Constrained}
M.~Yu, V.~Gupta, and M.~Kolar.
\newblock Constrained high dimensional statistical inference.
\newblock \emph{arXiv:1911.07319}, 2020{\natexlab{a}}, \href
  {http://arxiv.org/abs/1911.07319v1} {\ttfamily arXiv:1911.07319v1}.

\bibitem[Yu et~al.(2020{\natexlab{b}})Yu, Gupta, and Kolar]{Yu2019Simultaneous}
M.~Yu, V.~Gupta, and M.~Kolar.
\newblock Simultaneous inference for pairwise graphical models with generalized
  score matching.
\newblock \emph{Journal of Machine Learning Research}, 21\penalty0
  (91):\penalty0 1--51, 2020{\natexlab{b}}.

\bibitem[Zhang and Zhang(2013)]{Zhang2011Confidence}
C.-H. Zhang and S.~S. Zhang.
\newblock Confidence intervals for low dimensional parameters in high
  dimensional linear models.
\newblock \emph{J. R. Stat. Soc. B}, 76\penalty0 (1):\penalty0 217--242, 2013.

\bibitem[Zhang et~al.(2002)Zhang, Lee, and Song]{Zhang2002Local}
W.~Zhang, S.-Y. Lee, and X.~Song.
\newblock Local polynomial fitting in semivarying coefficient model.
\newblock \emph{J. Multivariate Anal.}, 82\penalty0 (1):\penalty0 166--188,
  2002.

\bibitem[Zhao et~al.(2014)Zhao, Kolar, and Liu]{Zhao2014General}
T.~Zhao, M.~Kolar, and H.~Liu.
\newblock A general framework for robust testing and confidence regions in
  high-dimensional quantile regression.
\newblock \emph{ArXiv e-prints, arXiv:1412.8724}, 2014, \href
  {http://arxiv.org/abs/1412.8724} {\ttfamily arXiv:1412.8724}.

\end{thebibliography}
}
\newpage

\section{Technical details}

\subsection{Notations}

We summarize the additional notation used throughout the appendix.
We 
let $\Psi_{\tau}(u) = \tau - \One{u < 0}$, 
$\rho_{\tau}(u) = u\Psi_{\tau}(u)$, 
and use
\begin{align*}
   W_i(\delta) &= \rho_\tau\rbr{y_i - (\fq + \delta)} - \rho_\tau\rbr{y_i - \fq}; \\ W_i^{\#}(\delta) &=-\delta \Psi_\tau(y_i - \fq); ~\text{and }\\
   W_i^{\natural}(\delta) &= \int_0^{\delta}\sbr{\One{y_i\leq \fq + z} - \One{y_i\leq \fq}} dz \\
   &= (y_i-(\fq+\delta))\sbr{\One{\fq + \delta \leq y_i < \fq} - \One{\fq \leq y_i < \fq + \delta}}.
\end{align*}
We denote the sum as $\EEn{\cdot} = \sum_{i \in [n]}\cdot$
and $\bE{\cdot} = \EE{\EEn{\cdot}}$. 
Denote the negative Hessian 
\begin{align*}
H &= \EEn{w_i  f_i(\fq) \cdot \Gamma_i\Gamma_i^\top}, \\
H^\star &=\EEn{w_i  f_i(\fql) \cdot \Gamma_i\Gamma_i^\top},\\
H(\delta) &= \bE{w_i \hat f_i(\delta) \cdot \Gamma_i\Gamma_i^\top}  = \EEn{w_i \cdot \EE{\hat{f}_i(\delta)} \cdot \Gamma_i\Gamma_i^\top}, ~\text{and}\\
\hat H(\delta) &= \EEn{w_i \hat f_i(\delta) \cdot \Gamma_i\Gamma_i^\top}, 
\end{align*}
where
\[
\hat f_i(\delta) = \frac{\One{\abr{y_i - \Gamma_i^\top\rbr{\bs + \delta}}\leq h_f}}{2h_f}.
\]
Let $\Delta_i =\tilde{q}_i - q_i$.
Recall that $\Vs \in \R^{2k\times 2p}$ are the rows related to $X_A$,
$X_{A}(U-u)$ of an approximate inverse of $H^\star$ such 
that $\norm{H^\star \Vs - E_a}_{\infty,F} \leq \lambda^{\star}$ and $\norm{\Vs}_{F,0} \leq s_2 = c_2s.$ 
Its estimator $\Vh$ is as defined in \eqref{eq:lasso_fixed:opt}. 
The one step correction estimator
$\check{a}^{OS}=\ah-S(\ah,\ch,\Vh)$, where
 $S(a,c,V)=-\sum_i w_i V \Gamma_i\Psi_{\tau}(y_i-\Gamma_i^\top (a^\top,c^\top)^\top)$.

\subsection{Proof of Theorem \ref{thm:normality}}\label{pf:normal}

Recall the definitions of $\check{a}^{OS}$ in \eqref{eqn:one step est} and $S_d(b,V)$ in \eqref{eqn:S_d}.
We have
\begin{multline*}
\check{a}^{OS}-\as = \ah - \as -  \left\{S_d((\ah^\top,\ch^\top)^\top,\Vs) - S_d((a^{\star\top},c^{\star\top})^\top,\Vs)\right\}\\ 
-\left\{S_d((\ah^\top,\ch^\top)^\top,\Vh) - S_d((\ah^\top,\ch^\top)^\top,\Vs)\right\} 
- S_d((a^{\star\top},c^{\star\top})^\top,\Vs).
\end{multline*}
By Lemma \ref{lem:os1} and Lemma \ref{lem:os2} 
(presented later in Section \ref{pf:norm-lem}), we have
\[\Norm{\ah - \as -  \left\{S_d((\ah^\top,\ch^\top)^\top,\Vs) - S_d((a^{\star\top},c^{\star\top})^\top,\Vs)\right\}}_2= o_p\left(\sqrt{\frac{1}{nh}}\right)\]
and 
\[\norm{S_d((\ah^\top,\ch^\top)^\top,\Vh) - S_d((\ah^\top,\ch^\top)^\top,\Vs)}_2 = o_p\left(\sqrt{\frac{1}{nh}}\right). \]
Therefore,
\begin{align*}
\check{a}^{OS}-\as &= - S_d((a^{\star\top},c^{\star\top}),\Vs)+o_p\left(\sqrt{\frac{1}{nh}}\right)\\
&=\sum_i w_i \Vs \Gamma_i\Psi_{\tau}(y_i-q_i)\\
&\qquad-\sum_i w_i \Vs \Gamma_i\left\{\Psi_{\tau}(y_i-q_i)-\Psi_{\tau}(y_i-\fql)\right\}+o_p\left(\sqrt{\frac{1}{nh}}\right)\\
&=\sum_i w_i \Vs \Gamma_i\Psi_{\tau}(y_i-q_i)+o_p\left(\sqrt{\frac{1}{nh}}\right).
\end{align*}
The last equality holds by Lemma \ref{lem:os3}, which we present later in Section \ref{pf:norm-lem}.

Because Assumptions \ref{assumption:density} and \ref{assumption:X} hold, by Lindeberg CLT, 
we have 
\[
\sqrt{nh}\sum_i w_i \Vs \Gamma_i\Psi_{\tau}(y_i-q_i)\rightarrow N(0,\Sigma)
\]
and, therefore,
\[
\sqrt{nh}\Sigma^{-1/2}\left(\check{a}^{OS}-\as\right)\xrightarrow{d} N(0, I_{2k}),
\]
where 
$\Sigma=\tau(1-\tau)\nu_2 \lim_{n\rightarrow \infty}\EEst{\left\{V^{\star}\Gamma\Gamma^\top V^{\star \top}\right\}}{U=u}$.
By Lemma \ref{lem:os4} (presented later in Section \ref{pf:norm-lem}), for both forms of $\widehat{\Sigma}_a$, we have \[\widehat{\Sigma}_a \xrightarrow{p} \Sigma.\] Therefore by Slutsky's Theorem,
\[
\sqrt{nh}\widehat{\Sigma}_a^{-1/2}\left(\check{a}^{OS}-\as\right)\xrightarrow{d} N(0, I_{2k}).
\]

\subsection{Asymptotic equivalence of decorrelated score, one-step and reparameterization estimators}

Denote
\begin{align*}
H_{ac}(f,V)&=\sum_i w_i f_i V \Gamma_i\Gamma_i^\top\left[\begin{array}{cc}0_{2(p-k)\times 2k} &I_{2(p-k)} \end{array}\right]^\top, \text{ and}\\
H_{aa}(f,V)&=\sum_i w_i f_i V \Gamma_i\Gamma_i^\top\left[\begin{array}{cc}I_{2k}& 0_{2k\times 2(p-k)} \end{array}\right]^\top.
\end{align*}
The decorrelated score estimator $\check{a}^{DS}$ in \eqref{eqn:DS est} that minimizes 
\begin{equation*}
S_d((a^\top,\ch^\top)^\top,\Vh)^\top 
\Psi^{-1}
S_d((a^\top,\ch^\top)^\top,\Vh),
\end{equation*}
where 
\[
\Psi = \sum_i w_i^2 \Vh \Gamma_i\Psi_{\tau}(y_i-\Gamma_i^\top \bh)\Psi_{\tau}(y_i-\Gamma_i^\top \bh)^\top \Gamma_i^\top \Vh^{\top},
\]
is asymptotically equivalent to the one-step estimator. 

To show this, 
given the optimization range 
\[A_{\tau}=\{a:\norm{a-a^{\star}}_2<\frac{C}{\log n}\},\]
we have $$\norm{S_d((\check{a}^{DS\top},\ch^{\top})^\top,\Vh)}_2=o_p\rbr{\sqrt{\frac{1}{nh}}}$$
and
\begin{align}
S_d&((\check{a}^{DS\top},\ch^\top)^\top,\Vh) \nonumber\\
&=S_d((a^{\star\top},c^{\star\top})^\top,\Vs)+\underbrace{\left( S_d((a^{\star\top},c^{\star\top})^\top,\Vh) - S_d((a^{\star\top},c^{\star\top})^\top,\Vs) \right)}_{=o_p\left(\sqrt{\frac{1}{nh}}\right)} \nonumber\\
&\quad +\underbrace{\left( S_d((a^{\star\top},\ch^{\top})^\top,\Vh) - S_d((a^{\star\top},c^{\star\top})^\top,\Vh) \right)}_{=o_p\left(\sqrt{\frac{1}{nh}}\right)} \nonumber\\
&\quad +\left( S_d((\check{a}^{DS\top},\ch^\top)^\top,\Vh) - S_d((a^{\star\top},\ch^{\top})^\top,\Vh) \right)\nonumber\\
&= S_d((a^{\star\top},c^{\star\top})^\top,\Vs)+H_{aa}(f,V^{\star})(\check{a}^{DS}-a^{\star})\nonumber\\
&\quad +o(\norm{\check{a}^{DS}-a^{\star}}_2)+o_p\left(\sqrt{\frac{1}{nh}}\right).
\end{align}
Therefore, we have 
\[
\check{a}^{DS}-a^{\star}=-H_{aa}(f,V^{\star})^{-1}S_d((a^{\star\top},c^{\star\top})^\top,V^{\star})+o_p\left(\sqrt{\frac{1}{nh}}\right),
\]
which is asymptotic equivalent to the one-step estimator.

For the reparameterization estimator in \eqref{eqn:RP est}, 
we need to assume that $\Vh$ can be decomposed as
$\Vh=\Vh_{11}\left[\begin{array}{cc} I_{2k} &-\widehat{v}\end{array}\right]$ 
where $\Vh_{11}$ is invertible with high probability,
and $\widehat{v}=-\Vh_{11}^{-1}\Vh_{12}$. Similarly we have
$V^{\star}=V_{11}^{\star}\left[\begin{array}{cc} I_{2k} &-v^{\star}\end{array}\right]$ where $v^{\star}=-[V_{11}^{\star}]^{-1}V^{\star}_{12}$. 
Let
\begin{align*}
s(a,c,v)&=\sum_i w_i \left[\begin{array}{cc}I_{2k}& -v \end{array}\right] \Gamma_i\Psi_{\tau}(y_i-\Gamma_i^\top (a,c)),\\
h_{ac}(f,v)&=\sum_i w_i f_i \left[\begin{array}{cc}I_{2k}& -v \end{array}\right]\Gamma_i\Gamma_i^\top\left[\begin{array}{cc}0_{2(p-k)\times 2k} &I_{2(p-k)} \end{array}\right]^\top, 
\text{ and}\\
h_{aa}(f,v)&=\sum_i w_i f_i \left[\begin{array}{cc}I_{2k}& -v \end{array}\right]\Gamma_i\Gamma_i^\top\left[\begin{array}{cc}I_{2k}& 0_{2k\times 2(p-k)} \end{array}\right]^\top.
\end{align*}
We have $\check{a}^{RP}$ minimizing $\mathcal{L}(a,\ch+\widehat{v}^\top(\ah-a))$ as defined
in \eqref{eqn:solve_decorr}. The optimization of the low-dimension 
quantile regression will approximately solve the following score 
\begin{multline*}
s(a,\ch+\widehat{v}^\top(\ah-a),\widehat{v})\\
=s(a,\ch,\widehat{v})+h_{ac}(\hat{f},\widehat{v})\widehat{v}^\top(\ah-a)+o(\norm{\widehat{v}^\top(\ah-a)}_2) + o_p\left(\sqrt{\frac{1}{nh}}\right)
\end{multline*}
in the sense that
\[s(\check{a}^{RP},\ch+\widehat{v}^\top(\ah-\check{a}),\Vh)=o_p\rbr{\sqrt{\frac{1}{nh}}}.\]
Since
\[
h_{ac}(\hat{f},\widehat{v})\widehat{v}^\top(\ah-a)=o_p\rbr{\sqrt{\frac{1}{nh}}},
\]
the equivalence of $\widehat{v}$ and the
lasso estimator from a regression of $\Gamma_{1:2k}$ on $\Gamma_{(2k+1):2p}$
implies that a similar expansion as decorrelated score \eqref{eqn:S_d} holds.
Therefore, we have 
\begin{multline*}
s(\check{a}^{RP},\ch,\widehat{v})\\
=s(a^{\star},c^{\star},v^{\star})+h_{aa}(f,v^{\star})(\check{a}^{RP}-a^{\star})+o(\norm{\check{a}^{RP}-a^{\star}}_2)+o_p\rbr{\sqrt{\frac{1}{nh}}}
\end{multline*}
and
$\check{a}^{RP}-a^{\star}$ is asymptotically equivalent to 
$[h_{aa}(f,v^{\star})]^{-1}s(a^{\star},c^{\star},v^{\star})$, 
which converges to a normal distribution.

\subsection{Lemmas for the normality results}\label{pf:norm-lem}
\begin{lemma}
\label{lem:os1}
Suppose that
Assumptions \ref{assumption:kernel}-\ref{assumption:X}
and the conditions \eqref{eqn:thm_b_1}-\eqref{eqn:thm_v_2} hold.
Then
\begin{equation*}
   \norm{\ah - \as -  \left\{S_d((\ah^\top,\ch^\top)^\top,\Vs) - S_d((a^{\star\top},c^{\star\top})^\top,\Vs)\right\}}_2= o_p\left(\sqrt{\frac{1}{nh}}\right). 
\end{equation*}
\end{lemma}

\begin{proof}
From the definition of
$S_d(b,V)$ in \eqref{eqn:S_d}, 
we can rewrite the objective as
\begin{align*}
\ah &- \as -\left\{S_d((\ah^\top,\ch^\top)^\top,\Vs) - S_d((a^{\star\top},c^{\star\top})^\top,\Vs)\right\} \\
 &= \ah - \as+\sum_i w_i \Vs \Gamma_i\Psi_{\tau}(y_i-\Gamma_i^\top (\ah^\top,\ch^\top)^\top) \\
 & \  - \sum_i w_i \Vs \Gamma_i\Psi_{\tau}(y_i-\Gamma_i^\top (\as^\top,\cs^\top)^\top)\\
 &= \ah - \as+\sum_i w_i \Vs \Gamma_i \left[\One{y_i\leq \Gamma_i \bs}- \One{y_i\leq \Gamma_i \bh} \right]\\
 &= \underbrace{\ah - \as+\sum_i w_i \Vs \Gamma_i \left[F_i(y_i\leq \Gamma_i \bs)- F_i(y_i\leq \Gamma_i \bh) \right] }_{I}\\
 &\  +\underbrace{ \sum_i w_i \Vs \Gamma_i \left[\left(\One{y_i\leq \Gamma_i \bs}-F_i(y_i\leq \Gamma_i \bs)\right) -
\left(\One{y_i\leq \Gamma_i \bh}-F_i(y_i\leq \Gamma_i \bh)\right)\right]}_{II}.
\end{align*}
For the term $I$, we use Taylor expansion and have
\begin{align*}
I &=\ah - \as+\sum_i w_i \Vs \Gamma_i \rbr{F_i(y_i\leq \Gamma_i \bs)- F_i(y_i\leq \Gamma_i \bh) } \\
&=\ah - \as+ \sum_i w_i \Vs \Gamma_i \rbr{ f_i(\Gamma_i \bs) \Gamma_i (\bs-\bh) +\frac {f_i'(\Gamma_i \bt_i)}{2}(\bs - \bh)^{\top}\Gamma_i^{\top} \Gamma_i (\bs - \bh) }\\
&=\ah - \as+ \Vs H^{\star} (\bs-\bh) + R_1 \\
&=\ah - \as+ E_a(\bs-\bh) + (\Vs H^{\star} -E_a) (\bs-\bh) + R_1\\
&= (\Vs H^{\star} -E_a) (\bs-\bh) + R_1,  
\end{align*}
where $\bt_i = t_i\bs + (1-t_i)\bh$
and 
\[
R_1=\sum_i w_i \Vs \Gamma_i\frac {f_i'(\Gamma_i \bt_i)}{2}(\bs - \bh)^{\top}\Gamma_i^{\top} \Gamma_i (\bs - \bh).
\]
By \eqref{eqn:thm_b_1}, 
\[\norm{R_1}_2\leq 2k B_V\bar{f'}\sum_i w_i |\Gamma_i (\bs - \bh)|^2 = O_p\left(\frac{sB_V\log (np)}{nh}\right).\]
We also have
\begin{multline*}
\norm{(\Vs H^{\star}-E_a) (\bs-\bh)}_2\\
\leq \norm{(\Vs H^{\star} -E_a)}_{\infty,F} \norm{\bs-\bh}_{1,2} 
=  O_p\rbr{\lambda^{\star} \cdot s\sqrt{\frac{\log (np)}{nh}}},
\end{multline*}
where the norm $\norm{\cdot}_{\infty,F}$ is defined as 
\[
\norm{V}_{\infty, F} = \sup_{i \in [k],j \in [p]} \norm{V_{(i,i+k),(j,j+p)}}_F
\]
and the second inequality is because of 
Assumption \ref{assumption: lam}  and \eqref{eqn:thm_b_1}. 
Furthermore, since $\lambda^{\star} = O\rbr{ B_V \sqrt{ \frac{\log p}{nhh_f} }}$ 
by Assumption \ref{assumption: lam}, we have 
\[
\norm{I}_2 =  O_p\rbr{\lambda^{\star}\cdot s\sqrt{\frac{\log (np)}{nh}}} = O_p\rbr{\frac{B_Vs\log (np)}{nh\sqrt{h_f}}} = o_p\rbr{\sqrt{\frac{1}{nh}}}.
\]

For the term $II$, by Lemma \ref{lem:os11} (presented later in Section \ref{pf:norm-lem}) we have
$$\norm{II}_2 = O_p\rbr{ B_K B_V \sqrt{\bar{f}B_X\frac{mr_b \log p}{nh}} }.$$ 
Plugging the rate for $r_b$ 
from condition \eqref{eqn:thm_b_2}, 
$\norm{II}_2 = o_p\rbr{\sqrt{\frac{1}{nh}}}$,
which completes the proof.
\end{proof}

\begin{lemma}
\label{lem:os2}
Suppose Assumptions \ref{assumption:kernel}-\ref{assumption:growth}
and conditions \eqref{eqn:thm_b_1}-\eqref{eqn:thm_v_2} hold. Then
\[
\norm{S_d((\ah^\top,\ch^\top)^\top,\Vh)-S_d((\ah^\top,\ch^\top)^\top,\Vs)}_2 = O_p\rbr{\frac{sB_V\log (np)}{nh\sqrt{h_f}}} = o_p\rbr{\sqrt{\frac{1}{nh}}}. 
\]
\end{lemma}

\begin{proof}
Using the H\"{o}lder's inequality, we have
\begin{align*}
&\norm{S_d((\ah^\top,\ch^\top)^\top,\Vh)-S_d((\ah^\top,\ch^\top)^\top,\Vs)}_2 \\
&\ \leq \norm{\Vh - \Vs}_{1,F} \norm{\sum_i w_i \Gamma_i \Psi_{\tau}(y_i-\Gamma_i^\top (\ah^\top,\ch^\top)^\top)}_{\infty,2}\\
&\ = \norm{\Vh - \Vs}_{1,F} \norm{\sum_i w_i \Gamma_i \Psi_{\tau}(y_i-\Gamma_i^\top \bh)}_{\infty,2}\\
&\ \leq \norm{\Vh - V^{\star}}_{1,F} \cdot \left[\norm{\sum_i w_i \Gamma_i \Psi_{\tau}(y_i-q_i)}_{\infty,2}\right.\\ 
& \qquad\qquad\qquad\qquad\qquad\qquad\left.+\norm{\sum_i w_i \Gamma_i [\Psi_{\tau}(y_i-\Gamma_i^\top \bh)-\Psi_{\tau}(y_i-q_i)]}_{\infty,2}\right].
\end{align*}
Note that \[
\norm{\Vh - V^{\star}}_{1,F} = O_p\rbr{s B_V \sqrt{\frac{\log (np)}{nhh_f}}}
\]
by \eqref{eqn:thm_v_2}
and 
\[
\norm{\sum_i w_i \Gamma_i \Psi_{\tau}(y_i-q_i)}_{\infty,2} =O_p(\sqrt{\frac{\log (p)}{nh}})
\]
by \eqref{eq:qr_fixed:lambda} and \eqref{eq:qr_fixed:score_1} in Lemma \ref{lem:qr_fixed:probability:score} (presented later in section \ref{pf:bh}). Furthermore,
\begin{align*}
&\norm{\sum_i w_i \Gamma_i [\Psi_{\tau}(y_i-\Gamma_i^\top \bh)-\Psi_{\tau}(y_i-q_i)]}_{\infty,2}\\
& \qquad \leq \norm{\sum_i w_i \Gamma_i [\Psi_{\tau}(y_i-\Gamma_i^\top \bh)-\Psi_{\tau}(y_i-\Gamma_i^\top \bs)]}_{\infty,2} \\
& \qquad \qquad + \norm{\sum_i w_i \Gamma_i [\Psi_{\tau}(y_i-\Gamma_i^\top \bs)-\Psi_{\tau}(y_i-q_i)]}_{\infty,2}.
\end{align*}
The first term in the last inequality can be bounded as
\begin{align*}
&\norm{\sum_i w_i \Gamma_i [\Psi_{\tau}(y_i-\Gamma_i^\top\bh)-\Psi_{\tau}(y_i-\Gamma_i^\top \bs)]}_{\infty,2}\\
&\ = \norm{\sum_i w_i \Gamma_i \rbr{\One{y_i\leq \Gamma_i \bs}- \One{y_i\leq \Gamma_i \bh} }}_{\infty,2}\\
&\ \leq \norm{\sum_i w_i \Gamma_i \rbr{F_i(y_i\leq \Gamma_i \bs)- F_i(y_i\leq \Gamma_i \bh) }}_{\infty,2}\\
 &\quad+ \norm{\sum_i w_i \Gamma_i  \rbr{\One{y_i\leq \Gamma_i \bs}-F_i(y_i\leq \Gamma_i \bs)) -(\One{y_i\leq \Gamma_i \bh}-F_i(y_i\leq \Gamma_i \bh))}}_{\infty,2}\\
&=  O_p\rbr{\frac{s\log (np)}{nh\sqrt{h_f}}}+O_p\rbr{ B_K B_V \sqrt{\bar{f}B_X\frac{mr_b \log (np)}{nh}} },
\end{align*}
where the first part of the last equation is the same as 
the proof in Lemma \ref{lem:os1} and 
the second part comes from Lemma \ref{lem:os11} 
(presented later in Section \ref{pf:norm-lem}), 
where $r_b \leq O_p\rbr{s\sqrt{\frac{\log (np)}{nh}}}$ 
because of \eqref{eqn:thm_b_2} and $B_V \asymp O(\log p)$.

For the second term, applying Lemma \ref{lem:os3} (presented next) with a union bound, we have
\[
\norm{\sum_i w_i \Gamma_i [\Psi_{\tau}(y_i-\Gamma_i^\top \bs)-\Psi_{\tau}(y_i-q_i)]}_{\infty,2} = o_p\left(\sqrt{\frac{\log (np)}{nh}}\right). 
\]

Combining the two bounds, we have
\begin{align*}
&\norm{S_d((\ah^\top,\ch^\top)^\top,\Vh)-S_d((\ah^\top,\ch^\top)^\top,\Vs)}_F\\
&\ = O_p\left(sB_V\sqrt{\frac{\log (np)}{nhh_f}}\right) \cdot \\
&\qquad\left\{ O_p\left(\frac{s\log (np)}{nh\sqrt{h_f}}\right)+O_p\left(\sqrt{\frac{s(\log (np))^{7/2}}{(nh)^{3/2}}}\right)+o_p\left(\sqrt{\frac{\log (np)}{nh}}\right)\right\}\\
&=  o_p\rbr{\sqrt{\frac{1}{nh}}}. 
\end{align*}
The last equality is because of Assumption \ref{assumption:growth}.
This completes the proof.
\end{proof}

\begin{lemma}
\label{lem:os3}
Under Assumptions \ref{assumption:kernel}, \ref{assumption:u}, \ref{assumption:qr}, and \ref{assumption:X}, 
for any $V \in \mathbb{C}(S_2) \subset \R^{2k \times 2p}$ 
such that $\max_{i \in [n]}\norm{V\Gamma_i}_2 = O(\log p) \leq B_V$
and $h \leq O(n^{-1/3})$ as assumed in Assumption \ref{assumption:growth}, 
we have
\begin{multline*}
\Norm{\sum_i w_i V \Gamma_i[\Psi_{\tau}(y_i-q_i)-\Psi_{\tau}(y_i-\tilde{q}_i)] }_2 \\
= O_p\rbr{
    \sqrt{
    B_V B_K\frac{\bar f }{nh} \cdot  
    (h^2 + \epsilon_R)
    }+
    \frac{\bar f  \kappa_+ \norm{V}_F}{nh} \cdot  
    (h^2 + \epsilon_R)
    }.
\end{multline*}
\end{lemma}

\begin{proof}
We have
\begin{align*}
&\Norm{\EEn{ w_i V \Gamma_i[\Psi_{\tau}(y_i-q_i)-\Psi_{\tau}(y_i-\tilde{q}_i)] } }_2 \\
& \qquad =\Norm{\EEn{ w_i V \Gamma_i \left[\One{y_i\leq \tilde{q}_i}-\One{y_i\leq q_i}\right] } }_2\\
&\qquad \leq \underbrace{\Norm{\Gn{w_i V \Gamma_{i} (\One{y_i\leq \tilde{q}_i}- \One{y_i\leq q_i})}}_2 }_{I} \\
&\qquad\qquad\qquad + \underbrace{\norm{\EEn{ w_i V \Gamma_i \left[F_i(\tilde{q}_i)-F_i(q_i)\right] } }_2}_{II}.
\end{align*}
By Lemma \ref{lem:os31} (presented later in this section), 
\[
I = O_p\rbr{B_V B_K
    \sqrt{
    \frac{\bar f}{nh} \cdot  
    (h^2 + \epsilon_R)
    }
    }.
\]
For the term $II$, using the mean value theorem
and the Cauchy–Schwarz inequality, we have
\begin{multline*}
II 
\leq \bar f \sqrt{\sum_i w_i \tr (V \Gamma_i \Gamma_i^\top V^\top)} \cdot \sqrt{\sum_i w_i (\tilde{q}_i-q_i)^2} 
\\
= O_p\rbr{\bar f \kappa_+ \norm{V}_F \cdot( h^2+\epsilon_R)},
\end{multline*}
where the last equality is because of Lemma \ref{lem:qr_fixed:approx_err} and Assumption~\ref{assumption:X}. 
The proof follows from the rate of $h$ given in
Assumption \ref{assumption:growth}. 
\end{proof}

\begin{lemma}
\label{lem:os4}

Let
\begin{align} \label{eqn: sigma1}
\widehat{\Sigma}_{a1} & :=nh\sum_i w_i^2 \Vh \Gamma_i\Psi_{\tau}(y_i-\Gamma_i^\top\bh)\Psi_{\tau}(y_i-\Gamma_i^\top \bh)^\top \Gamma_i^\top \Vh^{\top}, \\
\label{eqn: sigma2}
\widehat{\Sigma}_{a2} &:=\tau(1-\tau)\nu_2 \Vh\left\{\sum_j w_j\Gamma_j\Gamma_j^{\top}\right\}\Vh^{\top}, \\
\intertext{and}
\Sigma&=\tau(1-\tau)\nu_2 \lim_{n\rightarrow \infty}\EEst{V^{\star}\Gamma\Gamma^\top V^{\star \top} }{U=u}. \nonumber
\end{align}
Then $\widehat{\Sigma}_{ai}\xrightarrow{p} \Sigma$ for $i=1,2$.
\end{lemma}

\begin{proof}

From the consistency of $\Vh$ and $\bh$, 
we have $\norm{\Vh-\Vs}_F=o_p(1)$ and 
\[
\max_i |\Psi_{\tau}(y_i-\Gamma_i^\top \bh)-\Psi_{\tau}(y_i-q_i)|=o_p(1).
\]
Therefore,
\begin{align*}
\widehat{\Sigma}_{a1} 
&= nh\sum_i w_i^2 \Vh \Gamma_i\Psi_{\tau}(y_i-\Gamma_i^\top \bh)\Psi_{\tau}(y_i-\Gamma_i^\top \bh)^\top \Gamma_i^\top \Vh^{\top}\\
&=nh\sum_i w_i^2 \Vs \Gamma_i\Psi_{\tau}(y_i-\Gamma_i^\top \bh)\Psi_{\tau}(y_i-\Gamma_i^\top \bh)^\top \Gamma_i^\top V^{\star\top} + o_p(1)\\
&=(nh)^{-1}\sum_i K^2(\frac{U_i-u}{h}) \Vs \Gamma_i\Psi_{\tau}(y_i-\Gamma_i^\top \bh)\Psi_{\tau}(y_i-\Gamma_i^\top \bh)^\top \Gamma_i^\top V^{\star\top} + o_p(1)\\
&=\EE{(nh)^{-1}\sum_i K^2(\frac{U_i-u}{h}) \Vs \Gamma_i\Psi_{\tau}(y_i-q_i)\Psi_{\tau}(y_i-q_i)^\top \Gamma_i^\top V^{\star\top} }+o_p(1)\\
&=\Sigma+o_p(1),
\end{align*}
which shows \eqref{eqn: sigma1},

From the condition \eqref{eqn:V1}, 
we have $\norm{\Vh-\Vs}_F=o_p(1)$.
By the strong law of large numbers,
\[
\Norm{\sum_j w_j V^{\star}\Gamma_j\Gamma_j^\top V^{\star \top}-\EEst{V^{\star}\Gamma\Gamma^\top V^{\star \top}}{U=u}}_F=o_p(1).
\]
Then by the continuous mapping theorem, 
we have $\norm{\widehat{\Sigma}_{a2}-\Sigma}_F=o_p(1)$.
This shows \eqref{eqn: sigma2}.
The proof is complete now.

\end{proof}

\begin{lemma}
\label{lem:os11}
Suppose Assumptions \ref{assumption:kernel}, \ref{assumption:u}, \ref{assumption:density}, and \ref{assumption:X} hold.
For any $V \in \mathbb{C}(S_2) \subset \R^{2k\times 2p}$
that satisfies $\max_{i \in [n]}||V \Gamma_i||_2 = B_V = O(\log p)$
and $r_b \asymp s\sqrt{\frac{\log (np)}{nh}}$,
we have
\begin{multline*}
 \sup_{ \substack{\norm{\delta}_{0,2} \leq m \\ \norm{\delta}_{1,2} \leq r_b } } \Norm{ \Gn{w_i V \Gamma_{i} (\One{y_i\leq \Gamma_i^{\top} \bs}- \One{y_i\leq \Gamma_i^{\top} (\bs+\delta)})} }_2\\
 = O_p\rbr{ B_K B_V \sqrt{\bar{f}B_X\frac{mr_b \log 
 (np)}{nh}} }.
\end{multline*}

\end{lemma}
\begin{proof}
 Let ${\cal W} = \cbr{\tilde{W}_1, \ldots, \tilde{W}_K}$ 
 be the $\frac 12$-net
  for $\cbr{W \in \R^{2k} \mid \norm{W}_2 \leq 1}$.
 That is, for all $W \in \R^{2k}$ with $\norm{W}_2 \leq 1$, 
 there exists $\tilde{W} \in {\cal W} \subseteq \cbr{W \in \R^{2k} \mid \norm{W}_2 \leq 1}$
 such that $\norm{\tilde{W}-W}_2 \leq \frac{1}{2}$. We have that $K \leq 5^{2k}$.
 Then
  \begin{align*}
    &\sup_{ \substack{\norm{\delta}_{0,2} \leq m \\ \norm{\delta}_{1,2} \leq r_b } } \Norm{\Gn{w_i \cdot V \Gamma_i \rbr{ \One{y_i\leq \Gamma_i^{\top} \bs}- \One{y_i\leq \Gamma_i^{\top} (\bs+\delta)} }}}_2 \\
    &\leq 2 \cdot \max_{\tilde{W} \in {\cal W}}
      \sup_{ \substack{\norm{\delta}_{0,2} \leq m \\ \norm{\delta}_{1,2} \leq r_b } }  \\
      &\qquad\qquad\qquad
      {\Gn{w_i\cdot \rbr{ |\One{y_i\leq \Gamma_i^{\top} \bs}- \One{y_i\leq \Gamma_i^{\top} (\bs+\delta)} |}\cdot |\tilde{W}^\top V \Gamma_{i}|}}.
    \end{align*}
For the expectation, we have
\begin{align*}
 &
 \EE{
    | \One{y_i\leq \Gamma_i^{\top} \bs}- \One{y_i\leq \Gamma_i^{\top} (\bs+\delta)}|} \\
    &
    \qquad\qquad\qquad
    \leq { \EE{ -|\Gamma_i^{\top} \delta| \leq y_i- \Gamma_i^{\top} \bs \leq |\Gamma_i^\top \delta| } }\\
& \qquad\qquad\qquad
= {
F_i\rbr{\Gamma_i^\top \bs + |\Gamma_i^{\top} \delta|} - F_i\rbr{\Gamma_i^{\top} \bs- |\Gamma_i^{\top} \delta| } }\\
  & \qquad\qquad\qquad
  \leq 2 \bar f \cdot  \abr{\Gamma_i^\top \delta} \\
  & \qquad\qquad\qquad
  \leq 2\bar f B_X r_b.
\end{align*}  
For a fixed $\tilde{W} \in {\cal W}$ and $|S| \leq m$, define
\begin{align*}
a_i &= w_i \cdot |\tilde{W}^\top V \Gamma_{i}|, \ \text{and} \\
{\cal G}_S &=
\Big\{(y_i, x_i, u_i) \mapsto a_i \cdot \One{ -|\Gamma_i^\top \delta| \leq y_i- \Gamma_i^\top \bs \leq |\Gamma_i^\top \delta|} :\\ 
& \qquad\qquad\qquad\qquad\qquad\qquad\qquad\qquad\qquad \supp\rbr{\delta} = S, \norm{\delta}_{2} \leq r_b\Big\}.\\
{\cal G} &= \cup_{S:|S|\leq m} {\cal G}_S
\end{align*}
Let $G(\cdot)=(nh)^{-1}B_K \rbr{\tilde{W}^\top V \Gamma_i}$ be an envelope of ${\cal G}$.
Then $\norm{G}_{\infty} \leq \frac{B_K B_V}{nh}$.
For a fixed $g \in {\cal G}$, let
\[
g_i = g(y_i, x_i, u_i) = a_i \cdot \One{ -|\Gamma_i^\top \delta| \leq y_i- \Gamma_i^\top \bs \leq |\Gamma_i^\top \delta|}.
\]  
Therefore, the variance is bounded as
\begin{align*}
  \sigma_{\cal G}^2
  % = \sup_{g \in {\cal G}} \VV{\sum_{i\in[n]} g_i - \EE{g_i}}
  \leq
  \sup_{g \in {\cal G}} \sum_{i \in [n]} \EE{g_i^2} \lesssim
  \bar f B_X r_b \sum_{i \in [n]} a_i^2
  \lesssim  \frac{\bar f B_X B_K^2 B_V^2}{nh}\cdot r_b,
\end{align*}
since 
\[
  \sum_{i \in [n]} a_i^2
  \leq B_V^2 \sum_{i \in [n]} w_i^2
  \leq \frac{ B_K^2 B_V^2}{nh}.
\]
The VC dimension for the space \[{\cal F}_S =
\cbr{(y_i, x_i, u_i) \mapsto \One{ -|\Gamma_i^\top \delta| \leq y_i- \Gamma_i^\top \bs \leq |\Gamma_i^\top \delta|} :  \supp\rbr{\delta} = S, \norm{\delta}_{2} \leq r_b}\] is $O(m)$. 
Therefore, applying Lemma \ref{lem:VC_to_cover} (presented later in Section \ref{pf:ep}),   
\[\sup_Q N\rbr{\epsilon\cdot \frac{B_KB_V}{nh}, {\cal G}_S, \norm{\cdot}_{L_2(Q)}} 
\leq \rbr{\frac{C}{\epsilon}}^{cm}, 
\]
and
\[
\sup_Q N\rbr{\epsilon\cdot \frac{B_KB_V}{nh}, {\cal G}, \norm{\cdot}_{L_2(Q)}} \leq \rbr{\frac{C}{\epsilon}}^{cm}     \cdot p^m . 
\]
Applying Lemma~\ref{lem:emp_proc:supremum_vc} (presented later in Section \ref{pf:V}) 
with
$\sigma_{\cal G}=B_KB_V\sqrt{\frac{\bar f B_X r_b}{nh}}$, $\norm{G}_{\infty} \leq \frac{B_K B_V}{nh}$,
$V = cm$, $U = \frac{B_K B_V}{nh}$, and $A = C m^{1/cm} p^{1/c}$ then gives us 
\begin{align*}
&\EE{ \sup_{g \in {\cal G}} \sum_{i \in [n]} g_i - \EE{g_i}}\\
& \leq \rbr{cm \frac{B_K B_V}{nh}  \log \rbr{\frac{C  p^{1/c}}{\sqrt{\bar f B_X r_b h}}} +
B_KB_V\sqrt{\frac{\bar f B_X r_b}{nh}} \sqrt{cm \log \rbr{\frac{C  p^{1/c}}{\sqrt{\bar f B_X r_b h}}}}}
 \\
& = O\rbr{ B_K B_V \sqrt{\bar{f}B_X\frac{mr_b \log (p\vee n)}{nh}} },
\end{align*}
under the conditions of the lemma and the growth condition in
Assumption \ref{assumption:growth}. Lemma~\ref{lm:proc_deviation} then gives us
\begin{align*}
  \sup_{ g \in {\cal G} }
  \sum_{i \in [n]} g_i - \EE{g_i} = O_p\rbr{ B_K B_V \sqrt{\bar{f}B_X\frac{mr_b \log (np)}{nh}} }.
\end{align*}
A union bound over $\tilde W \in {\cal W}$ concludes the proof.
\end{proof}

\begin{lemma}
\label{lem:os31}
Suppose
Assumptions \ref{assumption:kernel}, \ref{assumption:u}, \ref{assumption:density}, \ref{assumption:qr}, \ref{assumption:X} and \ref{assumption:growth} hold.
For all 
$V \in \mathbb{C}(S_2) \subset \R^{2k\times 2p}$ with $B_V=\max_{i \in [n]}||V \Gamma_i||_2$, we have
\begin{multline*}
  \Norm{ \Gn{w_i V \Gamma_{i} (\One{y_i\leq \tilde{q}_i}- \One{y_i\leq q_i})} }_2
  \\
  = O_P\rbr{B_KB_V 
    \sqrt{
    \frac{\bar f }{nh} \cdot  
    (h^2 + \epsilon_R)
    }
    }.
\end{multline*}

\end{lemma}
\begin{proof}
 Let ${\cal W} = \cbr{\tilde{W}_1, \ldots, \tilde{W}_K}$ be the $\frac 12$-net
  for $\cbr{W \in \R^{2k} \mid \norm{W}_2 \leq 1}$. We have that $K \leq 5^{2k}$ and
  \begin{multline*}
     \Norm{\Gn{w_i \cdot V \Gamma_i \rbr{ \One{y_i\leq \tilde{q}_i}- \One{y_i\leq q_i} }}}_2 \\
    \leq 2 \cdot \max_{\tilde{W} \in {\cal W}}
        \Gn{w_i\cdot \rbr{ |\One{y_i\leq \tilde{q}_i}- \One{y_i\leq q_i |}\cdot |\tilde{W}^\top V \Gamma_{i}|}}.
    \end{multline*}
Let
\begin{align*}
  a_i & = w_i \cdot |\tilde{W}^\top V \Gamma_{i}|, \ \text{ and}  \\
  g_i & = g(y_i, x_i, u_i) = a_i \cdot \rbr{\One{y_i\leq \tilde{q}_i}- \One{y_i\leq q_i}}.
\end{align*}
Since
\begin{align*}
  \EE{
    | \One{y_i\leq \tilde{q}_i}- \One{y_i\leq q_i}|}
  & \leq 2 \bar f \cdot  \abr{q_i - \tilde{q}_i},
\end{align*}  
we have
\begin{multline*}
   \sum_{i \in [n]} \EE{g_i^2} 
   \lesssim
    \sum_i 2 \bar f \cdot  \abr{q_i - \tilde{q}_i}a_i^2=\sum_i 2 \bar f \cdot  \abr{q_i - \tilde{q}_i}w_i^2\rbr{\tilde{W}^\top V \Gamma_{i}}^2\\
\leq \frac{\bar f B_K^2B_V^2}{nh}   
     \sqrt{\EEn{w_i\rbr{q_i - \tilde{q}_i}^2}}.
\end{multline*}
Then by Lemma~\ref{lem:qr_fixed:approx_err} and Assumption~\ref{assumption:X}, we have
\begin{align*}
   \sum_{i \in [n]} \EE{g_i^2} 
   = O_p\rbr{
    \frac{\bar f B_K^2B_V^2}{nh} \cdot  
    (h^2 + \epsilon_R)
    }.
\end{align*}
The result follows from the Bernstein's inequality and the union bound over ${\cal W}$.
\end{proof}

\subsection{Consistency of the initial estimator $\bh^{\rm ini}$}\label{pf:bh}

We show the convergence guarantee of the initial estimator
$\bh^{\rm ini}$ defined in Section \ref{algorithm} Step 1. 
Notice that in the following two sections,
we slightly abuse the notation by denoting $\bh^{\rm ini}$ from Section \ref{algorithm} as $\bh$, 
and 
$\bh$ from Section \ref{algorithm} is defined as $\bh^{\lambda}$,
since it is obtained by thresholding at the level $\lambda$. 

Let
\begin{equation}
    \label{eq:qr_fixed:Wi_definition}
    W_i(\delta) = {\rho_\tau\rbr{y_i - (\fq + \delta)} - \rho_\tau\rbr{y_i - \fq}},
\end{equation}
which can be decomposed as
\begin{equation}
\begin{aligned}
  \label{eq:qr_fixed:knight}
W_i(\delta) 
& = -\delta \Psi_\tau(y_i - \fq) + \int_0^{\delta}\sbr{\One{y_i\leq \fq + z} - \One{y_i\leq \fq}} dz \\
& =: W_i^{\#}(\delta) + W_i^{\natural}(\delta), 
\end{aligned}
\end{equation}
using the Knight's identity. 
Note that we can also write
\begin{equation}
  \label{eq:qr_fixed:direct}
W_i^{\natural}(\delta) = (y_i-(\fq+\delta))\sbr{\One{\fq + \delta \leq y_i < \fq} - \One{\fq \leq y_i < \fq + \delta}}.
\end{equation}
With this notation, we study properties of the following penalized quantile regression estimator
\begin{equation}
\label{eq:eq_fixed:opt}
\bh = \arg \min_b \sum_{i \in [n]} w_i \cdot \rho_\tau(y_i - \Gamma_i^\top b) + \lambda \norm{b}_{1,2},
\end{equation}
where the groups are formed by pairs $(b_{0j}, b_{1j})$ 
for
$j \in [p]$ and $w_i = (nh)^{-1} K\rbr{\frac{U_i - u}{h}}$. 
The estimated quantile function is 
denoted as $\fhq = \Gamma_i^\top\bh$.

\begin{theorem}
  \label{thm:qr_fixed:main}
Under Assumptions \ref{assumption:kernel},\ref{assumption:density}, \ref{assumption:qr} and \ref{assumption:X}, 
we have 
\begin{align*}
\EEn{w_i \cdot \rbr{\Gamma_i^\top \rbr{\bh - \bs}}^2} & = O_p\rbr{\frac{s\log (np)}{nh}} \\
\intertext{and}
\norm{\bh-\bs}_{1,2} & = O_p\rbr{s\sqrt{\frac{\log (np)}{nh}}}.    
\end{align*}

 \end{theorem}
  
 \begin{proof}
Denote $S' = \supp \{\bs \}$. 
Let $r_b$ be a rate 
satisfying $r_b = O_p\rbr{\sqrt{\frac{s\log (np)}{nh}}}$.
Recall that 
\[
    \kappa_q = \inf_{
\substack{\norm{\delta}_{1,2} \leq \frac{7|S'| \cdot \sqrt{\log p}}{\kappa_- \sqrt{nh}} \\ \EEn{w_i(\Gamma_i^\top\delta)^2} = \frac{|S'|\log p}{nh} }
}
      \frac{\rbr{ \underline{f} \cdot \EEn{w_i\cdot\rbr{\Gamma_i^\top \delta}^2}}^{3/2}}{\bar{f}' \cdot \EEn{w_i\cdot\rbr{\Gamma_i^\top \delta}^3}}.
\]  
As $n$ grows, we have
\begin{equation}
  \label{eq:qr_fixed:rate_condition_rb}
  \kappa_q \geq r_b \sqrt{\underline f}.
\end{equation}
In order to establish that $\EEn{w_i \cdot \rbr{\Gamma_i^\top \rbr{\bh - \bs}}^2} \leq r_b^2$, 
we use the proof by contradiction.

Suppose that
$\EEn{w_i \cdot \rbr{\Gamma_i^\top \rbr{\bh - \bs}}^2} > r_b^2$. Since the objective function is convex, there
exists a vector 
\[ \bc = \bs + (\bh-\bs) \frac{r_b}{\sqrt{\EEn{w_i \cdot \rbr{\Gamma_i^\top \rbr{\bh - \bs}}^2}}} \] 
such that $\EEn{w_i \cdot \rbr{\Gamma_i^\top \rbr{\bc - \bs}}^2} = r_b^2$
and
\begin{equation*}
\EEn{w_i \cdot\rbr{\rho_\tau(y_i - \check{q}_i) - \rho_\tau(y_i - \fql)}} \leq
\lambda \rbr{\norm{\bc}_{1,2}-\norm{\bs}_{1,2}},
\end{equation*}
is satisfied. We separate our analysis into two parts,
according to whether $3\norm{(\bc - \bs)_{S'}}_{1,2} \geq \frac{2}{\lambda}\EEn{w_i \cdot W_i^{\natural}(\fql - \fq)}$ 
in \eqref{eq:qr_fixed:cone_condition} or not.

First, suppose that $3\norm{(\bc - \bs)_{S'}}_{1,2} \geq \frac{2}{\lambda}\EEn{w_i \cdot W_i^{\natural}(\fql - \fq)}$.
By Lemma \ref{lem:6} (presented next), $\norm{(\bc - \bs)_{N'}}_{1,2} \leq 6\norm{(\bc - \bs)_{S'}}_{1,2}$ and 
  \begin{multline}
    \label{eq:qr_fixed:l1_bound:1}
    \norm{\bc - \bs}_{1,2}
    \leq 7 \norm{(\bc - \bs)_{S'}}_{1,2}
    \leq 7 \sqrt{|S'|} \norm{(\bc - \bs)_{S'}}_{2} \\
    \leq \frac{7 \sqrt{|S'|} \cdot r_b}{\kappa_-} = O_p\rbr{s\sqrt{\frac{\log (np)}{nh}}}.
  \end{multline}
Starting from \eqref{eq:qr_fixed:start}, we have that
\begin{multline*} 
  \lambda \norm{\bc - \bs}_{1,2} 
  \geq \bE{w_i \cdot \rbr{W_i\rbr{\Gamma_i^\top\bc - \fq} - W_i\rbr{\Gamma_i^\top\bs - \fq}}} \\
   + \Gn{w_i \cdot \rbr{W_i\rbr{\Gamma_i^\top\bc - \fq} - W_i\rbr{\Gamma_i^\top\bs - \fq}}}.
\end{multline*}
Lemma~\ref{lem:qr_fixed:quad_lower_bound} (presented later in this section) gives us 
\[
  \bE{w_i \cdot \rbr{W_i\rbr{\Gamma_i^\top\bc - \fq} - W_i\rbr{\Gamma_i^\top\bs - \fq}}}
  \geq \frac{\underline{f} r_b^2 \wedge \kappa_q \sqrt{\underline{f}} r_b }{3} \geq \frac{\underline{f} r_b^2 }{3},
\]
where the second inequality follows under~\eqref{eq:qr_fixed:rate_condition_rb}. 
On the event $\cE_{\rm QR}(\lambda)$, we have
\[
  \Gn{w_i \cdot \rbr{W_i^{\#}\rbr{\Gamma_i^\top\bc - \fq} - W_i^{\#}\rbr{\Gamma_i^\top\bs - \fq}}}
  \geq -\frac{\lambda}{2}\norm{\bc - \bs}_{1,2}.
\]
Lemma~\ref{lem:qr_fixed:emp_proc} (presented later in this section) gives us
\begin{multline*}
  \sup_{
      \substack{\norm{\delta}_{1,2} \leq \frac{7 \sqrt{|S'|} \cdot r_b}{\kappa_-} \\ \EEn{w_i(\Gamma_i^\top\delta)^2} = r_b^2 }
    } \abr{
  \Gn{w_i \cdot \rbr{W_i^{\natural}\rbr{\Gamma_i^\top\bc - \fq} - W_i^{\natural}\rbr{\Gamma_i^\top\bs - \fq}}}}
  = o_p(r_b).
\end{multline*} 
Putting everything together, we obtain that
\begin{equation*}
  0
  \geq \frac{\underline{f} r_b }{3} - \frac{7 \lambda\sqrt{|S'|}}{2 \kappa_-} 
  - o_p(r_b)
    > 0,
\end{equation*}
which is a contradiction.

The second part of the upper bound is established in the case when
$$3\norm{(\bc - \bs)_{S'}}_{1,2} < (2/\lambda)\EEn{w_i \cdot W_i^{\natural}(\fql - \fq)}.$$
Then we have that 
\[
\norm{(\bc - \bs)_{N'}}_{1,2} \leq \frac{4}{\lambda}\EEn{w_i \cdot W_i^{\natural}(\fql - \fq)}
\]
and
\[
  \norm{\bc - \bs}_{1,2}  \leq \frac{6}{\lambda}\EEn{w_i \cdot W_i^{\natural}(\fql - \fq)} = O_p\rbr{s\sqrt{\frac{\log (np)}{nh}}},
\]
where the last equation is because of 
Lemmas \ref{lem:qr_fixed:probability:score} and \ref{lem:11}.
The same argument as above gives us a contradiction,
which completes the proof.

\end{proof}

\begin{lemma} \label{lem:6}
 On the event 
\begin{equation}
  \label{eq:score_c}
  \cE_{\rm QR}(\lambda) = \cbr{
  \max_{j \in [p]} \sup_{{v:=(v_0,v_1) \in \R^2} \atop {\norm{v}_2=1}}
  \sum_i w_i \cdot\rbr{ x_{ij}v_0 + x_{ij}(u_i-u)v_1} \cdot \Psi_\tau\rbr{y_i - \fq} \leq \frac{\lambda}{2} }
\end{equation}
we have 
\begin{equation}
  \label{eq:qr_fixed:cone_condition}
	\norm{(\bh - \bs)_{N'}}_{1,2} \leq 3\norm{(\bh - \bs)_{S'}}_{1,2} + \frac{2}{\lambda}\EEn{w_i \cdot W_i^{\natural}(\fql - \fq)},
 \end{equation}
where $N' = S'^c$ .
 \end{lemma}
\begin{proof}
  Our starting point is the observation that 
  \begin{equation}
    \label{eq:qr_fixed:start}
    \EEn{w_i \cdot\rbr{\rho_\tau(y_i - \fhq) - \rho_\tau(y_i - \fql)}} \leq
    \lambda \rbr{\norm{\bs}_{1,2}-\norm{\bh}_{1,2}},
  \end{equation}
  since $\bh$ minimizes \eqref{eq:eq_fixed:opt}. 
  Due to convexity of $\rho_{\tau}(\cdot)$, we have
  \begin{equation}
    \label{eq:qr_fixed:lower_bound:1}
    \EEn{w_i \cdot\rbr{\rho_\tau(y_i - \fhq) - \rho_\tau(y_i - \fq)}} \geq
    \EEn{w_i \cdot \rbr{\fq - \fhq} \cdot \Psi_\tau(y_i - \fq)}.
  \end{equation}
  Using \eqref{eq:qr_fixed:knight}, we have
  \begin{multline}
    \label{eq:qr_fixed:lower_bound:2}
    \EEn{w_i \cdot\rbr{\rho_\tau(y_i - \fql) - \rho_\tau(y_i - \fq)}} \\
    =
    \EEn{w_i \cdot \rbr{\fq - \fql} \cdot \Psi_\tau(y_i - \fq)} + \EEn{w_i \cdot W_i^{\natural}(\fql - \fq)}.
  \end{multline}
  Combining \eqref{eq:qr_fixed:lower_bound:1} and \eqref{eq:qr_fixed:lower_bound:2} with \eqref{eq:qr_fixed:start},
  we have
  \begin{equation}
    \label{eq:qr_fixed:1}
    \EEn{w_i \cdot \rbr{\fql - \fhq} \cdot \Psi_\tau(y_i - \fq)} - \EEn{w_i \cdot W_i^{\natural}(\fql - \fq)}
    \leq \lambda \rbr{\norm{\bs}_{1,2}-\norm{\bh}_{1,2}}.
  \end{equation}
On the event $\cE_{\rm QR}(\lambda)$,
  \[
    \EEn{w_i \cdot \rbr{\fql - \fhq} \cdot \Psi_\tau(y_i - \fq)}  \geq - \frac{\lambda}{2} \norm{\bs - \bh}_{1,2} .
  \]
  Combining with the display above, we obtain that
\begin{equation*}
  - \frac{\lambda}{2} \norm{\bs - \bh}_{1,2} \leq \EEn{w_i \cdot W_i^{\natural}(\fql - \fq)} + \lambda \rbr{\norm{\bs}_{1,2}-\norm{\bh}_{1,2}}.
\end{equation*}
Since
\[
\norm{\bs}_{1,2}-\norm{\bh}_{1,2} \leq
\norm{\rbr{\bs - \bh}_{S'}}_{1,2}-\norm{\rbr{\bs - \bh}_{N'}}_{1,2},
\]
we have
\begin{equation*}
\frac{\lambda}{2} \norm{(\bh - \bs)_{N'}}_{1,2} \leq \EEn{w_i \cdot W_i^{\natural}(\fql - \fq)} +\frac{3\lambda}{2} \norm{(\bh - \bs)_{S'}}_{1,2},
\end{equation*}
which completes the proof.
\end{proof}

\begin{lemma}
  \label{lem:qr_fixed:probability:score}
  Under Assumption \ref{assumption:X}, for 
  \begin{equation}
    \label{eq:qr_fixed:lambda}
    \lambda = 4 \cdot \rbr{\max_{j \in [p]} \EEn{w_i^2 x_{ij}^2}}^{1/2} \sqrt{ \log(4p/\gamma)} = O\left(\sqrt{\frac{\log p}{nh}}\right),
    \end{equation}
  we have
  \begin{equation*}
    \PP{\cE_{\rm QR}(\lambda) } \geq 1 - \gamma.
  \end{equation*}
\end{lemma}

\begin{proof}
  We prove that
  \begin{equation}
    \label{eq:qr_fixed:score_1} 
    \max_{j \in [p]} 
    \abr{    
    \EEn{ w_i \cdot x_{ij} \cdot \Psi_\tau\rbr{y_i - \fq} }
    }\leq \frac{\lambda}{2\sqrt{2}}    
  \end{equation}
  and
  \begin{equation}
    \label{eq:qr_fixed:score_2}
    \max_{j \in [p]} 
    \abr{    
    \EEn{ w_i \cdot x_{ij}(u_i - u) \cdot \Psi_\tau\rbr{y_i - \fq} }
    }\leq \frac{\lambda}{2\sqrt{2}}.
  \end{equation}
  Let 
  \[
  Z_i = \frac{w_i \cdot x_{ij} \cdot \Psi_\tau\rbr{y_i - \fq}}{\sqrt{ \EEn{w_i^2 x_{ij}^2}}}
  \]
  and note that $|Z_i| \leq 1$ and $\EE{Z_i} = 0$. 
  The Hoeffding's inequality \citep[Theorem 2.8][]{Boucheron2013Concentration} 
  gives us that
  \[
    \abr{\EEn{Z_i}} \leq \sqrt{2 \log(2/\gamma)}
  \]
  with probability $1-\gamma$.
  An application of the union bound gives us that
  \[
    \max_{j \in [p]} \EEn{w_i \cdot x_{ij} \cdot \Psi_\tau\rbr{y_i - \fq}} \leq  \rbr{\max_{j \in [p]} \EEn{w_i^2 x_{ij}^2}}^{1/2} \sqrt{2 \log(4p/\gamma)}
  \]
  with probability $1-\gamma/2$. This proves \eqref{eq:qr_fixed:score_1}. 
  Equation \eqref{eq:qr_fixed:score_2} is shown in the same way by noting that 
  \[
    \EEn{w_i^2 x_{ij}^2(u_i - u)^2} \leq     \EEn{w_i^2 x_{ij}^2}.
  \]
\end{proof}

\begin{lemma}
  \label{lem:qr_fixed:emp_proc}
   Let $\bc = \bs + \delta$,
  \begin{align*}
    g_i(\delta) 
    &= w_i \cdot \rbr{W_i^{\natural}\rbr{\Gamma_i^\top\bc - \fq} - W_i^{\natural}\rbr{\Gamma_i^\top\bs - \fq}}\\
    &= w_i\cdot    (y_i-\Gamma_i^\top\bc)\sbr{\One{\Gamma_i^\top \bc \leq y_i < q_i} - \One{q_i < y_i < \Gamma_i^\top \bc}} \\ 
    &\qquad - w_i \cdot (y_i-\Gamma_i^\top\bs)\sbr{\One{\Gamma_i^\top \bs \leq y_i < q_i} - \One{q_i < y_i < \Gamma_i^\top \bs}}.
  \end{align*}  
  Then 
\begin{multline*}
\sup_{
      \substack{\norm{\delta}_{1,2} \leq R_1 \\ \EEn{w_i(\Gamma_i^\top\delta)^2} \leq R_2 }
    } \abr{ \Gn{ w_i \cdot g_i(\delta)}} \\
    \lesssim   R_1 
      \sqrt{\frac{B_X^2B_K^2}{nh}\log(p)}
      + R_1 \sqrt{\rbr{B_w  B_X \sqrt{\frac{B_X^2B_K^2}{nh}\log(p)} + B_w \frac{R_2}{R_1}}\log(1/\gamma)} \\
      + R_1B_w  B_X \log(1/\gamma)         
\end{multline*}
with probability $1-\gamma$.  
\end{lemma}

\begin{proof}
  We will apply Lemma~\ref{lm:proc_deviation}.   
  Note that
  \begin{align*}
  |g_i(\delta)|  & 
  \leq w_i  (y_i-\Gamma_i^\top\bc) \cdot
  \left[
  \One{\Gamma_i^\top \bc \leq y_i < q_i} 
  - \One{q_i < y_i < \Gamma_i^\top \bc} \right.
  \\
  & \qquad\qquad\qquad\qquad\qquad
  \left.
  -\One{\Gamma_i^\top \bs \leq y_i < q_i} 
  + \One{q_i < y_i < \Gamma_i^\top \bs}\right]\\
  & \qquad + w_i (\Gamma_i^\top\bc - \Gamma_i^\top\bs)\sbr{ \One{\Gamma_i^\top \bs \leq y_i < q_i} - \One{q_i < y_i < \Gamma_i^\top \bs}}\\
    &\leq 2|w_i \Gamma_i^\top\delta| \leq 2 B_w  B_X R_1.
 \end{align*}
Therefore, 
$|g_i(\delta) - \EE{g_i(\delta)}| \leq 4 B_w  B_X R_1$.
For the variance, we have
  \begin{align*}
    \EE{\sum_i \rbr{g_i(\delta) - \EE{g_i(\delta)}}^2} \leq \EE{\sum_i g_i^2(\delta)} \leq 4 \sum_{i\in[n]} w_i^2 \rbr{\Gamma_i^\top \delta}^2 \leq 4 B_w  R_2.
\end{align*}
Finally, we bound the supremum of the process. We have 
\[
  \begin{aligned}
  & \EE{\sup_{\substack{\norm{\delta}_{1,2} \leq R_1 \\ \EEn{w_i(\Gamma_i^\top\delta)^2} = R_2 }} \abr{\Gn{g_i(\delta)}}} \\
  \overset{(i)}{\leq} & 2\EE{\sup_{\substack{\norm{\delta}_{1,2} \leq R_1 \\ \EEn{w_i(\Gamma_i^\top\delta)^2} = R_2 }} \abr{ \EEn{\epsilon_i \cdot g_i(\delta)}}} \\
  \overset{(ii)}{\leq} & 4 \EE{\sup_{\substack{\norm{\delta}_{1,2} \leq R_1 \\ \EEn{w_i(\Gamma_i^\top\delta)^2} = R_2 }}\abr{\EEn{\epsilon_i w_i \Gamma_i^\top\delta}}} \\
  \leq & 4 \norm{\delta}_{1,2} \EE{\max_{j \in [p]}\abr{\EEn{\epsilon_i w_i \Norm{\sbr{
            \begin{array}{c}
              x_{ij} \\
              x_{ij}(U_i - u)
            \end{array}}
        }_2}}} \\
  \overset{(iii)}{\leq} & 8 \norm{\delta}_{1,2} \sqrt{\frac{B_X^2B_K^2}{nh}\log(2p)},
\end{aligned}
\]
where $(i)$ follows from symmetrization \cite[Lemma~11.4][]{Boucheron2013Concentration},
$(ii)$ from contraction inequality \cite[Theorem~11.6][]{Boucheron2013Concentration},
and $(iii)$ from a maximum inequality for sub-Gaussian random variables \cite[Theorem~2.5][]{Boucheron2013Concentration}.
The result now follows by plugging the pieces into Lemma~\ref{lm:proc_deviation}.

\end{proof}

\begin{lemma}  
  \label{lem:qr_fixed:quad_lower_bound}
  Suppose Assumptions \ref{assumption:density} and \ref{assumption:X} hold.
  For a fixed $\delta$, we have
  \begin{multline*}
      \bE{w_i \cdot \rbr{W_i\rbr{\Gamma_i^\top\rbr{\bs+\delta} - \fq} - W_i\rbr{\Gamma_i^\top\bs - \fq}}} \\ \geq
    \frac{1}{3} \cdot \rbr{
      \underline{f} \EEn{w_i\cdot\rbr{\Gamma_i^\top \delta}^2} \wedge
      \kappa_q \cdot \rbr{ \underline{f} \cdot \EEn{w_i\cdot\rbr{\Gamma_i^\top \delta}^2} }^{1/2}  
    }.
  \end{multline*}
\end{lemma}

\begin{proof}
  Using \eqref{eq:qr_fixed:knight}, we have
  \begin{equation*}
    \EE{W_i^{\#}\rbr{\Gamma_i^\top\rbr{\bs+\delta} - \fq} - W_i^{\#}\rbr{\Gamma_i^\top\bs - \fq}}
    = -\rbr{\Gamma_i^\top\delta}\cdot\EE{\psi_\tau(y_i - \fq)}
    = 0
  \end{equation*}
  and
  \begin{equation*}
  \begin{aligned}
    & \EE{W_i^{\natural}\rbr{\Gamma_i^\top\rbr{\bs+\delta} - \fq} - W_i^{\natural}\rbr{\Gamma_i^\top\bs - \fq}}\\
    &\qquad \qquad=
\EE{\int_0^{\Gamma_i^\top \delta}\sbr{\One{y_i\leq\Gamma_i^\top\bs+z} - \One{y_i\leq\Gamma_i^\top\bs}} dz } \\
& \qquad \qquad= 
\int_0^{\Gamma_i^\top \delta}\sbr{F_i\rbr{\Gamma_i^\top\bs+z} - F_i\rbr{\Gamma_i^\top\bs}} dz  \\
& \qquad \qquad= 
\int_0^{\Gamma_i^\top \delta}\sbr{
zf_i\rbr{\Gamma_i^\top\bs} + \frac{z^2}{2}f_i'\rbr{\Gamma_i^\top\bs + \tilde z}} dz \qquad\qquad \rbr{\tilde z \in [0, z]}\\
& \qquad \qquad\geq 
\frac{\underline{f}}{2}\rbr{\Gamma_i^\top \delta}^2
- \frac{\bar{f}'}{6}\rbr{\Gamma_i^\top \delta}^3.
\end{aligned}
\end{equation*}
Combining the last two displays, we obtain
\begin{multline}
  \label{eq:qr_fixed:quad_lower_bound:1}
  \bE{w_i \cdot \rbr{W_i\rbr{\Gamma_i^\top\rbr{\bs+\delta} - \fq} - W_i\rbr{\Gamma_i^\top\bs - \fq}}}
  \\ \geq \frac{\underline{f}}{2} \EEn{w_i\cdot\rbr{\Gamma_i^\top \delta}^2} - \frac{\bar f' }{6}\EEn{w_i\cdot \rbr{\Gamma_i^\top \delta}^3}.  
\end{multline}
We lower bound the above display in two cases. 
First, consider the case where 
\[
  \rbr{ \underline{f} \cdot \EEn{w_i\cdot\rbr{\Gamma_i^\top \delta}^2}}^{1/2} \leq \kappa_q.
\]
From the definition of $\kappa_q$, we then obtain that
\[
  \bar{f}' \cdot \EEn{w_i\cdot\rbr{\Gamma_i^\top \delta}^3} \leq
  \underline{f} \cdot \EEn{w_i\cdot\rbr{\Gamma_i^\top \delta}^2},
\]
which combined with \eqref{eq:qr_fixed:quad_lower_bound:1} gives us
\begin{multline}
  \label{eq:qr_fixed:quad_lower_bound:2}
  \bE{w_i \cdot \rbr{W_i\rbr{\Gamma_i^\top\rbr{\bs+\delta} - \fq} - W_i\rbr{\Gamma_i^\top\bs - \fq}}}
 \\ \geq \frac{\underline{f}}{3} \EEn{w_i\cdot\rbr{\Gamma_i^\top \delta}^2}.  
\end{multline}
Next, we consider the case  where
\[
  \rbr{ \underline{f} \cdot \EEn{w_i\cdot\rbr{\Gamma_i^\top \delta}^2}}^{1/2} > \kappa_q.
\]
Let $\bar b = \bs + (1-\alpha) \delta$ for some
$\alpha \in (0,1)$ to be determined later. Using the convexity 
of $\rho_\tau(\cdot)$, we have
that 
\begin{equation}
  \label{eq:qr_fixed:quad_lower_bound:3}
\begin{aligned}
  & \bE{w_i \cdot \rbr{W_i\rbr{\Gamma_i^\top\rbr{\bs+\delta} - \fq} - W_i\rbr{\Gamma_i^\top\bs - \fq}}} \\
  &\quad\geq \frac{1}{1 - \alpha} \rbr{\bE{w_i \cdot \rbr{W_i(\Gamma_i^\top\bar b - \fq) - W_i(\fql - \fq)}}} \\
 &\quad\geq \frac{1}{1 - \alpha} \rbr{ \frac{\underline{f}}{2} \EEn{w_i\cdot\rbr{\Gamma_i^\top \rbr{\bar b - \bs}}^2} - \frac{\bar f' }{6}\EEn{w_i\cdot \rbr{\Gamma_i^\top \rbr{\bar b - \bs}}^3}}, 
\end{aligned}
\end{equation}
where the second inequality follows from \eqref{eq:qr_fixed:quad_lower_bound:1}. 
We want to choose $\alpha$ such that
\[
 \underline{f} \cdot \EEn{w_i\cdot\rbr{\Gamma_i^\top \rbr{\bar b - \bs}}^2} = \bar f'\EEn{w_i\cdot \rbr{\Gamma_i^\top \rbr{\bar b - \bs}}^3},
\]
which leads to
\[
  1 - \alpha = \frac{ \underline{f} \cdot \EEn{w_i\cdot\rbr{\Gamma_i^\top \delta}^2} }{ \bar f'\EEn{w_i\cdot \rbr{\Gamma_i^\top \delta}^3}}.  
\]
Combining with \eqref{eq:qr_fixed:quad_lower_bound:3}, we have
\begin{equation}
  \label{eq:qr_fixed:quad_lower_bound:4}
\begin{aligned}
   & \bE{w_i \cdot \rbr{W_i\rbr{\Gamma_i^\top\rbr{\bs+\delta} - \fq} - W_i\rbr{\Gamma_i^\top\bs - \fq}}}\\
   & \quad \geq \frac{1}{3} \cdot \frac{ \rbr{ \underline{f} \cdot \EEn{w_i\cdot\rbr{\Gamma_i^\top \delta}^2} }^ 2}{ \bar f'\EEn{w_i\cdot \rbr{\Gamma_i^\top \delta}^3}} \\
   & \quad \geq \frac{\kappa_q}{3} \cdot \rbr{ \underline{f} \cdot \EEn{w_i\cdot\rbr{\Gamma_i^\top \delta}^2} }^{1/2}.
\end{aligned}
\end{equation}
The proof follows by combining the lower bounds in \eqref{eq:qr_fixed:quad_lower_bound:2} and \eqref{eq:qr_fixed:quad_lower_bound:4}.
\end{proof}

\begin{lemma}\label{lem:11}
Under our model assumptions,
\begin{multline}
  \label{eq:qr_fixed:bound_approx}
\EEn{w_i\cdot W_i^{\natural}(\fql - \fq)} \\
\leq 2 \sqrt{2\log(2/\gamma)} \cdot \rbr{\bar f \cdot \EEn{w_i \cdot \rbr{\fql - \fq}^2}+\sqrt{\EEn{w_i^2 \rbr{\fql - \fq}^2}}}
\end{multline}
holds with probability $1-\gamma$.    
\end{lemma}
\begin{proof}
 We will prove the lemma using Theorem 2.16 of \cite{Pena2009Self}. Note that
  $W_i^{\natural}(\fql - \fq)$ is positive and
  \begin{equation*}
    \begin{aligned}
      \EE{W_i^{\natural}(\fql - \fq)}
      & = \int_0^{\fql - \fq} \sbr{F_i(\fq + z) - F_i(\fq)} dz\\ 
      & = \int_0^{\fql - \fq} f_i(\tilde z_i) z dz 
       \leq \frac{\bar f}{2} \rbr{\fql - \fq}^2,
    \end{aligned}
  \end{equation*}
  where $\tilde{z}_i$ is a point between $0$ and $z$.
  Therefore, the Markov's inequality gives us
  \[
    \PP{\EEn{w_i\cdot W_i^{\natural}(\fql - \fq)} \geq 2\bar f \cdot \EEn{w_i \cdot \rbr{\fql - \fq}^2}} \leq \frac14.
  \]
  Furthermore, since $|W_i^{\natural}(\fql - \fq)| \leq |\fql - \fq|$, we have that
  \[
    \PP{\EEn{\rbr{w_i\cdot W_i^{\natural}(\fql - \fq)}^2} \geq \EEn{w_i^2 \rbr{\fql - \fq}^2}} = 0 \leq \frac14.
  \]
  Invoking Theorem 2.16 of \cite{Pena2009Self}, 
  define 
  \begin{gather*}
  a = 2 \bar{f} \EEn{w_i \cdot \rbr{\fql - \fq}^2}, \qquad    
  b= \sqrt{\EEn{w_i^2 \rbr{\fql - \fq}^2}}, \\
  S_n = \EEn{w_i\cdot W_i^{\natural}(\fql - \fq)}, \qquad
  V_n^2 = \EEn{\rbr{w_i\cdot W_i^{\natural}(\fql - \fq)}^2},
  \end{gather*}
  and
  observe that $V_n \leq b^2$, we obtain that
  \[
 \PP{S_n \geq x(a+b+ V_n)} \leq 2 e^{\frac{-x^2}{2}},
  \]
  which completes the proof.
\end{proof}

\begin{lemma}
Under Assumptions \ref{assumption:kernel}, \ref{assumption:u}, \ref{assumption:qr}, and \ref{assumption:X}, we have that
\label{lem:qr_fixed:approx_err}
\begin{align*}
    \sum_i w_i  \rbr{\fql - \fq}^2 & \leq 2h^4 s B_X^2 B_{\beta} B_K + 2 \frac{\epsilon_R^2}{\underline{f}^2} = O\left(h^4 + \epsilon_R^2\right) \\
    \intertext{and}
    \sum_i w_i^2 (\tilde{q}_i - q_i)^2 & \leq B_w\cdot (2h^4 s B_X^2 B_{\beta} B_K + 2 \frac{\epsilon_R^2}{\underline{f}^2}) = O\left(\frac{h^3}{n} + \frac{\epsilon_R^2}{nh}\right).
\end{align*}
\end{lemma}

\begin{proof}
First, the assumption on the density $f_i$ proves that
\[\left|x_i^\top \beta^\star(\tau,u_i) - q_i\right| \leq \frac{R_i}{\underline{f}}.\]
Then
\begin{align*}
\sum_i &w_i (\tilde{q}_i - q_i)^2 \\
&\leq 2\sum_i w_i (\tilde{q}_i - x_i^\top \beta^\star(\tau,u_i)))^2 + 2\sum_i w_i (x_i^\top \beta^\star(\tau,u_i) - q_i)^2 \\
&\leq 2\sum_i w_i [x_i^\top (\beta^{\star}(\tau,u_i)- \beta^{\star}(\tau,u) - (u_i-u)\cdot \nabla_u \beta^{\star}(\tau,u))]^2 + 2\sum_i w_i\frac{R_i^2}{\underline{f}^2}\\
&\leq 2\sum_i w_i \|x_i\|^2_{\infty}\big\|\beta^{\star}(\tau,u_i)- \beta^{\star}(\tau,u) - (u_i-u)\cdot \nabla_u \beta^{\star}(\tau,u)\big\|^2_1 + 2\frac{\epsilon_R^2}{\underline{f}^2}\\
&\leq 2sB_X^2 B_\beta\sum_i w_i (u_i-u)^4 + 2\frac{\epsilon_R^2}{\underline{f}^2}\\
&= h^4\cdot  2sB_X^2 B_\beta\sum_i w_i \left(\frac{u_i-u}{h}\right)^4 + 2\frac{\epsilon_R^2}{\underline{f}^2}\\
&\leq 2h^4 sB_X^2 B_\beta B_K + 2\frac{\epsilon_R^2}{\underline{f}^2},
\end{align*}
which proves the first statement.

The second statement immediately follows since 
\[\sum_i w_i^2 (\tilde{q}_i - q_i)^2 \leq \|w\|_{\infty}\cdot\sum_i w_i (\tilde{q}_i - q_i)^2 \leq 2h^4 sB_wB_X^2 B_\beta B_K + 2B_w\frac{\epsilon_R^2}{\underline{f}^2}.\]
\end{proof}

\subsection{Proof of Theorem \ref{thm:b:consistent&sparse}}\label{pf:b}

Throughout the section, use
$\bh^{\lambda}$ to denote $\bh$ defined in 
Section \ref{algorithm}. In particular, 
$\bh$ is defined in \eqref{eq:eq_fixed:opt}, 
$\bh^{\lambda}$ is $\bh$ thresholded at level $\lambda$, i.e., 
$\bh^{\lambda}_j = \bh_j \cdot \One{\bh^2_j+\bh^2_{j+p} > \lambda^2}$, $1\leq j\leq p$, and 
$\bh^{\lambda}_j = \bh_j \cdot \One{\bh^2_j+\bh^2_{j-p} > \lambda^2}$, $p+1 \leq j\leq 2p$.

Let $S'= \supp(\bs)$.
By Assumption \ref{assumption:qr}, $|S'|\leq cs$ for some absolute constant $c$.
Therefore,
\begin{align*}
    \norm{\bh^{\lambda}-\bs}_{1,2} &\leq \norm{(\bh^{\lambda}-\bs)_{S'}}_{1,2} + \norm{(\bh^{\lambda})_{S'^c}}_{1,2}\\
     &\leq \norm{(\bh^{\lambda}-\bh)_{S'}}_{1,2} +\norm{(\bh-\bs)_{S'}}_{1,2}+ \norm{(\bh^{\lambda})_{S'^c}}_{1,2}\\
     &\leq cs\lambda +\norm{(\bh-\bs)_{S'}}_{1,2}+ \norm{(\bh^{\lambda})_{S'^c}}_{1,2}\\
     &\leq cs\lambda + \norm{\bh-\bs}_{1,2},
\end{align*}
where the third inequality comes from the definition of $\bh^{\lambda}$.
Furthermore, notice that $\norm{\bh^{\lambda}-\bs}_{1,2} \geq (\norm{\bh^{\lambda}}_{0,2} - |S'|) \lambda$.
Therefore we have
\begin{align*}
\norm{\bh-\bs}_{1,2} \geq \sbr{\norm{\bh^{\lambda}}_{0,2}-2cs}\lambda.
\end{align*}
Therefore, $\norm{\bh^{\lambda}}_{0,2}\leq 2cs + \norm{\bh-\bs}_{1,2}/\lambda.$ Because $\lambda = O(\sqrt{\frac{\log(np)}{nh}})$ and $\norm{\bh-\bs}_{1,2}=O_p\rbr{s\sqrt{\frac{\log(np)}{nh}}}$ from Theorem \ref{thm:qr_fixed:main}, $\norm{\bh^{\lambda}}_{0,2} \leq O_p(s)$. Now we have shown \eqref{eqn:thm_b_2} and \eqref{eqn:thm_b_3}.

To show \eqref{eqn:thm_b_1}, 
we first use the triangle inequality,
\begin{multline*}
\sqrt{\EEn{w_i \cdot \rbr{\Gamma_i^\top \rbr{\bh^{\lambda} - \bs}}^2}}\\
\leq \sqrt{\EEn{w_i \cdot \rbr{\Gamma_i^\top \rbr{\bh^{\lambda} - \bh}}^2}} + \sqrt{\EEn{w_i \cdot \rbr{\Gamma_i^\top \rbr{\bh - \bs}}^2}}.
\end{multline*}
Without loss of generality, we can order the components so that $(\bh^{\lambda}_j - \bh_j)^2 + (\bh^{\lambda}_{j+p} - \bh_{j+p})^2 $ is decreasing. Let $T_1$ be the set of $cs$ indices corresponding to the largest values of $(\bh^{\lambda}_j - \bh_j)^2 + (\bh^{\lambda}_{j+p} - \bh_{j+p})^2 $, similarly, let $T_k$ be the set of $cs$ indices corresponding to the largest values of $(\bh^{\lambda}_j - \bh_j)^2 + (\bh^{\lambda}_{j+p} - \bh_{j+p})^2 $ outside $\cup_{m=1}^{k-1} T_m$. By monotonicity, $\norm{(\bh^{\lambda}-\bh)_{T_k}}_{2}\leq \norm{(\bh^{\lambda}-\bh)_{T_{k-1}}}_{1,2}/\sqrt{cs}$. Then we have,
\begin{align*}
    & \sqrt{\EEn{w_i \cdot \rbr{\Gamma_i^\top \rbr{\bh^{\lambda} - \bh}}^2}} \\
    &= \sqrt{\sum_{k=1}^{\lceil \frac{p}{cs}\rceil}\EEn{w_i\rbr{\Gamma_i^\top \rbr{\bh^{\lambda} - \bh}_{T_k}}^2}}\\
    &\leq \sqrt{\EEn{w_i \cdot \rbr{\Gamma_i^\top \rbr{\bh^{\lambda} - \bh}_{T_1}}^2}} + \sqrt{\sum_{k=2}^{\lceil \frac{p}{cs}\rceil}\EEn{w_i\rbr{\Gamma_i^\top \rbr{\bh^{\lambda} - \bh}_{T_k}}^2}}\\
    &\leq \kappa_+\norm{(\bh^{\lambda}-\bh)_{T_1}}_{2} + \kappa_+\sum_{k=2}^{\lceil \frac{p}{cs}\rceil}\norm{(\bh^{\lambda}-\bh)_{T_k}}_{2}\\
    &\leq \kappa_+\norm{(\bh^{\lambda}-\bh)_{T_1}}_{2} + \kappa_+\sum_{k=1}^{\lceil \frac{p}{cs}\rceil}\norm{(\bh^{\lambda}-\bh)_{T_k}}_{1,2}/\sqrt{cs}\\
    &\leq \kappa_+\lambda + \kappa_+\norm{(\bh^{\lambda}-\bh)}_{1,2}/\sqrt{cs}\\
    &= O_p\rbr{\sqrt{\frac{s\log (np)}{nh}}}.
\end{align*}
In addition, from Theorem \ref{thm:qr_fixed:main}, $\EEn{w_i \cdot \rbr{\Gamma_i^\top \rbr{\bh - \bs}}^2} \leq O_p\rbr{\frac{s\log (np)}{nh}}$. Therefore, the first inequality holds.

\subsection{Proof of Theorem \ref{lem:qr_fixed:lasso:rate_v}}\label{pf:V}

  Our starting point is the basic inequality 
  \begin{multline}
  \label{eqn:Vbasic}
    \lambda_V  \rbr{
    \norm{\Vs}_{1,F} - \norm{\Vh}_{1,F}
    }
    \\
    \geq
      \textnormal{trace}\rbr{\frac{1}{2}\dv^\top \Hh(\db)\dv +
      \dv^\top\rbr{\Hh(\db) - H^\star}\Vs + \dv^\top\rbr{H^\star\Vs - E_a}},
  \end{multline}
 where $\db = \bh -\bs$ and $\dv = \Vh - \Vs$. 
 The above display can be rearranged as
 \begin{multline*}
   \textnormal{trace}\rbr{\frac{1}{2}\dv^\top \Hh(\db)\dv} \\
   \leq \lambda_V  \rbr{
    \norm{\Vs}_{1,F} - \norm{\Vh}_{1,F}
    } - \textnormal{trace}\rbr{\dv^\top\rbr{\Hh(\db) - H^\star}\Vs} \\
    -  \textnormal{trace}\rbr{\dv^\top\rbr{H^\star\Vs - E_a}}.
 \end{multline*}
Denote \begin{multline*}
D := B_A\left( 
   \sqrt{\frac{ \bar f  \cdot \log(p/\gamma)}{nhh_f} } \right. \\
 \left.+ \bar f' \rbr{ 2B_Kh_f + 2B_K\rbr{h^4s B_X B_\beta + \frac{\epsilon_R^2}{\underline{f}^2} }^{\frac{1}{2}} +  s\sqrt{\frac{B_K\log (np)}{nh}} }  \right),
\end{multline*} 
where $B_A$ is defined in Lemma~\ref{lem:V:qr_fixed:bound_approx}.
By Lemma~\ref{lem:V:qr_fixed:bound_approx}, 
with probability at least $1-2\gamma$,
\begin{equation*}
   \abr{\textnormal{trace}\rbr{\dv^\top\rbr{\Hh(\db) - H^\star}\Vs}} \leq \norm{\dv}_F \cdot D. 
\end{equation*}
By assumption \ref{assumption: lam}, 
\begin{equation*}
\abr{\textnormal{trace}\rbr{\dv^\top\rbr{H^\star\Vs - E_a}}} \leq \norm{\delta_v}_{1,F} \norm{H^\star \Vs -E_a}_{\infty, F} \leq \lambda^{\star} \norm{\delta_v}_{1,F}.
\end{equation*}
Since $\norm{\Vs}_{1,F} - \norm{\Vh}_{1,F} \leq \norm{(\delta_{v})_S}_{1,F} - \norm{(\delta_{v})_N}_{1,F}$,
we have
\begin{align*}
\lambda_V &\norm{(\delta_{v})_N}_{1,F} \\
&\leq \lambda_V \norm{(\delta_{v})_S}_{1,F}+ \abr{\textnormal{trace}\rbr{\dv^\top\rbr{\Hh(\db) - H^\star}\Vs}} + \abr{\textnormal{trace}\rbr{\dv^\top\rbr{H^\star\Vs - E_a}}}\\
&\leq \lambda_V \norm{(\delta_{v})_S}_{1,F}+ \norm{\dv}_F\cdot D + \lambda^{\star} \norm{\delta_v}_{1,F}\\
    &\leq \lambda_V \norm{(\delta_{v})_S}_{1,F}+ \norm{\dv}_F\cdot D + \frac{\lambda_V}{2} \norm{\delta_v}_{1,F}.
\end{align*}
Therefore,
\begin{align}\label{eqn:dv-basic}
\frac{\lambda_V}{2} \norm{(\delta_{v})_N}_{1,F} \leq \frac{3\lambda_V}{2} \norm{(\delta_{v})_S}_{1,F}+ \norm{\dv}_F \cdot D.
\end{align}
We consider two cases according to whether
$\frac{3\lambda_V}{2} \norm{(\delta_{v})_S}_{1,F} \geq \norm{\delta_v}_F \cdot D$ or not.

If 
\[
\frac{3\lambda_V}{2} \norm{(\delta_{v})_S}_{1,F} \geq \norm{\delta_v}_F \cdot D,
\]
then
\[
\norm{(\delta_{v})_N}_{1,F} \leq 6  \norm{(\delta_{v})_S}_{1,F}.
\]
Therefore, we have
\begin{align}
\label{eqn:delta_V_LB}
\norm{\Vh -\Vs}_{1,F} \leq 7\norm{(\Vh-\Vs)_S}_{1,F}\leq 7\sqrt{s_2}\norm{(\Vh - \Vs)_S}_F \leq 7\sqrt{s_2} \norm{\delta_v}_F.
\end{align}
On the other hand, from the basic inequality \eqref{eqn:Vbasic},
\begin{align*}
\lambda_V &\norm{\Vh -\Vs}_{1,F} \\
&\geq \textnormal{trace}\rbr{\frac{1}{2}\delta_v^\top \Hh(\db)\delta_v +
      \delta_v^\top\rbr{\Hh(\db) - H^\star}\Vs + \delta_v^\top\rbr{H^\star\Vs - E_a}}\\
&\geq \frac{\underline{f}\kappa^2_{-}}{2}\norm{\delta_v}_F^2 - o_p(1)\rbr{\norm{\dv}_F + \frac{\norm{\dv}_{1,F}}{\sqrt{s_2}}}^2- \frac{3\lambda_V}{2}\norm{(\delta_v)_S}_{1,F} - \lambda^{\star}\norm{\delta_v}_{1,F},
\end{align*}
where the second inequality above is because 
$\dv \in \mathbb{C}(S_2)$. Therefore, 
Assumption~\ref{assumption:X} holds and we
can apply Lemma \ref{lem:qr_fixed:lasso:quadratic_lower_bound}. 
Because $\lambda_V \geq 2\lambda^{\star}$, after rearrangement and combining with \eqref{eqn:delta_V_LB}, we get
\[
\norm{\dv}_F \leq \frac{24\lambda_V\sqrt{s_2}}{\underline{f}\kappa_{-}^2 -o_p(1)} = O_p\rbr{\sqrt{\frac{s\log p}{n h h_f}}}
\]
and
\[\norm{\dv}_{1,F} \leq 7\sqrt{s_2}\norm{\dv}_F = O_p\rbr{s\sqrt{\frac{\log p}{nhh_f}}}.\]

On the other hand, if 
\begin{equation} \label{eqn:V1}
\frac{3\lambda_V}{2} \norm{(\dv)_S}_{1,F} \leq  \norm{\dv}_F \cdot D, 
    \end{equation}
then, from \eqref{eqn:dv-basic}, we have
\begin{equation}\label{eqn:V2}
\frac{\lambda_V}{2} \norm{(\dv)_N}_{1,F} \leq 2 \norm{\dv}_F \cdot D .
\end{equation}
Therefore
\begin{align*}
    \frac{\lambda_V}{2} \norm{(\dv)_N}_{1,F} &\leq  2D \norm{\dv}_F \leq 2D \norm{\dv}_{1,F} = 2D \rbr{\norm{(\dv)_N}_{1,F} + \norm{(\dv)_S}_{1,F}}, 
\end{align*}
which implies
\begin{align*}    
   \norm{(\dv)_N}_{1,F} &\leq \frac{\frac{4D}{\lambda_V}}{1-\frac{4D}{\lambda_V}} \norm{(\dv)_S}_{1,F}. 
\end{align*}
From Assumption \ref{assumption: lam}, 
we can see that $B_A \lesssim B_X B_V$.
Then under Assumption \ref{assumption:growth}, 
$h \asymp n^{-1/3}$ and $h_f \asymp n^{-1/3}$, and 
by Assumption \ref{assumption:qr}, 
$\epsilon_R \asymp \sqrt{\frac{\log np}{nh}}$, we have 
\[
D \asymp \lambda^{\star} \asymp B_V\sqrt{\frac{\log p}{nhh_f}}.
\]
Therefore, there exists $\lambda_V > 2 \lambda^{\star}$
so that $\frac{4D}{\lambda_V}\leq \frac{6}{7}$.
With such a choice of $\lambda_V$, $\dv \in \mathbb{C}(S)$. 
Therefore, from \eqref{eqn:V1} and \eqref{eqn:V2},
\[\lambda_V \norm{\dv}_{1,F} \leq \frac{14}{3} \norm{\dv}_F D.\]
On the other hand, by applying Lemma \ref{lem:qr_fixed:lasso:quadratic_lower_bound} 
to the basic inequality \eqref{eqn:Vbasic}, we have
\[\lambda_V\norm{\dv}_{1,F} \geq \underline{f}\kappa_{-}^2 \norm{\dv}_F^2-o_p(1)\rbr{\norm{\dv}_F + \frac{\norm{\dv}_{1,F}}{\sqrt{s_2}}}^2- 2\norm{\dv}_F D.  \]
Combining the two we have 
\[\frac{20}{3}D\norm{\dv}_F \geq \underline{f}\kappa_{-}^2 \norm{\dv}_F^2-o_p(1)\rbr{\norm{\dv}_F + \frac{14D}{3\lambda_V\sqrt{s_2}}\norm{\dv}_F }^2.\]
Because $D \asymp \lambda_V = O(B_V \sqrt{\frac{\log p}{nhh_f}})$, we have
\[\norm{\dv}_F \leq \frac{20}{3} \frac{D}{\underline{f}\kappa_-^2 -o_p(1)}= O_p\rbr{B_V\sqrt{\frac{s\log p}{nhh_f}}} \]
and
\[\norm{\dv}_{1,F} \leq \frac{14D}{3\lambda_V} \sqrt{s_2} \norm{\dv}_F = O_p\rbr{s B_V \sqrt{\frac{\log p}{nhh_f}}}.\]

To complete the proof, we need to establish a few technical lemmas next.

\begin{lemma}
\label{lem:V:qr_fixed:bound_approx}
Suppose that the growth conditions in
Assumption~\ref{assumption:growth} is satisfied, 
and $r_b \asymp s\sqrt{\frac{\log (np)}{nh}}$. 
Define $B_A = \max_l \norm{A_{il}}_F$ and $A_{il} = \Gamma_{il}\Gamma_i^{\top}\Vs$ where $l \in [p]$ and $\Gamma_{il} = (\Gamma_{i,l}, \Gamma_{i,l+p})^{\top}$.
For any $\gamma > 0$ such that $\bar f \geq (nhh_f)^{-1} \log(p/\gamma)$ and $r_b = O\rbr{h_f \log(p/\gamma) / B_X}$,
we have
 \begin{align*}
\max_{l \in [p]} &\sup_{\substack{\norm{\db}_{0,2} \leq m \\ \norm{\db}_{1,2} \leq r_b }}  \Norm{\EEn{w_i (\hat{f}_i (\db) - f_i(\tilde{q}_i)) \Gamma_{il} \Gamma_i^{\top}\Vs }}_{F} \\
     & \leq 
    B_A\left(
   \sqrt{\frac{ \bar f  \cdot \log(p/\gamma)}{nhh_f} }  \right. \\ 
 & \qquad \left. + \bar f' \rbr{ 2B_Kh_f + 2B_K\rbr{h^4s B_X B_\beta + \frac{\epsilon_R^2}{\underline{f}^2} }^{1/2} +  s\sqrt{\frac{\log (np)}{nh}}\sqrt{B_K} }  \right)
 \end{align*}
 with probability $1-2\gamma$.
\end{lemma}

\begin{proof}
We have 
 \begin{align*}
& \max_{l \in [p]}  \sup_{\substack{\norm{\db}_{0,2} \leq m \\ \norm{\db}_{1,2} \leq r_b }}\Norm{\EEn{w_i (\hat{f}_i (\db) - f_i(\tilde{q}_i)) \Gamma_{il} \Gamma_i^{\top}\Vs }}_{F}\\
\leq & \max_{l \in [p]} \sup_{\substack{\norm{\db}_{0,2} \leq m \\ \norm{\db}_{1,2} \leq r_b }} \Bigg[ \Norm{\Gn{w_i (\hat{f}_i (\db) - \hat f_i(0)) \cdot A_{i l} }}_{F} \\
& + \Norm{\Gn{w_i \cdot (\hat f_i(0)-f_i(\tilde{q}_i)) \cdot A_{i l} }}_{F} + \Norm{\bE{w_i (\hat{f}_i (\db) - \hat{f}_i (0)) \Gamma_{il} \Gamma_i^{\top}\Vs }}_{F} \Bigg]\\
& + \Norm{\bE{w_i (\hat{f}_i (0) - f_i(\tilde{q}_i)) \Gamma_{il} \Gamma_i^{\top}\Vs }}_{F} \Bigg]\\
\leq & \max_{l \in [p]} \sup_{\substack{\norm{\db}_{0,2} \leq m \\ \norm{\db}_{1,2} \leq r_b }} \Bigg[ \Norm{\Gn{w_i (\hat{f}_i (\db) - \hat f_i(0)) \cdot A_{i l} }}_{F} \\
& + \Norm{\Gn{w_i \cdot \hat f_i(0) \cdot A_{i l} }}_{F} + \Norm{\bE{w_i (\hat{f}_i (\db) - \hat{f}_i (0)) \Gamma_{il} \Gamma_i^{\top}\Vs }}_{F} \Bigg]\\
 & + \Norm{\bE{w_i (\hat{f}_i (0) - f_i(\tilde{q}_i)) \Gamma_{il} \Gamma_i^{\top}\Vs }}_{F} \Bigg]\\
 \leq & B_A\left(
   \sqrt{\frac{ \bar f  \cdot \log(p/\gamma)}{nhh_f} } \right.\\
 & \quad \left. + \bar f' \rbr{ B_Kh_f + 2B_K\rbr{h^4s B_X B_\beta + \frac{\epsilon_R^2}{\underline{f}^2} }^{1/2} +  s\sqrt{\frac{\log (np)}{nh}}\sqrt{B_K} + B_Kh_f}  \right),
\end{align*}
where the last inequality follows by first combining Lemmas \ref{lem:V:qr_fixed:score_bound}, \ref{lem:4},  \ref{lem:new} and \ref{lem:5}
and plugging in our condition for $r_b$, $h$ and $h_f$. 
\end{proof}

\begin{lemma}
    \label{lem:V:qr_fixed:score_bound}
    Under the conditions of Lemma~\ref{lem:V:qr_fixed:bound_approx}, we have 
  \begin{multline*}
      \max_{l \in [p]} \sup_{\substack{\norm{\db}_{0,2} \leq m \\ \norm{\db}_{1,2} \leq r_b }} \Norm{\Gn{w_i\cdot\rbr{\hat f_i(\db) - \hat f_i(0)}\cdot A_{i l}}}_F
   \\ \lesssim
   \frac{ B_K B_A }{ h_f }
    \sqrt{ \bar f B_X  \frac{r_b \rbr{m \log p + \log(1/\gamma)} }{nh}}
  \end{multline*}
  with probability $1-\gamma$, where $A_{il}$ and $B_A$ are 
  defined in Lemma \ref{lem:V:qr_fixed:bound_approx}.
\end{lemma}
\begin{proof}
  Let ${\cal W} = \cbr{\tilde{W}_1, \ldots, \tilde{W}_K}$ be the $\frac 12$-net
  for $\cbr{W \in \R^{2k \times 2k} \mid \norm{W}_F \leq 1}$. We have that $K \leq 5^{4k^2}$ and
  \begin{multline*}
    \max_{l \in [p]} \sup_{ \substack{\norm{\db}_{0,2} \leq m \\ \norm{\db}_{1,2} \leq r_b } } \Norm{\Gn{w_i \cdot \rbr{\hat f_i(\db) - \hat f_i(0)}\cdot A_{i l}}}_F \\
    \leq 2 \cdot \max_{\tilde{W} \in {\cal W}}
      \max_{l \in [p]} \sup_{ \substack{\norm{\db}_{0,2} \leq m \\ \norm{\db}_{1,2} \leq r_b } }  {\Gn{w_i\cdot(\hat f_i(\db) - \hat f_i(0))\cdot \tr{\rbr{\tilde{W}^\top A_{i l}}}}}.
    \end{multline*}
Our goal is to apply Lemma~\ref{lm:proc_deviation} to bound the right hand side.

Note that
\begin{align*}
  2h_f\rbr{\hat f_i(\db) - \hat f_i(0)}
  & =  \One{ h_f < y_i - \Gamma_i^\top \bs \leq h_f + \Gamma_i^\top \db}  \\
  &\quad - \One{ h_f + \Gamma_i^\top \db < y_i - \Gamma_i^\top \bs \leq h_f } \\
  &\quad - \One{ - h_f \leq y_i - \Gamma_i^\top \bs < -h_f + \Gamma_i^\top \db }  \\
  &\quad + \One{ - h_f + \Gamma_i^\top \db \leq y_i - \Gamma_i^\top \bs < -h_f  }.
\end{align*}
We proceed to bound
\[
  \max_{\substack{\tilde{W} \in {\cal W}\\ l \in [p]\\|S| \leq m}}  \sup_{ \substack{\supp\rbr{\db} = S \\ \norm{\db}_{1,2} \leq r_b } }
  \Gn{\frac{w_i}{2h_f}\cdot \One{ h_f < y_i - \Gamma_i^\top \bs \leq h_f + \Gamma_i^\top \db} \cdot |\tr{\rbr{\tilde{W}^\top A_{i l}}}|},
\]
while the other terms are bounded similarly.
For a fixed $\tilde{W} \in {\cal W}$, $l \in [p]$ and $|S| \leq m$, define
\begin{align*}
a_i &= \frac{w_i}{2h_f} \cdot \tr{\rbr{\tilde{W}^\top A_{i l}}}, \ \text{and} \\
{\cal G}_S &=
\Big\{(y_i, x_i, u_i) \mapsto a_i \cdot \One{ h_f < y_i - \Gamma_i^\top \bs \leq h_f + \Gamma_i^\top \db} : \\
& \qquad\qquad\qquad\qquad\qquad\qquad\qquad\qquad\qquad \supp\rbr{\db} = S, \norm{\db}_{1,2} \leq r_b \Big\},\\
{\cal G} &= \cup_{S:|S|\leq m} {\cal G}_S
\end{align*}
Let $G(\cdot)$ be an envelope of ${\cal G}$ and note that $\norm{G}_{\infty} \leq \frac{B_K B_A}{nhh_f}$.
For a fixed $g \in {\cal G}$, let
\[
g_i = g(y_i, x_i, u_i) = a_i \cdot \One{ h_f < y_i - \Gamma_i^\top \bs \leq h_f + \Gamma_i^\top \db}.
\]
We have that
\begin{align*}
  \EE{
     \One{ h_f < y_i - \Gamma_i^\top \bs \leq h_f + \Gamma_i^\top \db}}
& = {
F_i\rbr{\Gamma_i^\top \bs + h_f + \Gamma_i^\top \db} - F_i\rbr{\Gamma_i^\top \bs + h_f} }\\
  & \leq \bar f \cdot  \abr{\Gamma_i^\top \db} \\
  & \leq \bar f B_X r_b,
\end{align*}
and, therefore, the variance is bounded as
\begin{align*}
  \sigma_{\cal G}^2
  \leq
  \sup_{g \in {\cal G}} \sum_{i \in [n]} \EE{g_i^2} \leq
  \bar f B_X r_b \sum_{i \in [n]} a_i^2
  \leq \bar f B_X  B_K^2 B_A^2\cdot \frac{r_b}{nhh_f^2},
\end{align*}
since
\[
  \sum_{i \in [n]} a_i^2
  \leq \frac{B_A^2}{h_f^2} \sum_{i \in [n]} w_i^2
  \leq \frac{B_K^2 B_A^2}{nhh_f^2}.
\]

The VC dimension for the space \[{\cal F}_S =
\cbr{(y_i, x_i, u_i) \mapsto \cdot \One{ h_f < y_i - \Gamma_i^\top \bs \leq h_f + \Gamma_i^\top \db} :  \supp\rbr{\db} = S, \norm{\db}_{1,2} \leq r_b}\] 
is $|S|\leq m$. Therefore, applying Lemma \ref{lem:VC_to_cover}
and Lemma \ref{lem:entropy_bound_product},   
\[\sup_Q \log N\rbr{\epsilon\cdot \frac{B_KB_A}{nhh_f}, {\cal G}_S, \norm{\cdot}_{L_2(Q)}} 
\lesssim m \log(1/ \epsilon). \]
Since there are $p \choose m$ different supports $S$ in ${\cal G}$, we have
\[\sup_Q \log N\rbr{\epsilon\cdot \frac{B_KB_A}{nhh_f}, {\cal G}, \norm{\cdot}_{L_2(Q)}} 
\lesssim m(\log(1/\epsilon) + \log(p)). \]
Applying  Lemma~\ref{lem:emp_proc:supremum_vc}
with $\sigma_{\cal G}= \frac{B_K B_A}{h_f} \sqrt{\frac{\bar f r_b B_X}{nh}}$,
$\norm{G}_{\infty} \leq \frac{B_K B_A}{nhh_f}$,
$V = cm$, and $A = C  p^{1/c}$, we have
\[
\begin{aligned}
\EE{ \sup_{g \in {\cal G}} \sum_{i \in [n]} g_i - \EE{g_i}}
& \lesssim \rbr{m \frac{B_K B_A}{nhh_f}\log \frac{p}{\sqrt{\bar f r_b B_X h}} + \frac{B_K B_A}{h_f} \sqrt{\frac{\bar f r_b B_X}{nh}} \sqrt{m \log \frac{p}{\sqrt{\bar f r_b B_X h}} } } \\
& \lesssim \frac{ B_K B_A }{ h_f }\sqrt{ \bar f B_X \frac{mr_b\log p}{nh} }
\end{aligned}
\]
where the last inequality follows from the conditions on $r_b$ in Lemma \ref{lem:V:qr_fixed:bound_approx} 
and Assumption \ref{assumption:growth}.
Finally, Lemma~\ref{lm:proc_deviation} gives us
\begin{align*}
  \sup_{ g \in {\cal G} }
  \sum_{i \in [n]} g_i - \EE{g_i}
   \lesssim
  \frac{ B_K B_A }{ h_f }
    \sqrt{ \bar f B_X  \frac{r_b \rbr{m \log p + \log(1/\gamma)} }{nh}},
\end{align*}
with probability $1-\gamma$, and, 
by the union bound over $\tilde W \in {\cal W}$, $l \in [p]$,
\begin{multline*}
 \max_{\substack{\tilde{W} \in {\cal W}\\ l \in [p]\\|S| \leq m}} 
  \sup_{ \substack{\supp\rbr{\db} = S \\ \norm{\db}_{1,2} \leq r_b } }  {\Gn{w_i\cdot\One{ h_f < y_i - \Gamma_i^\top \bs \leq h_f + \Gamma_i^\top \db}\cdot \tr{\rbr{\tilde{W}^\top A_{i l}}}}} \\
  \lesssim
 \frac{ B_K B_A }{ h_f }
    \sqrt{ \bar f B_X  \frac{r_b \rbr{2 m \log p + \log(5^{4k^2}/\gamma)} }{nh}}.
\end{multline*}
Handling other terms in the same way, we obtain 
\begin{multline*}
    \max_{l \in [p]} \sup_{ \substack{\norm{\db}_{0,2} \leq m \\ \norm{\db}_{1,2} \leq r_b } } \Norm{\Gn{w_i \cdot \rbr{\hat f_i(\db) - \hat f_i(0)}\cdot A_{i l}}}_F \\
   \lesssim \frac{ B_K B_A }{ h_f }
    \sqrt{ \bar f B_X  \frac{r_b \rbr{m \log p + \log(1/\gamma)} }{nh}},
\end{multline*}
with probability $1-\gamma$, which completes the proof.

\end{proof}

\begin{lemma}
\label{lem:4}
  Under the conditions of Lemma~\ref{lem:V:qr_fixed:bound_approx}, we have 
  \[
    \max_{l \in [p]} \Norm{\Gn{w_i\cdot \hat f_i(0) \cdot A_{i l}}}_F \leq  \sqrt{\bar f B_K^2B_A^2 \frac{ \log(p/\gamma)}{nhh_f} }
  \]
  with probability $1-\gamma$.
\end{lemma}
\begin{proof}
  Let ${\cal W}$ be as in the proof of Lemma~\ref{lem:V:qr_fixed:score_bound}. Then
  \begin{multline*}
    \max_{l \in [p]}\Norm{\Gn{w_i \cdot \hat f_i(0) \cdot A_{i l}}}_F \\
    \leq 2 \max_{\tilde{W} \in {\cal W}}
      \max_{l }  {\Gn{w_i\cdot \hat f_i(0) \cdot \tr{\rbr{\tilde{W}^\top A_{i l}}}}}.
    \end{multline*}
    Let $Z_i = w_i\cdot \hat f_i(0) \cdot \tr{\rbr{\tilde{W}^\top A_{i l}}}$. Then
    \begin{equation*}
      \sum_{i\in[n]} \EE{Z_i^2}
      \leq \sum_{i\in[n]} \frac{w_i^2}{h_f^2} \cdot \tr{\rbr{\tilde{W}^\top A_{i l}}}^2 \cdot \EE{\One{\abr{y_i - \Gamma_i^\top\bs}\leq h_f}}
      = O\rbr{ \frac{\bar f B_K^2B_A^2}{nhh_f} }
    \end{equation*}
    and
   \begin{equation*}
     \max_{i \in [n]} |Z_i| \leq O\rbr{ \frac{\bar f B_K B_A}{nhh_f}}.
    \end{equation*}
    The result follows from Lemma~\ref{lm:proc_deviation} and the union bound.
\end{proof}

\begin{lemma}
\label{lem:new}
  Suppose conditions of Lemma~\ref{lem:V:qr_fixed:bound_approx} hold. Then
  \[
    \max_{l \in [p]} \Norm{E_l^\top \rbr{\EE{\Hh(\db)} - \EE{\Hh(0)}}\Vs}_F
    \leq
    \bar f' B_A \rbr{s\sqrt{\frac{B_K\log (np)}{nh}} + B_Kh_f}.
  \]
\end{lemma}
\begin{proof}
  For a fixed $\delta_b$, the mean value theorem gives us
  \[
    \begin{aligned}
    &2h_f \abr{\EE{\hat f_i(\db)}-\EE{\hat f_i(0)}}\\
     = ~& \left|F_i\rbr{\Gamma_i^\top(\bs+\db) + h_f} - F_i\rbr{\Gamma_i^\top(\bs+\db) - h_f}-F_i\rbr{\Gamma_i^\top\bs + h_f} + F_i\rbr{\Gamma_i^\top\bs - h_f}\right| \\
    \leq ~ & 2h_f\bar f'(\abr{\Gamma_i^\top\db}+h_f),
  \end{aligned}
\]
 Therefore, we have that
  \begin{align*}
    & \max_{l\in [p]}\Norm{E_l^\top \rbr{\EE{\Hh(\db)} - \EE{\Hh(0)}}\Vs}_F \\
    & \qquad
      \leq \bar f' B_A \sum_{i \in [n]} w_i \cdot \rbr{\abr{\Gamma_i^\top\db}+ h_f} \\
    & \qquad
      \leq \bar f' B_A
      \rbr{\sqrt{\sum_{i \in [n]} w_i\rbr{\Gamma_i^\top\db}^2}\sqrt{\sum_{i \in [n]} w_i}+ h_f\sum_{i \in [n]} w_i 
      }    
         \\
    &\qquad
      \leq \bar f' B_A \rbr{s\sqrt{\frac{\log (np)}{nh}}\sqrt{B_K} + B_Kh_f} ,
  \end{align*}
  where the last inequality follows from Lemma~\ref{lem:qr_fixed:approx_err} and Assumption \ref{assumption:X}.

\end{proof}

\begin{lemma}
\label{lem:5}
  Suppose conditions of Lemma~\ref{lem:V:qr_fixed:bound_approx} hold. Then
  \[
    \max_{l \in [p]} \Norm{E_l^\top \rbr{\EE{\Hh(0)} - H^\star}\Vs}_F
    \leq
    \bar f' B_A \rbr{ B_Kh_f + 2\rbr{h^4s B_X B_\beta + \frac{\epsilon_R^2}{\underline{f}^2} }^{1/2} }.
  \]
\end{lemma}
\begin{proof}
  The mean value theorem gives us
  \[
    \begin{aligned}
    2h_f \EE{\hat f_i(0)}
    & = F_i\rbr{\Gamma_i^\top\bs + h_f} - F_i\rbr{\Gamma_i^\top\bs - h_f} \\
    & = 2h_f\cdot f_i(\fq) + 2h_f\cdot \rbr{f_i(\check q_i) - f_i(\fq)},
  \end{aligned}
\]
where $\check \fq$ is a point between $\Gamma_i^\top\bs - h_f$ and
$\Gamma_i^\top\bs  + h_f$.
Therefore, we have
\[
  \abr{\EE{\hat f_i(0)} - f_i(q_i)}
  \leq \bar f' \abr{\check \fq - \fq}
  \leq \bar f' \rbr{  \abr{\fql - \fq} + h_f}
  , ~\textnormal{and}
\]
\[
  \abr{f_i(\tilde{q}_i) - f_i(q_i)}
  \leq \bar f' \abr{\fql - \fq}.
\]
 Finally, we have
  \begin{align*}
    & \max_{l\in [p]}\Norm{E_l^\top \rbr{\EE{\Hh(0)} - H^\star}\Vs}_F \\
    & \qquad
      \leq \bar f' B_A \sum_{i \in [n]} w_i \cdot (2\abr{\fql - \fq} + h_f) \\
    & \qquad
      \leq \bar f' B_A
      \rbr{h_f \sum_{i \in [n]} w_i +
      2\rbr{ \sum_{i \in [n]} w_i }^{1/2} \cdot
      \rbr{
      \sum_{i \in [n]} w_i \rbr{\fql - \fq}^2 
      }^{1/2}
      }    
         \\
    &\qquad
      \leq \bar f' B_A B_K \rbr{  h_f + 2\rbr{h^4s B_X^2 B_\beta + \frac{\epsilon_R^2}{\underline{f}^2} }^{1/2} },
  \end{align*}
  where the last inequality follows from Lemma~\ref{lem:qr_fixed:approx_err} and Assumption \ref{assumption:X}.

\end{proof}

\begin{lemma}
  \label{lem:qr_fixed:lasso:quadratic_lower_bound}
  Assume $B_V$ satisfies Assumption \ref{assumption: lam}, that $(nh)^{-1}sB_V^2\log p = o(1)$ and $\log(B_V^2h_fh) = O(\log p)$.
  Let $S_1$ be the support of $\bs$ and $S_2$
  be the support of $\Vs$ as defined in Assumption \ref{assumption:X},
  with $|S_1|=s_1$ and $|S_2|=s_2$. Define
  \begin{align*}
  \Delta_b(r_b, s_1) &= \{\delta \in \R^{2p}, \norm{\delta}_2\leq r_b, \norm{\delta}_0\leq s_1\} \quad \textnormal{and}\\
  \mathbb{C}(S) &= \{\Theta \in \R^{2k\times 2p}: \norm{\Theta_{S^c}}_{1,F} \leq 6 \norm{\Theta_{S}}_{1,F} \} .
  \end{align*}
  Then
  \[
    \textnormal{trace}\rbr{\dv^\top \Hh(\db)\dv } \geq \underline{f}\kappa_-\norm{\dv}_F^2 - o_p(1)\rbr{\norm{\dv}_F + \frac{\norm{\dv}_{1,F}}{\sqrt{s_2}}}^2
  \]
  for all $\db \in \Delta_b(r_b, s_1)$ and $\dv \in \mathbb{C}(S_2)$.
\end{lemma}
\begin{proof}
  For a fixed $\db$, we have that $f_i(\db) \geq \underline{f}$ and $\dv \in \mathbb{C}(S_2)$.
  Therefore,
  \[
      \textnormal{trace}\rbr{\dv^\top H(\db) \dv }
      \geq \underline{f} \sum_{i \in [n]} w_i \textnormal{trace}(\dv^\top\Gamma_i\Gamma_i^\top \dv)
      \geq \underline{f}\kappa_{-}^2\norm{\dv}_F^2.
  \]
  The proof now follows from Lemma~\ref{Cunhui Zhang} and
  Lemma~\ref{lem:qr_fixed:lasso:quadratic_lower_bound:deviation} (presented next),
  \begin{align*}
   &\textnormal{trace}\rbr{\dv^\top \Hh(\db)\dv }\\
    &\ \geq \textnormal{trace}\rbr{\dv^\top H(\db) \dv } - \textnormal{trace}\rbr{\dv^\top \rbr{\Hh(\db) - H(\db)} \dv}\\
    &\ \geq
    \textnormal{trace}\rbr{\dv^\top H(\db) \dv } \\
   &\quad -
    \sup_{\delta \in \Delta_v(s_2)} \sup_{\db \in \Delta_b(r_b, s_1)}
    \abr{ \textnormal{trace}\rbr{\delta^\top \rbr{\Hh(\db) - H(\db)} \delta } }\rbr{\norm{\dv}_F + \frac{\norm{\dv}_{1,F}}{\sqrt{s_2}}}^2 \\
    &\ \geq \underline{f}\kappa_{-}^2 \norm{\dv}_F^2 - o_p(1)\rbr{\norm{\dv}_F + \frac{\norm{\dv}_{1,F}}{\sqrt{s_2}}}^2,
 \end{align*}
 where $\Delta_v(s_2)$ is as defined in Lemma~\ref{lem:qr_fixed:lasso:quadratic_lower_bound:deviation}. 
\end{proof}

\begin{lemma}
\label{Cunhui Zhang}
(Based on proposition 5 in \citep{Sun2012Sparse}). 
For any fixed matrix $M \in R^{p\times p}$ and matrices $u \in \R^{k\times p}$ and $s \in \mathbb{N}$, 
\[
\textnormal{trace}\rbr{u^{\top}M u} \leq \rbr{\norm{u}_F + \frac{\norm{u}_{1,F}}{\sqrt{s}}}^2 \norm{M}_{\mathcal{S}_{s}},
\]
where $\mathcal{S}_{s} = \{u \in \R^{k\times p }| \norm{u}_F = 1, \norm{u}_{0,F} \leq s\}$ and
\[\norm{M}_{\mathcal{S}_s} = \max_{u,v \in \mathcal{S}_s} \textnormal{trace}\rbr{u^{\top}M v}.\]
\end{lemma}

\begin{lemma}
  \label{lem:qr_fixed:lasso:quadratic_lower_bound:deviation}
  Under the conditions of Lemma~\ref{lem:qr_fixed:lasso:quadratic_lower_bound}, we have
  \begin{multline*}
  \sup_{\dv \in \Delta_v(s_2)}  \sup_{\db \in \Delta_b(r_b, s_1)}
	\abr{
    \textnormal{trace}\rbr{\dv^\top \rbr{\Hh(\db) - H(\db)} \dv } } \\
    = O_p\left(
    \sqrt{\frac{(s_1+s_2)\bar{f} \kappa_+ B_K B_V^2\log(p)}{nhh_f}}    \right),
  \end{multline*}
  where
  \[\Delta_v(s_2) = \{\delta \in \R^{2k\times 2p}, \norm{\delta}_F =1 , \norm{\delta}_{0,F} \leq s_2\},\]
  \[\Delta_b(r_b, s_1) = \{\delta \in \R^{2p}, \norm{\delta}_2\leq r_b, \norm{\delta}_0\leq s_1\}.\]
\end{lemma}
\begin{proof}
 We have 
  \begin{align*}
  &\sup_{\dv \in \Delta_v(s_2)} \sup_{\db \in \Delta_b(r_b, s_1)}
	\abr{
  \textnormal{trace}\rbr{\dv^\top \rbr{\Hh(\db) - H(\db)} \dv } } \\
   &\ =  \sup_{\dv \in \Delta_v(s_2)} \sup_{\db \in \Delta_b(r_b, s_1)}(\mathbb{S}_n - \mathbb{ES}_n)\cbr{w_i (2h_f)^{-1} \cdot \One{\abr{y_i - \Gamma_i^\top(\bs +\db)}\leq h_f} }\cdot \norm{\dv \Gamma_i}_2^2 \\
   &\ \leq 2 \sup_{\dv \in {\cal N}_\epsilon} \sup_{\db \in \Delta_b(r_b, s_1)}(\mathbb{S}_n - \mathbb{ES}_n)\cbr{w_i (2h_f)^{-1}\cdot \One{\abr{y_i - \Gamma_i^\top(\bs +\db)}\leq h_f} }\cdot \norm{\dv \Gamma_i}_2^2,
  \end{align*}
  where ${\cal N}_\epsilon$ is an $\epsilon$-net for $\Delta_v(s_2)$.
  We have $|{\cal N}_\epsilon| \leq 5^{2ks_2}$.
  Fix $\dv \in {\cal N}_\epsilon$.
  Define
 \begin{align*}
 a_i &= (2h_f)^{-1}w_i \cdot \textnormal{trace}\cbr{\dv \Gamma_i\Gamma_i^{\top}\dv^{\top}},\\
 {\cal G}_S &= \Big\{ (y_i, x_i, u_i) \mapsto a_i\cdot \One{\abr{y_i - \Gamma_i^\top(\bs +\db)}\leq h_f}
 :\supp(\db)=S,\norm{\db}_2\leq r_b\Big\},\\
 {\cal G} &= \cup_{S:|S|\leq s_1}{\cal G}_{S}.\end{align*}
Let $G(\cdot) = \frac{B_K B_V^2}{2nhh_f}$ be an
envelope of ${\cal G}$.
For a fixed $g \in {\cal G}$, let
\[g_i = g(y_i, x_i, u_i) = a_i\cdot \One{\abr{y_i - \Gamma_i^\top(\bs +\db)}\leq h_f}. \]
Therefore, the variance is bounded as
\[\sigma_{{\cal G}}^2 
\leq \sup_{g \in {\cal G}} \sum_{i \in [n]} \EE{g_i^2} 
\lesssim \bar{f} h_f \sum_{i \in [n]} (4h_f^2)^{-1}w_i^2\cdot \norm{\dv\Gamma_i}_2^4 
\lesssim \frac{\bar{f} \kappa_+ B_K B_V^2}{nhh_f}.\]
The VC dimension for the space ${\cal G}_S$ is 
$O(|S|)$. Therefore, using Lemma \ref{lem:VC_to_cover},  
\[\sup_Q \log N(\epsilon, {\cal G}, \norm{\cdot}_{L_2(Q)}) \lesssim s_1\rbr{\log(p) + \log(1/\epsilon)}.\]
Applying Lemma~\ref{lem:emp_proc:supremum_vc} 
we obtain
\begin{align*}
 \EE{\sup_{g \in {\cal G}} \sum_{i \in [n]} g_i - \EE{g_i}}
 \lesssim
 %\frac{s_1B_K B_V^2\log(n)}{nh} +
 \sqrt{\frac{s_1\bar{f} \kappa_+ B_K B_V^2\log(p)}{nhh_f}},
\end{align*}
under our assumptions. 
Using Lemma \ref{lm:proc_deviation},
\begin{equation*}
    \sup_{g \in {\cal G}} \sum_{i \in [n]} g_i - \EE{g_i} 
    = O_p\rbr{ \sqrt{\frac{s_1B_V^2\log (p)}{nhh_f}} }.
\end{equation*}
A union bound over ${\cal N}_\epsilon$ completes the proof.
\end{proof}

\subsection{Empirical Process Results}\label{pf:ep}

\begin{definition}
The covering number $N(\epsilon, \cF, \norm{\cdot})$ is the minimal number of balls $\cbr{g \mid \norm{g-f}\leq\epsilon}$ of radius $\epsilon$ needed to cover the set $\cF$.

\end{definition}

Let $\|\cF\|_{\infty} = \sup\{\|f\|_{\infty}, f\in \cF\}$. Furthermore, define
\begin{equation}
\label{eq:empirical_process_variance}
\Sigma^2_\cF = \EE{\sup_{f \in \cF} \sum_{i\in[n]} W_i^2(f)}
\quad\text{and}\quad
\sigma^2_\cF = \sup_{f \in \cF} \sum_{i\in[n]} \EE{W_i^2(f)}
\end{equation}
where $W_i(f)$, $f\in\cF$, $i\in[n]$ are real valued random variables.

\begin{lemma}
  \label{lem:emp_proc:supremum_vc}
  Let ${\cal F}$ be a measurable uniformly bounded class of functions
  satisfying 
  \[
  N\rbr{ \epsilon \norm{F}_{L_2(P)}, {\cal F}, L_2(P) } \leq \rbr{\frac{A}{\epsilon}}^V
  \]
  for all probability measures $P$, where $F := \sup_{f \in {\cal F}}|f|$ 
  is the envelope function and $A$, $V$ are constants dependent on ${\cal F}$.
  Let $\sigma_{\cal F}^2 = \sup_{f \in {\cal F}} \sum_{i \in n}\EE{ \rbr{f_i - \EE{f_i}}^2}$ and
  $U \geq \sup_{f \in {\cal F}} \norm{f}_{\infty}$ be such that $0  < \sigma_{\cal F} \leq \sqrt{n} U$.
  Then there exists a universal constant $C$ such that
  \[
    \EE{\sup_{f \in {\cal F}} \sum_{i \in [n]} f_i - \EE{f_i}}
    \leq C \sbr{V U \log \frac{\sqrt{n} AU}{\sigma_{\cal F}} + \sigma_{\cal F}\sqrt{V\log\frac{\sqrt{n} AU}{\sigma_{\cal F}}}}.
  \]
\end{lemma}
\begin{proof}
 This is essentially Proposition 2.1 of \cite{Gine2001consistency} combined with symmetrization \cite{Koltchinskii2010Sparsity}.
\end{proof}

\begin{lemma}[Theorem 2.6.7 of \cite{vanderVaart1996Weak}] 
\label{lem:VC_to_cover}Suppose $\cF$ is a function class with a bounded 
VC-dimension, V, and an envelope $F$. Then there exist absolute constants $c,C>0$ such that 
\[
\sup_Q N\rbr{\epsilon\norm{F}_{Q,2}, \cF, \norm{\cdot}_{L_2(Q)}} \leq \rbr{\frac{C}{\epsilon}}^{cV}
\]
for all $\epsilon \in (0,1)$ and the probability measure $Q$ ranges over distributions such that $\norm{F}_{Q,2} > 0$.
\end{lemma}

\begin{lemma}[Lemma 22 of \cite{Nolan1987Uprocess}]
  \label{lm:covering_kernel}  
  Let $K: \R \mapsto \R$ be a bounded variation function. The function class
  \begin{align} 
  	\label{eq:kernel_function_class}
    \cK = \cbr{K\rbr{\frac{s - \cdot}{h}} \mid h > 0, s \in \R },
  \end{align}
  indexed by the kernel bandwidth and the location $s$, satisfies the uniform entropy condition 
  \begin{equation}
    \label{eq:uniform_entropy}
     \sup_Q N(\epsilon, \cK, \|\cdot\|_{L_2(Q)})
    \leq C \epsilon^{-v},\quad \text{for all }\epsilon \in (0,1),
  \end{equation}
  for some $C > 0$ and $v > 0$. 
\end{lemma}

\begin{lemma}[Lemma 26  of \cite{Lu2015Posta}]
  \label{lem:entropy_bound_product}
  Let $\cF$ and $\cG$ be two function classes satisfying 
  \[
  \sup_Q N(a_1\epsilon, \cF, \| \cdot\|_{L_2(Q)}) \leq C_1 \epsilon^{-v_1}  
  \qquad\text{ and } \qquad
  \sup_Q N(a_2\epsilon, \cG, \| \cdot\|_{L_2(Q)}) \leq C_2 \epsilon^{-v_2}
  \]
  for
  some $C_1, C_2, a_1, a_2, v_1, v_2 > 0$ and any $0 < \epsilon < 1$.
  Define
  \[
    \cF_{\times} = \{fg \mid f \in \cF, g \in \cG \}
  \quad\text{ and }\quad
    \cF_{+} = \{f+g \mid f \in \cF, g \in \cG \}.
  \]
  Then for any $\epsilon \in (0,1)$,
  \begin{align*}
  \sup_Q N(\epsilon, \cF_{\times}, \| \cdot\|_{L_2(Q)}) 
    &\leq C_1C_2 \rbr{\frac{2a_1U}{\epsilon}}^{v_1}
                \rbr{\frac{2a_2U}{\epsilon}}^{v_2}\\
\intertext{and}
    \sup_Q  N(\epsilon, \cF_{+}, \| \cdot\|_{L_2(Q)}) 
    &\leq C_1C_2 \rbr{\frac{2a_1}{\epsilon}}^{v_1}
                \rbr{\frac{2a_2}{\epsilon}}^{v_2},
  \end{align*}
where $U = \|\cF\|_{\infty} \vee \|\cG\|_{\infty}$.  
\end{lemma}

\begin{lemma}
\label{lm:proc_deviation}
  Let 
  \[
  Z = \sup_{f\in\cF} \sum_{i\in[n]} W_i(f)
  \]
  where $\EE{W_i(f)} = 0$ and $|W_i(f)| \leq M$ for all   $i \in [n]$ and $f \in \cF$. Then 
  \[
  Z \leq \EE{Z} + 4\rbr{\sqrt{\rbr{4M\EE{Z}+\sigma_\cF^2}\log(1/\delta)}\bigvee M\log(1/\delta)}
  \]
  with probability $1-\delta$.
\end{lemma}
\begin{proof}
  The lemma is a simple consequence of Theorems 11.8 and 12.2 in \cite{Boucheron2013Concentration}. Assume that $M=1$. Then Theorem 12.2 in \cite{Boucheron2013Concentration} gives us
\[
\PP{Z \geq \EE{Z} + t} \leq \exp\rbr{-\frac{t^2}{2\rbr{2\rbr{\Sigma^2_\cF  + \sigma^2_\cF}+t}}}.
\]
Hence, with probability $1-\delta$, we have 
\[
Z \leq \EE{Z} + \sqrt{8\rbr{\Sigma^2_\cF  + \sigma^2_\cF}\log(1/\delta)} \bigvee 4\log(1/\delta).
\]
Furthermore, Theorem 11.8 in \cite{Boucheron2013Concentration} gives us that 
\[
\Sigma^2_\cF  + \sigma^2_\cF \leq 8\EE{Z} + 2 \sigma^2_\cF.
\]
Combining with the display above, we get 
\[
Z \leq \EE{Z} + 4\rbr{\sqrt{\rbr{4\EE{Z} + \sigma^2_\cF}\log(1/\delta)} \bigvee \log(1/\delta)}
\]
with probability $1-\delta$. We can rescale the equation above by $M$ to conclude the proof of the lemma.
\end{proof}

\section{Numerical studies} \label{app:sim}

\paragraph{Detailed data settings} In our numerical study, we set the parameters $$(a_0,a_1,b_0,b_1,c_0,c_1,d_0,d_1,\rho,\sigma_e) = (1,0.1,1,0.1,1,0.5,1,0.2,0.2,1)$$ to represent a general setting. We study the cases where $\gamma = 0$ and $\gamma = 1$,
and choose $c_y$ and $c_x$ to form different combinations of $(R_y^2,R_x^2)$.

\paragraph{Additional simulation results}
In Table \ref{table2}, we present the performance of methods with different combinations of $(R_x^2,R_y^2)$.
\begin{table}[h] 
\begin{center}
\begin{tabular}{c|cccccc}
\hline
$\epsilon$ distribution &$(R_x^2,R_y^2)$&Method & Bias & SD & ESE & CR\\
\hline
\multirow{10}{*}{Normal}&\multirow{5}{*}{$(0.3,0.3)$}&One Step & -0.007 & 0.081 & 0.080 & 0.94\\
&&Decorrelated score& 0.003 & 0.077 & 0.092 & 0.98\\
&&Reparameterization & 0.011 & 0.081 & 0.080 & 0.94\\
&&Naive & 0.017 & 0.089 & 0.090 & 0.96\\
&&Oracle & -0.012 & 0.075 & 0.091 & 0.97\\
\cline{2-7}
&\multirow{5}{*}{$(0.3,0.7)$}&One Step & 0.001 & 0.111 & 0.095 & 0.93\\
& &Decorrelated score& 0.05 & 0.090 & 0.110& 0.97\\
& &Reparameterization & 0.05 & 0.088 & 0.095 & 0.93\\
& &Naive & 0.260& 0.215 & 0.095 & 0.47\\
& &Oracle & -0.012 & 0.076 & 0.091 & 0.97\\
\cline{2-7}
&\multirow{5}{*}{$(0.7,0.7)$}&One Step & -0.029 & 0.179 & 0.188 & 0.96\\
& &Decorrelated score& 0.155 & 0.202 & 0.249& 0.92\\
& &Reparameterization & 0.068 & 0.162 & 0.188 & 0.95\\
& &Naive & 0.388 & 0.248 & 0.146 & 0.38\\
& &Oracle & -0.013 & 0.183 & 0.197 & 0.95\\
\hline
\multirow{10}{*}{$t(3)$}&\multirow{5}{*}{$(0.3,0.3)$}&One Step & 0.008 & 0.132 & 0.139 & 0.96\\
& &Decorrelated score& 0.010 & 0.130 & 0.145& 0.96\\
& &Reparameterization & 0.008 & 0.130 & 0.139 & 0.94\\
& &Naive & 0.008 & 0.130 & 0.148 & 0.98\\
& &Oracle & -0.010 & 0.098 & 0.119 & 0.97\\
\cline{2-7}
&\multirow{5}{*}{$(0.3,0.7)$}&One Step & -0.007 & 0.118 & 0.126 & 0.96\\
& &Decorrelated score& 0.016 & 0.112 & 0.145& 0.97\\
& &Reparameterization & 0.024 & 0.102 & 0.126 & 0.95\\
& &Naive & 0.066 & 0.163 & 0.125 & 0.90\\
& &Oracle & -0.010 & 0.100 & 0.119 & 0.97\\
\cline{2-7}
&\multirow{5}{*}{$(0.7,0.7)$}&One Step & -0.027 & 0.382 & 0.441 & 0.95\\
& &Decorrelated score& 0.032 & 0.328 & 0.454& 0.95\\
& &Reparameterization & 0.031 & 0.302 & 0.441 & 0.97\\
& &Naive & 0.045& 0.314 & 0.318 & 0.93\\
& &Oracle & 0.001 & 0.211 & 0.268 & 0.98\\
\hline
\end{tabular}
\end{center}
\caption{Simulation results with data settings varying $(R_x^2,R_y^2)$ with heterogeneous $\epsilon$ (i.e. $\gamma = 1$).}
\label{table2}
\end{table}

\section{Remarks on Assumption \ref{assumption: lam}}\label{app:c}
Assumption \ref{assumption: lam} holds when $X_A$ follows a 
multivariate approximately sparse linear model, where we 
require the coefficients to be approximately linear, sparse, and smooth.
Specifically, we assume that there exists a smooth and sparse $\zeta^{\star}_1(u),\cdots, \zeta^{\star}_k(u)$,
that is,
  \begin{itemize}
      \item  $u \mapsto \zeta_j^{\star}(u)$ is differentiable for $j=1,\cdots,k$ and 
          \[
    \norm{\zeta^{\star}_j(u') - \zeta^{\star}_j(u)}_2 \leq B_\zeta\norm{u'-u};
    \]
    \[
    \norm{\zeta^{\star}_j(u') - \zeta^{\star}_j(u) - (u' - u)\cdot \nabla_u \zeta^{\star}_j(u) }_2 \leq B_\zeta(u'-u)^2;
    \]
    \item the supports of $\beta^{\star}(u, \tau)$ and $\partial_u \beta^{\star}(\tau, u)$ are sparse, i.e. \[
S_j := \cbr{j \in [p] \mid \zeta_j^{\star}(u) \neq 0}
~\text{ and }~
S'_j := \cbr{j \in [p] \mid \zeta^{\star}_j(u) \neq 0 \text{ or }  \partial_u \zeta^{\star}_j(u) \neq 0},
\]
$s_j := |S| \ll n$ and $ |S'|\leq s_{1j} := c_{1j} s_j$
    for some constants $c_{1j}$;
  \end{itemize}
such that with $\zeta^{\star}(u)=(\zeta_1^{\star}(u), \cdots, \zeta_k^{\star}(u))^\top$ the residual $r_i=x_{i, A}-\zeta^{\star}(u_i) x_{i, A^c}$ is approximately orthogonal to $x_{i, A^c}$ weighted by $w_if_i(\tilde{q}_i)$. Specifically, 
   \begin{multline*}
       \max\big\{\norm{\sum_{i} w_if_i(\tilde{q}_i)r_ix_{i, A^c}^\top}_{\infty},
       \\
       \norm{\sum_{i} w_if_i(\tilde{q}_i)r_i\frac{u_i-u}{h}x_{i, A^c}^\top}_{\infty},\norm{\sum_{i} w_if_i(\tilde{q}_i)r_i\frac{(u_i-u)^2}{h^2}x_{i, A^c}^\top}_{\infty}\big\} \\
       = \epsilon_r^2= O\rbr{\frac{\log(np)}{nhh_f}}.
   \end{multline*} 
Based on our model assumption, 
for a fixed $u$, 
let 
\[
r^\star=r^\star(u)=\left(\begin{array}{cc}\zeta^\star(u)&h\nabla_u \zeta^\star(u)\\0 & \zeta^\star(u)\end{array}\right).
\]
We have the following sparse linear regression model
\begin{equation*}
w_if_i(\tilde{q}_i)\left(\begin{array}{c}x_{i, A}\\ \frac{u_i-u}{h}x_{i, A}\end{array}\right)= w_if_i(\tilde{q}_i)r^\star\left(\begin{array}{c}x_{i, A^c}\\ \frac{u_i-u}{h}x_{i, A^c}\end{array}\right)+w_if_i(\tilde{q}_i)\tilde{r}_i
\end{equation*}
where 
\[
\tilde{r}_i=\left(\begin{array}{c}r_i+(\zeta^{\star}(u_i) - \zeta^{\star}(u) - (u_i- u)\cdot \nabla_u \zeta^{\star}(u))x_{i, A^c}\\\frac{u-u_i}{h}(r_i+\rbr{\zeta^{\star}(u_i) - \zeta^{\star}(u)}x_{i, A^c})\end{array}\right).
\]
If $\Sigma_{11}=\sum_{i} w_if_i(\tilde{q}_i)\tilde{r}_i\tilde{r}_i^\top$ is invertible,
then we can have a specific $V^{\star}$ in the form of $\Sigma_{11}^{-1}(I_{2k},-r^\star)$ 
satisfy Assumption 5. The sparsity of
$\zeta^{\star}(u)$ guarantees the sparsity of $V^{\star}$, 
so we just need to show 
$\norm{H^\star V^\star-E_a}_{\infty,F}\leq \lambda^{\star}$, 
where the norm $\norm{\cdot}_{\infty,F}$ for $V \in \R^{2k\times 2p}$ is defined as 
$\norm{V}_{\infty, F} = \sup_{i \in [k],j \in [p]} \norm{V_{(i,i+k),(j,j+p)}}_F$.
If we show that 
\begin{equation*}
\norm{\Sigma_{11}^{-1}\sum_{i}w_if_i(\tilde{q}_i)\tilde{r}_i\Gamma_{i,A^c}^\top}_{\infty,F}\leq \lambda^{\star},
\end{equation*}
then
\begin{align*}
&\norm{\Sigma_{11}^{-1}\sum_{i}w_if_i(\tilde{q}_i)\tilde{r}_{i}\Gamma_{i,A}^\top-I_{2k}}_{\infty,F} \\
&\qquad=\norm{\Sigma_{11}^{-1}(\sum_{i}w_if_i(\tilde{q}_i)\tilde{r}_{i}\Gamma_{i,A}^\top-\sum_{i}w_if_i(\tilde{q}_i)\tilde{r}_{i}\tilde{r}_{i}^\top)}_{\infty,F}\\
&\qquad=\norm{\Sigma_{11}^{-1}(\sum_{i}w_if_i(\tilde{q}_i)\tilde{r}_{i}\Gamma_{i,A^c}^\top)}_{\infty,F}\leq \lambda^{\star}
\end{align*}
Since $\Sigma_{11}$ is invertible, it is sufficient
to bound $\norm{\sum_i w_if_i(\tilde{q}_i)\tilde{r}_i\Gamma_{i,A^c}^\top}_{\infty, F}$.
Given that
\begin{align*}
&\norm{\sum_i w_if_i(\tilde{q}_i)\tilde{r}_i\Gamma_{i,A^c}^\top}_{\infty, F}\\
&\leq
\sup_{l \in [k],j \in [p-k]}\\
&\norm{\sum_i w_if_i(\tilde{q}_i)\left(\begin{array}{c}r_{il}+x_{i, A^c}(\zeta^{\star}_l(u_i) - \zeta^{\star}_l(u) - (u_i- u)\cdot \nabla_u \zeta^{\star}_l(u))\\\frac{u-u_i}{h}(r_{il}+x_{i, A^c}\rbr{\zeta^{\star}_l(u_i) - \zeta^{\star}_l(u)})\end{array}\right)(x_{i,A^c_j},\frac{u_i-u}{h}x_{i,A^c_j})}_{F}\\
&= \sup_{l \in [k],j \in [p-k]} \{ [\sum_i w_if_i(\tilde{q}_i)(r_{il}+x_{i, A^c}(\zeta^{\star}_l(u_i) - \zeta^{\star}_l(u) - (u_i- u)\cdot \nabla_u \zeta^{\star}_l(u)) \cdot x_{i,A^c_j} ]^2\\ 
&+[\sum_i w_if_i(\tilde{q}_i)(r_{il}+x_{i, A^c}(\zeta^{\star}_l(u_i) - \zeta^{\star}_l(u) - (u_i- u)\cdot \nabla_u \zeta^{\star}_l(u)) \cdot \frac{u_i-u}{h}x_{i,A^c_j} ]^2\\
&+ [\sum_i w_if_i(\tilde{q}_i)\frac{u-u_i}{h}(r_{il}+x_{i, A^c}\rbr{\zeta^{\star}_l(u_i) - \zeta^{\star}_l(u)})\cdot x_{i,A^c_j}]^2\\
&+ [\sum_i w_if_i(\tilde{q}_i)\frac{u-u_i}{h}(r_{il}+x_{i, A^c}\rbr{\zeta^{\star}_l(u_i) - \zeta^{\star}_l(u)})\cdot \frac{u_i-u}{h}x_{i,A^c_j}]^2\}^{-1/2}\\
&=O(\epsilon_r+h)\leq \lambda^\star,
\end{align*}
we have that Assumption \ref{assumption: lam} holds.
%%% Local Variables:
%%% TeX-master: "QR_score.tex"
%%% End:

\end{document}